\documentclass[11pt]{article}

\usepackage[round]{natbib}
\usepackage[T1]{fontenc} % only needed for the underscores in the bib entry for Kimura (1958) and Glémin et al (2019)

\usepackage[pass]{geometry}
\usepackage{amsfonts,amsmath,amssymb}
\usepackage{latexsym,amsthm, microtype}
\usepackage{graphicx}
\usepackage{tikz}
\usepackage{xcolor}
\usepackage{mathtools}
\usepackage{placeins}
\usepackage[all]{xy}
\usepackage{bbm}
\usepackage{lineno}

\newtheorem{theorem}{Theorem}[section]
\newtheorem{lemma}[theorem]{Lemma}
\newtheorem{corollary}[theorem]{Corollary}
\newtheorem{proposition}[theorem]{Proposition}
\newtheorem{definition}[theorem]{Definition}
\newtheorem{conjecture}[theorem]{Conjecture}
\theoremstyle{plain}
\newtheorem{remark}[theorem]{Remark}

\newcommand\Z{\ensuremath{\mathbb{Z}}}
\newcommand\R{\ensuremath{\mathbb{R}}}
\newcommand\PP{\ensuremath{\mathbb{P}}}
\newcommand{\Psame}{\mathbb{P}_{\rm{same}}}
\newcommand{\Pdiff}{\mathbb{P}_{{\rm diff}}}
\newcommand{\Esame}{\mathbb{E}_{\rm{same}}}
\newcommand{\Ediff}{\mathbb{E}_{\rm{diff}}}

\long\def\metanote#1#2{{\color{#1}\
\ifmmode\hbox\fi{\sffamily\mdseries\upshape [#2]}\ }}

\renewcommand{\epsilon}{\varepsilon}
\usepackage[para]{footmisc}
\renewcommand\footnotemark{}
\usepackage{authblk}

\definecolor{gr}{rgb}{0.03, 0.7, 0.29}
\definecolor{Blue}{rgb}{0,0,1}

\title{Conditional gene genealogies given the population pedigree \\ for a diploid Moran model with selfing }
\author[1]{Maximillian Newman}
\author[2]{John Wakeley}
\author[1,2,*]{Wai-Tong (Louis) Fan }

\affil[1]{Department of Mathematics, Indiana University, Bloomington, IN}
\affil[2]{Department of Organismic and Evolutionary Biology, Harvard University, Cambridge, MA}
\affil[*]{Corresponding author: L. Fan (wlfan@fas.harvard.edu)}

\date{\today}

\thanks{\small Email: 
	\texttt{newmama@iu.edu} (Maximillian Newman), 
	\texttt{wakeley@fas.harvard.edu} (John Wakeley), 
	\texttt{wlfan@fas.harvard.edu} (Louis Fan) }

\begin{document}

\maketitle

\begin{abstract}
We introduce a stochastic model of a population with overlapping generations and arbitrary levels of self-fertilization versus outcrossing.  We study how the global graph of reproductive relationships, or population pedigree, influences the genealogical relationships of a sample of two gene copies at a genetic locus.  Specifically, we consider a diploid Moran model with constant population size $N$ over time, in which a proportion of offspring are produced by selfing.  We show that the conditional distribution of the pairwise coalescence time at a single locus given the random pedigree converges to a limit law as $N$ tends to infinity.  The distribution of coalescence times obtained in this way predicts variation among unlinked loci in a sample of individuals. Traditional coalescent analyses implicitly average over pedigrees and generally make different predictions.  We describe three different behaviors in the limit depending on the relative strengths, from large to small, of selfing versus outcrossing: partial selfing, limited outcrossing, and negligible outcrossing.  In the case of partial selfing, coalescence times are related to the Kingman coalescent, similar to what is found in traditional analyses.  In the case of limited outcrossing, the retained pedigree information forms a random graph, with coalescence times given by the meeting times of random walks on this graph.  In the case of negligible outcrossing, which represents complete or nearly complete selfing, coalescence times are determined entirely by the fixed times to common ancestry of diploid individuals in the pedigree.     
\end{abstract}

\section{Introduction}\label{sec:intro}

The organismal genealogy of all individuals in a population for all time, also known as the population pedigree, has not traditionally been of concern in population genetics.  But simulation studies, beginning with \citet{ball1990} and including \citet{kuoavise2008}, \citet{WakeleyEtAl2016} and \citet{WiltonEtAl2017}, have suggested that the population pedigree can be important in structuring genetic variation.  Of course, some of the earliest work in population genetics concerned how pedigrees affect patterns of variation, specifically identity by descent  \citep{Wright1921a,Wright1921b,Wright1921c,Wright1921d,Wright1921e,Wright1922}.  But the pedigrees in these early works were either ones laid down by regular systems of mating---see \citet{Lachance2009} for a recent treatment---or pedigrees of arbitrary structure extending back a limited number of generations.  These latter, relatively small pedigrees were assumed to be embedded in an infinite population with global properties, such as Hardy-Weinberg equilibrium and no identity by descent, and within which there was no explicit accounting of individuals or their relationships.

Finite population models including coalescent models are models of identity by descent.  But population pedigrees did not figure in their development.  The approach instead was to average over pedigrees by averaging over the outcomes of reproduction in each generation.  See \citet{DFBW24} for more about this history and a discussion of the roles of pedigrees in forward-time versus backward-time models.  Averaging over pedigrees is usually done implicitly, as in the statement that two ancestral lines coalesce with probability $1/(2N)$ each generation under the diploid monoecious Wright-Fisher model.  It is done explicitly, for example, in equation (4.1) in \citet{kingman1982} for the coalescence probability under Cannings' haploid exchangeable model \citep{Cannings1974}.  Predictions for diploid organisms based on this traditional approach are marginal to the population pedigree. 

Coalescent theory describes the ancestry of gene copies or alleles in a sample of individuals from a population.  Coalescent models are often applied by comparing predicted distributions of genetic variation at a single locus to data from multiple unlinked loci in a sample of individuals.  This is done for example in analyses of the site-frequency spectrum \citep{Tajima1989,BravermanEtAl1995,Fu1995,Achaz2009,GutenkunstEtAl2009,GaoAndKeinan2016} or when likelihoods are multiplied across unlinked loci \citep{Watterson1985,Wakeley1999,Nielsen2000,AdamsAndHudson2004,WangAndHey2010,LohseEtAl2011,CWH2017,CWH2021}. But a sample of individuals has a single shared pedigree, fixed in time by past outcomes of reproduction.  Further, in large populations with substantial outcrossing, the pedigree of any sample of individuals quickly spreads out to encompass the whole population \citep{Chang1999,DerridaEtAl1999,Lachance2009,BartonAndEtheridge2011}.  So it is important to know whether any specific features of the population pedigree might affect our predictions about genetic variation.

Coalescent models conditional on the pedigree give the correct sampling structure for multiple loci in a sample of individuals.  In particular, unlinked loci represent conditionally independent realizations of the process of genetic transmission within the same fixed pedigree.  Methods such as multiplying likelihoods across unlinked loci are obviously justified when conditioning on the pedigree.  Modeling unlinked loci as independent draws from the distributions of traditionally formulated coalescent models implicitly assumes that each locus has its own population pedigree.  

So far there have been only two analytical treatments of coalescence conditional on the population pedigree.  \citet{TyukinThesis2015} proved that the standard neutral coalescent model \citep{kingman1982,Hudson1983a,tajima1983} can still be used when coalescence is conditioned on the pedigree, for populations which would be described by that model under the traditional approach.  
%This can be attributed to the characteristically short mixing time of the pedigree in such populations \citep{Chang1999,DerridaEtAl1999,BartonAndEtheridge2011}.  
But it is not true in general that extensions of the standard neutral coalescent model are similarly applicable to coalescent processes conditional on pedigrees generated under the corresponding population models.  \citet{DFBW24} proved this for populations in which very large families occur with some frequency, and went on to describe the resulting conditional coalescent models, in which predictions about genetic variation depend on the fixed times and sizes of large families.

In this work, we consider how coalescence conditional on the population pedigree is affected by autogamy or self-fertilization (hereafter just `selfing').  Modes of reproduction are of course fundamental aspects of evolution \citep{CharlesworthAndWright2001,GleminEtAl2019}.  Selfing is when two gametes produced by meiosis by the same individual unite to form a zygote.  This type of inbreeding has been especially important in the evolution of plants \citep{Jain1976,UyenoyamaEtAl1993,Charlesworth2006,OlsenEtAl2021}.  The common alternative to selfing in both plants and animals is outcrossing, in which zygotes are formed by the union of two gametes from different individuals.  We consider just these two possibilities, and characterize populations by the probability individuals are produced by selfing.  We describe three distinct kinds of conditional coalescent models which hold for different magnitudes of the selfing probability.

These models exist in the limit as the population size $N$ tends to infinity in the same way that the Wright-Fisher diffusion and the corresponding standard neutral coalescent model exist in this limit \citep{Ewens2004}.  We use $\alpha_N$ to denote the selfing probability, with explicit dependence on $N$ so we can control the importance of selfing in the model.  Our results for coalescence in large populations conditional on the pedigree depend strongly on the relative magnitude of $\alpha_N$, specifically how $\alpha_N$ changes as $N\to\infty$.  Most previous analyses, in addition to not conditioning on the pedigree, have assumed a fixed selfing probability.  This probability is often denoted $s$, though various notations have been adopted; e.g.\ $l$ in \citet{Haldane1924}, $h$ in \citet{Wright1951}, $\beta$ in \citet{Pollak1987} and $\alpha$ in \citet{Karlin1968a,Karlin1968b}.  As \citet{Uyenoyama2024} recently emphasized, some authors use $s$ for the entire selfing probability while others use it to measure selfing in excess of a background probability of $1/N$ due to random mating.  Our model includes these as special cases, namely if $\alpha_N = s$ or $\alpha_N = s + (1-s)/N$, respectively. Both of these fixed-$s$ models make the same predictions in the limit and both have been called `partial self-fertilization' or `partial selfing'.

We use $\alpha\in[0,1]$ to denote the limit of $\alpha_N$ as $N\to\infty$.  Conditioning coalescence on the pedigree reveals three markedly different behaviors near the boundary $\alpha=1$, depending on how quickly $\alpha_N\to 1$.  The critical case, which we call `limited outcrossing', is when $1-\alpha_N$ is of order $1/N$.  Then the conditional coalescent process is characterized by a system of coalescing random walks on a random ancestral graph like the ones previously described for recombination \citep{Griffiths1991,GriffithsAndMarjoram1997} and selection \citep{kroneneuhauser1997,NeuhauserAndKrone1997}.  If outcrossing is even less frequent, for example $1-\alpha_N$ of order $1/N^2$, the limiting conditional coalescent process is fully determined by the random pedigree of the population, which becomes just a Kingman coalescent tree connecting diploid individuals.  We call this `negligible outcrossing'. If outcrossing is more frequent than in the critical case, for example $1-\alpha_N$ of order $1/\sqrt{N}$, the limiting conditional coalescent model resembles the pedigree-averaged or unconditional coalescent model of \citet{NordborgAndDonnelly1997} and \citet{Mohle1998a}.  Here we adopt the previous term `partial selfing'.  The latter model holds for all $\alpha\in[0,1]$ as long as the rate of approach to $\alpha=1$ is not too fast, whereas limited outcrossing and negligible outcrossing hold just for $\alpha_N\to 1$.

In the next two sections, we review previous works which generally use $s$ for the fixed selfing probability, take $s\in[0,1]$, and do not condition on the population pedigree.  We will reserve our notation $\alpha_N$ for Section~\ref{sec:pedigrees} and later.  In Section~\ref{sec:pedigrees}, we explain how previous population-genetic models of selfing in fact average over pedigrees.  In Sections~\ref{sec:broad} and~\ref{sec:specific} we use $s$ for continuity with what one finds in much of the evolutionary and population-genetic literature.

\subsection{Prevalence and evolutionary-genetic consequences of selfing}\label{sec:broad}

Although the ability to self-fertilize is rare among animals \citep{JarneAndAuld2006,AviseAndMank2009,EscobarEtAl2011} it is found in other taxa including eukaryotic microbes and marine invertebrates \citep{SassonAndRyan2017,YadavEtAl2023} and is common in plants \citep{AbbottAndGomes1989,HartfieldEtAl2017,TeterinaEtAl2023}.  About $10$-$15$\% of plant species are predominantly selfing \citep{WrightEtAl2013}, as are two of the best-studied genetic model species, \textit{Arabidopsis thaliana} and \textit{Caenorhabditis elegans} \citep{AbbottAndGomes1989,SellingerEtAl2020,BarriereAndFelix2005}.    
The empirical distribution of selfing probabilities among plant species is markedly bimodal, with relatively fewer species having $s\in(0.2,0.8)$ \citep{SchemskeAndLande1985,VoglerAndKalisz2001,GoodwillieEtAl2005}.  Whether selfing is favored over outcrossing or vice versa in a given species depends on a number of ecological and evolutionary factors \citep{Fisher1941,Moran1962,LandeAndSchemske1985,Charlesworth2006,Kamran‐DisfaniAndAgrawal2014}.

Selfing is a form of inbreeding which has three important consequences for the distribution of genetic variation.  First, selfing increases single-locus homozygosity via increased identity by descent \citep{Haldane1924,Rousset2002}.  In terms of Wright's individual inbreeding coefficient $F$ \citep{Wright1931,Wright1951}, what \citet{Haldane1924} showed was that the reduction in the proportion of heterozygous individuals in the population at equilibrium is given by $F=s/(2-s)$.  A simple coalescent interpretation is that $F$ is the probability that the two alleles or gene copies at a locus in an individual coalesce in the recent ancestry of that individual, rather than being separated into different individuals by an outcrossing event \citep{NordborgAndDonnelly1997}.

Second, selfing reduces the effective rate of recombination \citep{BennettAndBinet1956,WeirAndCockerham1973,GoldingAndStrobeck1980,VitalisAndCouvet2001}.  \citet{Nordborg2000} provided an interpretation in terms of coalescence.  Backward in time, a recombination event places ancestral genetic material onto two chromosomes.  These are necessarily in the same parental individual which might itself have been produced by selfing.  Then the same coefficient $F$ is the probability that the recombination event gets healed or repaired by a coalescent event within that parental individual's recent ancestry.  So only a fraction $1-F=2(1-s)/(2-s)$ of recombination events is observable.  

Third, selfing generates non-random associations of diploid genotypes at different loci, even if the loci are unlinked \citep{Haldane1949,BennettAndBinet1956,Kimura1958,Narain1966,CockerhamAndWeir1968,CockerhamAndWeir1977,CockerhamAndWeir1983,CockerhamAndWeir1984,Smith2023,Uyenoyama2024}.  \citet{WeirAndCockerham1973} partitioned this effect into a sum of two terms, one due to linkage, which is equal to zero for unlinked loci, and one due to selfing, which is called the `identity disequilibrium' and is equal to the theoretical variance of $F$ among individuals in the population: 
\begin{linenomath*}
\begin{equation}
\sigma^2_F = \frac{4s(1-s)}{(2-s)^2(4-s)} . \label{eq:sigma2F}
\end{equation}
\end{linenomath*}
For reference, the quantity \eqref{eq:sigma2F} is denoted  $\sigma^2_{F_1}$ in \citet[p.\ 260]{WeirAndCockerham1973}.

A coalescent interpretation of identity disequilibrium is as follows.  Let $Y_i$ be an indicator variable, equal to $1$ if the two gene copies at locus $i$ in a particular individual coalesce in its recent ancestry, and equal to $0$ otherwise.  The recent ancestry of a randomly sampled individual is given by the number of generations of selfing in its ancestral line before the first outcrossing event.  \citet{WeirAndCockerham1973} note that this idea originates with \citet{CockerhamAndRawlings1967}.  For an individual with $k$ generations of selfing in its immediate ancestry, the probability the two gene copies at a locus remain distinct (i.e.\ do not coalesce) in any of these generations is $2^{-k}$.  If $U$ denotes the random number of selfing generations for an individual,  
\begin{linenomath*}
\begin{equation}
\PP(U=k) = s^k (1-s) \qquad k = 0,1,2,\ldots \label{eq:PUk} 
\end{equation}
\end{linenomath*} 
and the coefficient $F$ can be seen as the average    
\begin{linenomath*}
\begin{equation}
F \coloneq \mathbb{E}\left[Y_i\right] = \sum_{k=0}^{\infty} \left(1-2^{-k}\right) s^k (1-s) = \frac{s}{2-s} .  \label{eq:Fdef}
\end{equation}
\end{linenomath*}
The covariance of $Y_i$ and $Y_j$ for two unlinked loci, which are conditionally independent given $k$, is  
\begin{linenomath*}
\begin{equation}\label{eq:covxixj}
\mathrm{Cov}{\left[Y_i,Y_j\right]} = \sum_{k=0}^{\infty} \left(1-2^{-k}\right)^2 s^k (1-s) - \left(\frac{s}{2-s}\right)^2, 
\end{equation}
\end{linenomath*}
and this is identical to \eqref{eq:sigma2F} because $\sigma^2_F \coloneq \mathrm{Var}{\left(Y_i\right)}$ is also given by \eqref{eq:covxixj}.  Interpreting identity disequilibria as covariances is sensible because they occur between loci.

In anticipation of later results conditional on the population pedigree, note that \eqref{eq:PUk} can be interpreted in two different ways.  A given pedigree has fixed values of $k$ for every individual in the population.  Following \citet{Haldane1949}, \citet{WeirAndCockerham1973} and others,  \eqref{eq:sigma2F}, \eqref{eq:Fdef} and \eqref{eq:covxixj} can be interpreted as averages over individuals in an effectively infinite population.  Alternatively, \eqref{eq:PUk} can be seen as the prior distribution for a single individual in a given population.  In this case, \eqref{eq:sigma2F}, \eqref{eq:Fdef} and \eqref{eq:covxixj} are seen as averages over all possible population pedigrees.  One reason this distinction is important is that while \eqref{eq:sigma2F} or \eqref{eq:covxixj} may be greater than zero, the covariance of $Y_i$ and $Y_j$ for two unlinked loci in a single individual conditional on the population pedigree must be equal to zero.

Owing to these three effects---increased homozygosity, reduced recombination, and identity disequilibrium---partially selfing populations follow different evolutionary paths than randomly mating populations \citep{CharlesworthEtAl1990,CharlesworthEtAl1993,UyenoyamaAndWaller1991a,UyenoyamaAndWaller1991b,UyenoyamaAndWaller1991c,LandeEtAl1994,NordborgEtAl1996,RozeAndLenormand2005,Roze2016,AbuAwadAndRoze2020}.  The evolution of selfing also involves a variety of ecological phenomena and associated morphologies \citep{Baker1955,Stebbins1957,Ornduff1969,Jain1976,Cutter2019}.  We note that some but not all evolutionary studies have explicitly modeled variation in the number of selfing generations in the recent ancestries of individuals \citep{CockerhamAndRawlings1967,Campbell1986,Kelly1999a,Kelly1999b,Kelly2007,LandeAndPorcher2015}, given here under neutrality by \eqref{eq:PUk}. 

\subsection{Standard neutral models with a fixed selfing probability but no pedigree}\label{sec:specific}

Our goal in this work is to understand how selfing affects the sampling structure of gene genealogies under neutrality when the process of coalescence is conditioned on the pedigree of the population.  For simplicity we use a diploid Moran model of reproduction \citep{Moran1958,Moran1962,linder2009,coron2022pedigree}, which we describe in Section~\ref{section: moran_model}.  Here we review some relevant previous theoretical treatments which do not include pedigrees and which assume Wright-Fisher reproduction \citep{Fisher1930,Wright1931}.  

Both the Wright-Fisher model and the Moran model assume a single well mixed population of constant in size $N$ without selection.  In the usual simplest application, i.e.\ to diploid, monoecious, randomly mating organisms, selfing occurs with probability $1/N$.  When $N$ can be assumed to be large, both converge to the Wright-Fisher diffusion with its corresponding coalescent process \citep{Ewens2004}, which exist in the limit $N\to\infty$ with time measured in units proportional to $N$ generations.  Arbitrary levels of selfing have been incorporated by assuming a fixed selfing probability of $s\in[0,1]$ in the limit $N\to\infty$.  Based on the considerations of Section~\ref{sec:broad}, it may be that selfing species are even less likely than outcrossing ones to meet the assumptions of these standard models of population genetics, but they provide a starting point for studying more complicated populations.

\citet{Pollak1987} took a forward-time approach to studying, among other things, the equilibrium probabilities of identity by descent of two gene copies in the same individual or in different individuals, when $s \gg N^{-1}$.  Let us call these probabilities $\Phi_1$ and $\Phi_2$, respectively.  For large $N$, \citet{Pollak1987} found that the expression for $\Phi_2$ was identical to that for a randomly mating population, if $N_e = N/(1+F) = (2-s) N / 2$ is used in place of $N$.  As $s$ increases from $0$ to $1$, $N_e$ decreases from $N$ to $N/2$, effectively taking the value of $\Phi_2$ from that for a diploid (randomly mating) population to that for a haploid population.  \citet{Pollak1987} identified this as a separation-of-time-scales result, stemming from the assumption $s \gg N^{-1}$, and suggested that the two-times-scales diffusion approximation of \citet{EthierAndNagylaki1980} could be applied to prove convergence to the Wright-Fisher diffusion with this value of $N_e$.  The separation of time scales can be seen in the relationship \citet{Pollak1987} found between $\Phi_1$ and $\Phi_2$, namely that $\Phi_1 = F + (1-F) \Phi_2$, which may best be understood in terms of subsequent descriptions of the following coalescent process. 

\citet{NordborgAndDonnelly1997} took a backward-time coalescent approach to this same model, which \citet{Nordborg1997} placed in a more general framework of separation of time scales.  \citet{NordborgAndDonnelly1997} described a coalescent process with two time scales: one involving just selfing which is fast and plays out over relatively small numbers of generations, and one involving coalescence which is slow in the usual sense of coalescent theory where time is measured in units proportional to $N$ generations.  The fast process occurs when one or more pairs of lineages are in the same individual(s).  They showed (though without a formal proof) that in the limit $N\to\infty$, if a sample of size $n$ contains $2m$ gene copies together as pairs in $m$ individuals and $n-2m$ gene copies each in different individuals, there is an instantaneous adjustment in which a random number $X \sim \textrm{binomial}(m,F)$ of the $m$ pairs coalesce.  The rest of the ancestry begins with the resulting $n-X$ lineages and is given by the standard neutral coalescent process \citep{kingman1982,Hudson1983a,tajima1983}, if time is measured in units of $2N_e = (2-s) N$ generations. A somewhat different approach, involving an algorithm with these same essential features, was taken by \citet{Fu1997}.

It was in this context of partial selfing that \citet{Mohle1998a} proved an important general result \citep[Lemma 1, Theorem 1]{Mohle1998a} for Markov processes with two time scales.  Owing to the duality relation between the coalescent and diffusion models \citep{Mohle1999} this result can be seen as the backward-time counterpart to the result of \citet{EthierAndNagylaki1980} cited by \citet{Pollak1987}.  For the selfing model with fixed $s$ as $N\to\infty$, this gave rigorous justification to the model of \citet{NordborgAndDonnelly1997}: whenever a pair of lineages is in the same individual, they coalesce with probability $F=s/(2-s)$ instantaneously, otherwise ending up in different individuals with probability $1-F$, while samples from different individuals obey the standard neutral coalescent process with time in units of $(2-s) N$ generations \citep{Mohle1998a}.

\subsection{Accounting for the pedigree of the population is important}\label{sec:pedigrees}

The coalescent models of partial selfing in  \citet{Fu1997}, \citet{NordborgAndDonnelly1997} and \citet{Mohle1998a} average over the process of reproduction and thus implicitly over the pedigree.  Clearly, one consequence of conditioning on the pedigree should be that each sampled individual has its own realized number of generations of selfing before the first outcrossing event in its ancestry.  This, as we will show, is a single draw from the distribution \eqref{eq:PUk}.  Previous approaches instead used the average $F$ in \eqref{eq:Fdef}.  The averaging procedure is illustrated by the matrix $A^m$ in \citet[p. 496]{Mohle1998a} which gives the $m$-generation transition probabilities of the fast process.  For example, 
\begin{linenomath}
\begin{equation*}
\left(A^m\right)_{32} = (1-s)\frac{1-(s/2)^m}{1-s/2} 
\end{equation*}
\end{linenomath}
is the probability that the ancestors of two gene copies currently in the same individual remain distinct and are in two different individuals in past generation $m$.   This can be expressed as a sum over pedigrees using the distribution \eqref{eq:PUk}, namely 
\begin{linenomath}
\begin{equation}
\left(A^m\right)_{32} = \sum_{k=0}^{m-1} \left(\frac{1}{2}\right)^k s^k (1-s) . \label{eq:mohleaverging}
\end{equation}
\end{linenomath}
Lemma 1 and Theorem 1 of \citet{Mohle1998a} use the limit of the fast process, $P \coloneq \lim_{m\to\infty} A^m$, in which \eqref{eq:mohleaverging} converges to $2(1-s)/(2-s)=1-F$, or $\mathbb{E}\left[1-Y_i\right]$ for the indicator random variable $Y_i$ defined in Section~\ref{sec:broad}. 

Partial selfing is just one of three possibilities when the coalescent process is conditioned on the pedigree.  Here we reintroduce our notation, in which $\alpha_N$ denotes the selfing probability in a population of size $N$, and $\alpha$ denotes its value in the limit $N\to\infty$.  When selfing is the primary mode of reproduction, an entirely new sort of model emerges in the limit, one which is completely invisible in the unconditional pedigree-averaged process.  This novel limit process occurs when the outcrossing probability $1-\alpha_N$ is of order $1/N$.  It is parameterized by $\lambda\coloneq\lim_{N\to\infty}N(1-\alpha_N)$, which we note is the same way mutation rates are assumed to scale in the Wright-Fisher diffusion and corresponding standard neutral coalescent model \citep{Ewens2004} and how migration rates are scaled in the structured coalescent model \citep{Takahata1988,Notohara1990,Herbots1997}.  In the limit we consider, $\lambda$ is the instantaneous rate of outcrossing events per ancestral lineage.  
%Similar to coalescent models generally, this new conditional coalescent model is meant as an approximation for large populations with small outcrossing probabilities.

Figure~\ref{fig:1} shows simulation results of pedigrees and pairwise coalescence times for the model of selfing used by \citet{Mohle1998a}, which corresponds to $\alpha_N=\alpha$ constant in our model.  Specifically, for each of three selfing probabilities, $50$ independent pedigrees were simulated and for each of these two different individuals were sampled at random.  Colored lines, one for each pedigree-plus-sample, depict the cumulative distribution function (CDF) of the coalescence time for two gene copies, one from each of the two sampled individuals.  On the scale used in Figure~\ref{fig:1}, the corresponding expectation for the pedigree-averaging model in all three cases would be the CDF of an exponential distribution with parameter $(2-\alpha)^{-1}$.  Many of the CDFs in Figure~\ref{fig:1}(a) have something close to this shape, but those in Figure~\ref{fig:1}(b) and Figure~\ref{fig:1}(c) look nothing like it.  Instead they are characterized by random jumps in probability at random times. For any one population, these times are fixed and determine the predicted distribution of coalescence times among unlinked loci.

\begin{figure}[ht]
\includegraphics[scale=0.88]{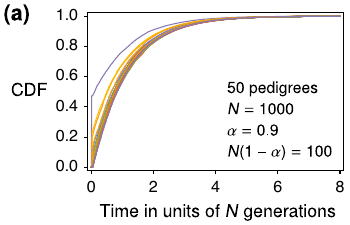}$\quad$\includegraphics[scale=0.88]{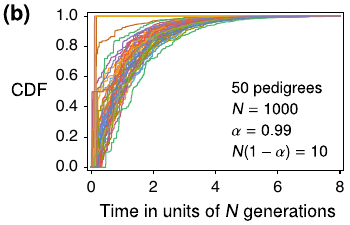}$\quad$\includegraphics[scale=0.88]{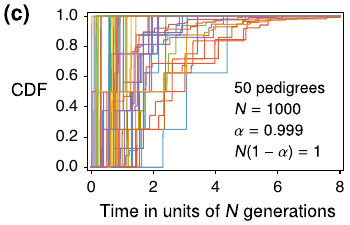}
\caption{\small Cumulative distribution functions of the coalescence times for a sample of size $n=2$ on each of $50$ population pedigrees simulated under the Wright-Fisher model with partial selfing and $N=1000$ individuals.  Panels (a), (b) and (c) show the results for three different values of the probability of selfing, respectively, for $\alpha\in\{0.9,0.99,0.999\}$.  Coalescence probabilities in each generation given the pedigree and the sampled individuals were calculated exactly, up to numerical precision} by the method in \citet{WakeleyEtAl2012}.\label{fig:1}
\end{figure}

The deviations from the predictions of the unconditional coalescent model displayed in Figure~\ref{fig:1} are not due to the relatively small population size used in these simulations ($N=1000$).  Instead they illustrate that what we will prove is true in the limit $N\to\infty$ with $N(1-\alpha_N)\to\lambda$ holds even for such small populations.  The value $\alpha=0.99$ in Figure~\ref{fig:1}(b) is similar to estimates of the selfing probabilities for the model species \textit{Arabidopsis thaliana} and \textit{Caenorhabditis elegans} \citep{AbbottAndGomes1989,SellingerEtAl2020,BarriereAndFelix2005}.  Of course, $N=1000$ would be a very small population size for either of these species, though it may be realistic for some local populations.  In any case, species or populations for which $N(1-\alpha)$ is not very large will have genetic ancestries very unlike what is expected under the unconditional or pedigree-averaging model.  Instead, they are captured by the conditional coalescent model with `limited outcrossing'  described in Section~\ref{S:condition}.

\subsection{Plan of the present work}\label{sec:plan}

Our analysis is based on a diploid Moran model of reproduction with constant population size $N$ and selfing probability $\alpha_N$, described in detail in Section~\ref{section: moran_model}.  We focus on the coalescence time for a sample of size $n=2$ at a single locus, which already reveals the non-trivial effects of the pedigree on gene genealogies. In the case of partial selfing, we prove convergence in distribution of the coalescence time via an analysis of two independently assorting or unlinked loci on the population pedigree, cf.\ \citet{Birkneretal2012} and \citet{DFBW24}.  In the case of limited outcrossing, we take a different approach to establish the conditional limit.  Here we focus on the set of diploid individual ancestors of the sample and prove that this converges weakly to an ancestral graph.

The organization of the paper is as follows.  In Section~\ref{section: moran_model}, after detailing  our Moran model, we make precise the pedigree and the gene genealogy.  Our results for the unconditional and the conditional distributions of pairwise coalescence times are presented in Section~\ref{S:Uncondition} and Section~\ref{S:condition}, respectively. We briefly discuss the case $n > 2$ and offer a conjecture in Section~\ref{S:condition_n}. In Section \ref{sec:limitedprops}, we establish some fundamental properties of the scaling limit under limited outcrossing and of the associated ancestral graph. We present the proofs of our results in the Appendix, Section~\ref{sec:appendix}.

%In the main result section, Section~\ref{S:condition}, we focus on the case $n=2$.  In Section~\ref{S:condition_n} we briefly discuss the case $n > 2$ and offer a conjecture. In Section \ref{sec:limitedprops}, we establish some fundamental properties of the scaling limit under limited outcrossing. This scaling limit is a coalescing random walk on a random  ancestral graph. Proofs for our main theorems, Theorem~\ref{T:MAIN_conditional} and Theorem~\ref{T:MAIN_conditional_same}, will be given in Section \ref{section: proofs}. 

%In subsections \ref{section: proofs_critical}-\ref{section: proofs_subcritical}, we consider convergences for each of the 3 regimes. 

%In Sections 5, 6 and 7 we give proofs for the cases of limited outcrossing, negligible outcrossing, and partial selfing, respectively.

\section{A discrete-time diploid Moran model with selfing} \label{section: moran_model}
            
%We consider a diploid Moran model with partial selfing and constant population size $N$. This model is a generalization of the model  in \citet{coron2022pedigree}, in which we introduce an extra parameter $\alpha_N$ which is the probability that a given individual offspring is produced by selfing. Our model is analogous to the one in \citet{linder2009}, except we allow for the selfing probability $\alpha_N$ to scale with the population size $N$ and to take any value in $[0,1]$.  In particular, $\alpha_N$ may converge to $1$.
        
%\subsection{Description of our population model} \label{construction:diploid_moran}

We consider a diploid, monoecious, panmictic (well mixed and randomly mating) population of constant size $N$. Individuals are labeled $I = \{1,...,N\}$. The population has overlapping generations and changes at discrete time-steps.   %and a sequence of independent times $\tau_i \sim \text{Exp}(1)$ for $i\in \mathbb{Z}_{>0}$. Let $t_k = \sum_{i \leq k} \tau_i$. 
At each time-step, a randomly chosen individual dies and is replaced by an offspring produced in one of two ways. With probability $1-\alpha_N$ the offspring has two distinct parents; with probability $\alpha_N$ the offspring has exactly one parent. The parent(s) are chosen uniformly at random from the previous time-step. In the case of two parents, these are chosen without replacement.  The reproduction events at different time-steps are independent and identically distributed.
            
Explicitly, for each non-negative integer $k\in\mathbb{Z}_+=\{0, 1, 2, 3, ...\}$, the reproduction event
between consecutive time-steps $k$ and $k+1$ in the past is as follows:
\begin{enumerate}
    \item (distinct parents) With probability $1-\alpha_N$, a single offspring is produced by outcrossing.  A triplet containing the two parents and one offspring $(\pi_{k,1},\, \pi_{k,2},\, \epsilon_k) \in I ^ 3$ is chosen uniformly at random with $\pi_{k,1}\neq \pi_{k,2}$ (i.e.\ the parents are distinct).
    The offspring $\epsilon_k$ is an individual in time-step $k$, while the offspring's parents are the  individuals in time-step $k + 1$ with labels $\pi_{k,1}$  and $\pi_{k,2}$. 
        
    The offspring has two genes copies, one inherited from each parent.  Genetically, the offspring is produced according to Mendel's laws, which means each of the two gene copies in a parent is equally likely to be the one transmitted to the offspring. An example outcrossing event in a population of size $N=7$ is depicted below. 
    \smallskip
    \begin{center}
         \def\pgone{{1,2,4,4,5,6,7}} % parents for gene 1 in each ind.
        \def\pggone{{0,0,0,0,0,0,0}} % index of parental gene for gene 1 in each ind.
        \def\pgtwo{{1,2,2,4,5,6,7}} % parents for gene 2 in each ind.
         \def\pggtwo{{1,1,0,1,1,1,1}} % index of parental gene for gene 2 in each ind.
                  
         \begin{tikzpicture}
        \foreach \y in {0,1}
         \foreach \x in {1,2,3,4,5,6,7} {
            \filldraw[black] (\x -0.1,\y) circle[radius=0.05]
            (\x +0.1,\y) circle[radius=0.05];
            \draw (\x,\y) ellipse[x radius=0.35, y radius=0.25];
        };
                
         \foreach \i in {0,1,2,3,4,5,6} {
            \pgfmathsetmacro{\j}{\pgone[\i]+0.1*(2*\pggone[\i]-1)}
            \draw (\i+1-0.1,0) -- (\j,1) ;
            \pgfmathsetmacro{\j}{\pgtwo[\i]+0.1*(2*\pggtwo[\i]-1)}
            \draw (\i+1+0.1,0) -- (\j,1) ;
        }
                
        \draw (0.2,0) node[left] {($\epsilon_k = 3$) \quad\quad time-step $k$} ;
        \draw (0.2,1) node[left] {($\pi_{k,1} = 2, \pi_{k,2} = 4$) \quad time-step $k$$+$$1$} ;
        \end{tikzpicture}
    \end{center}
             
    \item (selfing) With probability $\alpha_N$, 
    %one individual is chosen uniformly at random to have a selfing event whereby the mother-father-offspring triplet consists of an identical mother-father pair. 
    a parent-offspring pair with labels $(\pi_{k},\, \epsilon_k) \in I ^ 2$ is chosen uniformly at random without replacement. %In particular, $\pi_{k,1}= \pi_{k,2}$ (selfing).
    The offspring $\epsilon_k$ in time-step $k$ is the offspring of a single individual (the parent) with label $\pi_{k}$ at time-step $k + 1$. %This can be viewed as having $\pi_{k,2}=\pi_{k,1}$ in the previous reproduction event.
                
    As above, the offspring is produced according to Mendel's laws which means each of the two gene copies in the parent is independently and equally likely to be inherited by each of those of the offspring. %In particular, with probability 1/2 both gene copies of the offspring came from the same gene copy of the parent.
    An example selfing event in a population of size $N=7$ is depicted below. 
    \smallskip
    \begin{center}
      \def\pgone{{1,2,4,4,5,6,7}} % parents for gene 1 in each ind.
      \def\pggone{{0,0,1,0,0,0,0}} % index of parental gene for gene 1 in each ind.
      \def\pgtwo{{1,2,4,4,5,6,7}} % parents for gene 2 in each ind.
      \def\pggtwo{{1,1,0,1,1,1,1}} % index of parental gene for gene 2 in each ind.
      
      \begin{tikzpicture}
        \foreach \y in {0,1}
        \foreach \x in {1,2,3,6,7} {
          \filldraw[black] (\x -0.1,\y) circle[radius=0.05]
                           (\x +0.1,\y) circle[radius=0.05];
          \draw (\x,\y) ellipse[x radius=0.35, y radius=0.25];
        }
        \foreach \y in {1}
            \foreach \x in {4,5} {
          \filldraw[black] (\x -0.1,\y) circle[radius=0.05]
                           (\x +0.1,\y) circle[radius=0.05];
          \draw (\x,\y) ellipse[x radius=0.35, y radius=0.25];
        };
        \foreach \y in {0}
            \foreach \x in {2,3,4,5,6} {
          \filldraw[black] (\x -0.1,\y) circle[radius=0.05]
                           (\x +0.1,\y) circle[radius=0.05];
          \draw (\x,\y) ellipse[x radius=0.35, y radius=0.25];
        };
    ;
    
        \foreach \i in {0,1,2,3,4,5,6} {
          \pgfmathsetmacro{\j}{\pgone[\i]+0.1*(2*\pggone[\i]-1)}
          \draw (\i+1-0.1,0) -- (\j,1) ;
          \pgfmathsetmacro{\j}{\pgtwo[\i]+0.1*(2*\pggtwo[\i]-1)}
          \draw (\i+1+0.1,0) -- (\j,1) ;
        }
    
        \draw (0.2,0) node[left] {($\epsilon_k = 3$) \quad\quad \quad time-step $k$} ;
        \draw (0.2,1) node[left] {\qquad \qquad ($\pi_{k,1} = \pi_{k,2} = 4$) \quad \quad time-step $k$$+$$1$} ;
      \end{tikzpicture}
    \end{center}
\end{enumerate}

Reproduction events in different time-steps are independent.
%These two possibilities happen independently for all time-steps $k\in \Z_{\geq 0}$. 
Note that at time-step $k$, the offspring $\epsilon_k$ is a new individual introduced and it replaces the one with the same label in the previous step $k+1$ (death-birth for the label $\epsilon_k$). The remaining $N-1$ individuals in time-step $k$ are the same individuals as those in time-step $k+1$. Hence this model has overlapping generations.

\begin{remark}\rm
When $\alpha_N=0$, our model is exactly the one introduced in \citet{coron2022pedigree}. When $\alpha_N$ is equal to a fixed $p$ in $[0,1)$, our model is exactly the one in \citet{linder2009}.
\end{remark}

\subsection{The pedigree as important partial information of the population}  
The above population dynamics generates a random graph, which we will call $\mathcal{G}_N$, with vertex set $I \times \mathbb{Z}_{+}$ (the $N$ individuals at all time-steps) and directed edges joining each individual with its parent(s). This random graph $\mathcal{G}_N$ is the population pedigree.  For the moment, let us consider this pedigree without specifying outcomes of genetic transmission. 
 
Between consecutive time-steps with two distinct parents, the pedigree has one edge from offspring $(\epsilon_k, k)$  to  parent $(\pi_{k,1}, k + 1)$, one from offspring $(\epsilon_k, k)$ to  parent $(\pi_{k,2}, k + 1)$, and  $N - 1$ single edges for the same individual from $(j, k)$ to $(j, k + 1)$ for $j \in I \setminus \{\epsilon_k\}$. The portion of the pedigree corresponding to the  example outcrossing event above is
\begin{linenomath*}
\begin{equation}\label{Fig: Moran_non_selfing}
  \xymatrix@=1.8em@C=1.8em
    {
    {( \pi_{k,1} = 2, \pi_{k,2} = 4) \quad \text{time-step }k+1} & 1  & 2  & 3 & 4 & 5 & 6 & 7
    \\ 
    {({ \epsilon_k = 3}) \qquad \text{time-step }k} & 1 \ar[u]\ar[u] & 2 \ar[u]\ar[u] & 3 \ar[ul]\ar[ur] & 4 \ar[u]\ar[u] & 5\ar[u]\ar[u] & 6\ar[u]\ar[u] & 7\ar[u]\ar[u]
    }
\end{equation}
\end{linenomath*}
Between consecutive time-steps with  selfing, the pedigree has a double edge from $(\epsilon_k, k)$ to $(\pi_{k}, k + 1)$, and again $N - 1$ single edges from $(j, k)$ to $(j, k + 1)$ for $j \in I \setminus \{\epsilon_k\}$. The portion of the pedigree corresponding to the  example selfing event above is
\begin{linenomath*}
\begin{equation}\label{Fig: Moran_selfing}
      \xymatrix@=1.8em@C=1.8em
    {
    {(\pi_{k,1}  = \pi_{k,2}= 4 ) \quad \text{time-step }k+1} & 1  & 2  & 3 &  4 &  5 & 6 & 7
    \\ 
    {(\epsilon_k = 3) \qquad \text{time-step }k}& 1 \ar[u]\ar[u] & 2 \ar[u]\ar[u] & 3 \ar@/_/[ur]\ar[ur] & 4 \ar[u]\ar[u] & 5 \ar[u]\ar[u] & 6 \ar[u]\ar[u] & 7\ar[u]\ar[u]
    }
\end{equation} 
\end{linenomath*}
Hence, the offspring always has 2 edges coming out of it which trace upward (towards the past). All the other $N-1$ individuals have a single edge.

Repetition of this process with these two possibilities for each time step backward into the indefinite past results in a single realized pedigree of the population.  Patterns of genetic ancestry at each locus are outcomes of Mendelian transmission conditional on this one pedigree.   Unlinked loci are transmitted independently through the pedigree.  In Figure~\ref{fig:sample_genealogy}, we illustrate a realization of both the population dynamics with genetic transmission and the corresponding pedigree for the first 4 generations in the past for a population of size $6$.

\begin{figure}[ht]
\centering
\includegraphics[scale=0.6]{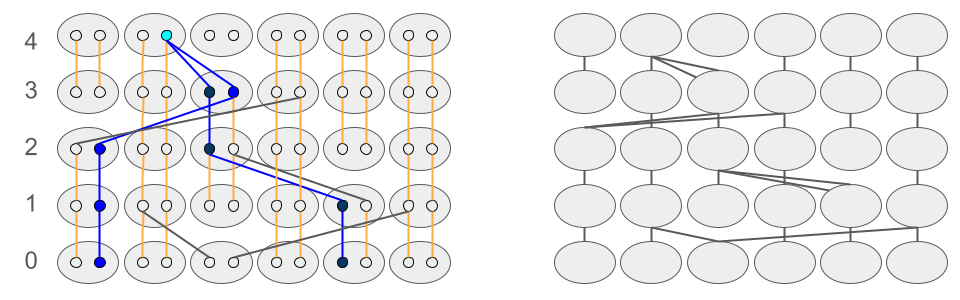}
\caption{\small A realization of our diploid Moran process with $N = 6$ individuals including the genetic transmission events at one locus (left image) and the corresponding pedigree (right image). Also on the left, two gene copies $(X_0, Y_0) = (2, 9)$ are sampled in the present time-step $0$, and their lineages are highlighted by the solid dots. Namely, $\{(X_k, Y_k)\}_{k=0}^4=\{(2,9), (2,9), (2,5), (6,5), (4,4)\}$. These two lineages coalesce in time-step $4$ because $X_4=Y_4$ but $X_k\neq Y_k$ for $k=0,1,2,3$.}
\label{fig:sample_genealogy}
\end{figure}

\subsection{Pairwise coalescence time: unconditional versus conditional on the pedigree} 

In our population model each of the $N$ individuals has two copies of each genetic locus. %one inherited from a mother and one from a father when the individual has two parents. 
Suppose we number these $J = \{1,...,2N\}$ such that the individual with label $i\in I$ has gene copies $ 2i-1$ and $2i$.
%the $j$-th position corresponds to the first or second gene of the $\floor{\frac{j + 1}{2}}$th individual, depending on parity. 
Broadly, our aim is to describe the statistical properties of  random samples of gene copies from $J$ by following their ancestral genetic lineages backward in time.  

Suppose we sample two gene copies $X_0$ and $Y_0$ without replacement from the $2N$ present at time $k=0$. Let $X_k$ and $Y_k$ be their ancestral lineages $k$ time-steps in the past. Then $(X_k)_{k\in\Z_+}$ and $(Y_k)_{k\in\Z_+}$ are two correlated Markov chains with state space $J$ that have the same transition probabilities. Our main object of study, for sample size $n=2$, is 
\begin{linenomath*}
\begin{equation}\label{Def:tauN2}
    \tau^{(N)}:= \inf\{k\in\Z_+:\,X_k = Y_k\},
\end{equation}
\end{linenomath*}
the (pairwise) coalescence time of the sample of two genes. The left image in Figure~\ref{fig:sample_genealogy} shows a realization of the population process in which $\tau^{(6)}=4$. 

The value of $\tau^{(N)}$ is completely determined by the population process. However, knowing the pedigree is {\it not enough} to determine the exact value of $\tau^{(N)}$, due to Mendelian randomness. For instance, even if we know the pedigree is exactly the one on the right in Figure~\ref{fig:sample_genealogy} and we know that $X_0=2$, we do not know what $X_3$ is: $X_3$ can be in one of the genes $\{5,6,7,8\}$ with equal probability.  

%It may also be understood as the first time-step $k$ where $|z_k ^ {(2,N)}| = 1$.

\section{Unconditional distribution of pairwise coalescence time}\label{S:Uncondition}

%Consider the ancestral process of a sample of two genes,  obtained by tracing the ancestral lineages of the two gene copies backward-in-time.
The unconditional distribution of $\tau^{(N)}$ (i.e. averaging over the randomness of the pedigree) can be obtained by considering a Markov process with 3 states $\{\text{diff},\, \text{same}, \,\text{coal}\}$ representing respectively that the two gene copies are in two different individuals, the gene copies are in the same individual but have not coalesced, and the gene copies have coalesced.
%where  $``\text{coal}''$ represents that the gene copies have coalesced, $``\text{same}''$ represents the two lineages of the gene copies are in the same individual and have not coalesced, and $``\text{diff}''$ represents that the two lineages are in two difference individuals.
\begin{linenomath*}
    \begin{align*}
    {
    \setlength{\arraycolsep}{3pt}
        \text{diff} = \begin{pmatrix}
            {\color{blue} \bullet }&  
        \end{pmatrix} \begin{pmatrix}
            {\color{teal} \bullet }&  
        \end{pmatrix} 
        \,
        \text{same} = \begin{pmatrix}
            {\color{blue} \bullet }&  
             {\color{teal} \bullet }
        \end{pmatrix},\ 
        \text{coal}=\begin{pmatrix}
            {\color{cyan} \bullet }
              &
        \end{pmatrix}
        }
    \end{align*}
\end{linenomath*}
The one-step transition matrix $\mathbf{\Pi}_N$  of this Markov process, under our diploid Moran model, is 
\begin{center}
	$\mathbf{\Pi}_N\,:=\,$
	\begin{tabular} {c|ccc}
		& $\text{diff}$  & $\text{same}$ & $\text{coal}$\\ \hline
		$\text{diff}$  &  $1 - \frac{2}{N(N-1)}$ & $\frac{1}{N(N-1)}$ & $\frac{1}{N(N-1)}$ \\[6pt]
		$\text{same}$  & $(1-\alpha_N) \frac{1}{N}$ & $\frac{N-1}{N} + \alpha_N \frac{1}{2N}$  &  $\alpha_N \frac{1}{2N}$\\[6pt]
		$\text{coal}$  &  $0$  & $0$ &  $1$
	\end{tabular}
\end{center}

For $q\in \{ \text{diff},\,\text{same}, \,\text{coal}\}$, we let  $\mathbb P_q$  be the  probability when the initial state (i.e.\ the sampling configuration) is $q$. We will be interested in the probability the coalescence time exceeds some particular value, or the survival function, which is one minus the cumulative distribution function of the coalescence time.  Thus our statements about convergence, in Theorem~\ref{theorem: unconditional_pairwise_coalescence_time_convergence} and later Theorems~\ref{T:MAIN_conditional} and~\ref{T:MAIN_conditional_same}, are about convergence in distribution.   

%The unconditional (averaging over the pedigree) limiting distribution of $N^{-2}\tau^{(N)}$ by is as follows.   
\begin{theorem}[Unconditional limiting distribution]\label{theorem: unconditional_pairwise_coalescence_time_convergence}
Suppose $\alpha_N \rightarrow \alpha \in[0,1]$ as $N\to\infty$. Then $N^{-2}\tau^{(N)}$ converges in distribution: for any fixed $t \geq 0$
\begin{linenomath*}
    \[
    \PP_q{\left( N^{-2} \tau^{(N)} > t\right)} \to
    \begin{cases}
    \frac{2-2\alpha}{2-\alpha}e^{-\frac{2}{2-\alpha}t} \qquad &\text{when }\; q = \rm same\\[0.5em]
    e^{-\frac{2}{2-\alpha} t} \qquad &\text{when }\; q = \rm diff
    \end{cases}
    \]
\end{linenomath*}
%    %to an exponential random variable with rate $\frac{2}{2-\alpha}$ under each of the laws $\PP_{\rm diff}(\cdot)$ and $\PP_{\rm same}(\cdot | \mathcal{O} \geq 1)$, and $N^{-2}\tau^{(N)}$ converges to zero under  the law  $\PP_{\rm same}(\cdot | \mathcal{O}=0)$,
%     where
% \begin{linenomath*}
%     $$\mathcal{O}:=|\{k\geq 1:\,\hat{X}_k= \hat{Y}_k,\,\hat{X}_{k-1}\neq \hat{Y}_{k-1}\}|$$
% \end{linenomath*}
%     is the total  number of time-steps when the two sample lineages transition from belonging to two distinct individuals to belonging to a single individual,   
%     and $\hat{X}_k, \hat{Y}_k$ denote the labels of the individuals to which the genes $X_k$ and $Y_k$ belong.
    \end{theorem}

The convergence under $\Pdiff$ of the rescaled coalescence time $N^{-2}\tau^{(N)}$ to an exponential random variable ${\rm Exp}\left(\frac{2}{2-\alpha}\right)$ in Theorem~\ref{theorem: unconditional_pairwise_coalescence_time_convergence} shows that our Moran model gives the same result as the Wright-Fisher model with partial selfing \citep{NordborgAndDonnelly1997,Mohle1998a} when one averages over the pedigree for a fixed probability of selfing $\alpha\in[0,1]$. \citet[Lemma 2]{Kogan2023.10.18.563014} recently obtained this result in the broader context of selfing plus two-locus recombination under Wright-Fisher reproduction (and averaging over the pedigree).  

The proof of Theorem~\ref{theorem: unconditional_pairwise_coalescence_time_convergence} is given in the Appendix, Section~\ref{A:unconditional}.  Two kinds of events are key to our analysis: \textit{overlap events} where two sample lineages in distinct individuals transition in one time-step to belonging to the same individual, and \textit{splitting events} where two sample lineages in the same individual transition to belonging to two distinct individuals. We note that these correspond, respectively, to the `parent-sharing' events and `separation' events in the two-locus unconditional model of \citet{Uyenoyama2024}. Splitting events occur at outcrossing events, whereas overlap events may be due either to selfing or outcrossing. 
    % We use $\mathcal{O}$ to denote the number of overlap events. Note that an overlap event results also in coalescence with probability $1/2$.  %The event $\{ \mathcal{O} =0\}$ is the one where two sample lineages never overlap which is impossible 
    % In addition, under $\Psame$ the two lineages may coalesce before they split, in which case $\mathcal{O} = 0$.
    % %Conversely, under $\Psame$ we have that $\mathcal{O} \geq 1$ only if the two sample lineages undergo a split before coalescing.
    % Under $\Pdiff$ it holds that $\mathcal{O}\geq 1$ almost surely.
    % We prove in Lemma~\ref{lemma: unconditional_tau_N_2_decomposition} that
    % \begin{linenomath*}
    %         \begin{equation}\label{E:POzero}
    %          \mathbb{P}_{\rm diff}(\mathcal{O}=0)=0  \quad \text{and}\quad    \mathbb{P}_{\rm same}(\mathcal{O}=0)=   \frac{\alpha_N}{2-\alpha_N}.
    %         \end{equation}
    % \end{linenomath*}

\begin{remark}\rm
    For any $k\in\Z_+$,
\begin{linenomath*}
    \begin{align*}
        \PP_{\rm diff}(\tau^{(N)} > k)
        =&\, 
        \left(1,0,0\right)\mathbf{\Pi}_N^{k}\left (1,1,0\right)^{T} \text{ and }\\
        \PP_{\rm same}(\tau^{(N)} > k)
        =&\, 
        \left(0,1,0\right)\mathbf{\Pi}_N^{k}\left (1,1,0\right)^{T}
    \end{align*}
\end{linenomath*}
exactly in discrete time.
%The initial vector $(1,0,0)$ enforces our assumed starting state, ${\rm diff}$. The end vector $(1,1,0)^{T}$ enforces the requirement that the lineages remain distinct at time-step $k$.    
Like its Wright-Fisher counterpart \citep{NordborgAndDonnelly1997,Mohle1998a}, Theorem~\ref{theorem: unconditional_pairwise_coalescence_time_convergence} is a separation-of-timescales result.  Direct application of \citet[Lemma 1]{Mohle1998a} or its generalization by \citet{mohle2016extension} is complicated by the extra factor of $N$ in our diploid Moran model, i.e.\ to go from time in discrete steps to time in units of $N$ generations.   \citet[Proposition 1]{KroumiAndLessard2015} have extended M\"{o}hle's lemma to the Moran model. In Section~\ref{A:unconditional}, we instead use our notions of overlap and splitting events to prove convergence.    %When $\alpha_N\to \alpha \in [0,1)$, for instance, the matrix $\mathbf{\Pi}_N$ is of the form $A+\frac{B}{N}+\frac{C}{N^2}$ and  \citet[Theorem 1]{Mohle1998a} does not apply directly. 
    %the time-rescaling for M\"{o}hle's lemma is of order $N$. The time-rescaled process given by the lemma is trivial, jumping instantaneously under $\Psame$ to either coal or diff before remaining fixed in whatever state it jumped to, and remaining fixed in diff under $\Pdiff$.
    %One must therefore make a more careful analysis of the distribution of $\tau^{(N)}$ to achieve a limiting distribution for the proper time-rescaling. 
    %The goal of the proceeding section is to describe, for each fixed population size $N$, the distribution of $\tau^{(N)}$, and to discover a limiting distribution for an appropriate time-rescaling of $\tau^{(N)}$. 
\end{remark}

    % \JW{can we use \citet{mohle2016extension}?}

    % \MN{It appears that a direct proof is difficult using \citet{mohle2016extension}. Indeed, if we
    % expand $\Pi_N$ as $I+\frac{1}{N}Q_N + \frac{1}{N(N-1)}B_N$ our matrix $Q$ depends on $N$, which violates the assumptions of the theorem. And even a straightforward adaption like supposing $\frac{c_N}{d_N} Q_N$ converges to the zero matrix may not be enough.
    
    % It should be said that for the critical (now limited-outcrossing) case a straightforward use of \citet{mohle2016extension} is possible. The trouble is in the subcritical (now partial-selfing) regime.}

    % \MN{Yes, we can.}
    % For $N(1-\alpha_N)\to\infty$, we can write
    % \begin{equation*}
    %     \Pi_N = I + \frac{1}{N}
    %     \begin{pmatrix}
    %         0 & 0 & 0 \\
    %         1-\alpha & -1 + \frac{\alpha}{2} & \frac{\alpha}{2}\\
    %         0 & 0 & 0
    %     \end{pmatrix}
    %     + \frac{1}{N(N-1)}
    %     \begin{pmatrix}
            
    %     \end{pmatrix}
    % \end{equation*}
    
   % The key in our analysis is to keep track of two kinds of event: overlapping events where two sample lineages in distinct individuals transition to belonging to the same individual, and splits/splittings events where two sample lineages in the same individual transition to belonging to two distinct individuals.
        
\section{Main results: conditional coalescence times given the pedigree}\label{S:condition}

    Our main results are about the conditional distribution of $\tau^{(N)}$ given the pedigree and the sampled pair of individuals.  
    We consider a sequence of our diploid Moran model indexed by $N$, under three different regimes for the selfing probability $\alpha_N$ in the limit.  
    The first corresponds to the model of partial selfing previously studied by others, but holds only when $\alpha_N \to 1$ more slowly than $1/N$.
    The second is what we call limited outcrossing, and happens when $\alpha_N \to 1$ at rate $\lambda/N$ for some $\lambda \in (0,\infty)$.  
    The third is when outcrossing is negligible, because $\alpha_N \to 1$ faster than $1/N$.  As we will show, limited outcrossing interpolates between partial selfing and negligible outcrossing.  As $\lambda\to\infty$, it becomes indistinguishable from partial selfing with $\alpha=1$.  As $\lambda\to0$ it becomes indistinguishable from negligible outcrossing.  

    Following the notation in \citet{DFBW24}, we let $\mathcal{A}_N$ be the $\sigma$-algebra generated by both the pedigree $\mathcal{G}_N$ and the labels $\hat{X}_0$ and $\hat{Y}_0$ of the sampled individuals.
    That is, we let
    \begin{linenomath*}
    \begin{equation}\label{Def:AN}
        \mathcal{A}_N :=\sigma{\left(\mathcal{G}_N,\,\hat{X}_0,\,\hat{Y}_0\right)}
    \end{equation}
    \end{linenomath*}
    which by definition contains all the information relevant to the distribution of the coalescence time.
    Our main results (Theorems~\ref{T:MAIN_conditional}~and~\ref{T:MAIN_conditional_same}) say that the conditional distributions of $N^{-2}\tau^{(N)}$ given $\mathcal{A}_N$ converge to different distributions,
    depending on the sampling configuration (diff or same) and on the three regimes of the selfing probability. Furthermore, the limiting distributions can retain information of the random pedigree even as $N\to\infty$.   

    We now  state our main results for coalescence times in Theorems~\ref{T:MAIN_conditional}~and~\ref{T:MAIN_conditional_same}. 
    %Recall the pairwise coalescence time $\tau^{(N)}$ defined in \eqref{Def:tauN2}.
    \begin{theorem}\label{T:MAIN_conditional} 
        For any $t\in\R_+ =[0,\infty)$, as $N\to\infty$,
            \begin{linenomath*}
            \[
            \Pdiff{\left( N^{-2} \tau^{(N)}>t \,|\, \mathcal{A}_N\right)} \to 
                        \begin{dcases}
                            e^{- \frac{2}{2-\alpha}t} 
                            & \quad\text{if} \quad N(1-\alpha_N)\rightarrow \infty \;\text{ and }\,\alpha_N \rightarrow \alpha \in [0,1]\\[0.25em]
                            \mathbb{P}(T_\lambda>t\, |\, G_\lambda)
                            & \quad\text{if} \quad N(1-\alpha_N) \rightarrow \lambda \in \R_+\\[0.25em]
                            % \textbf{1}_{\{{\rm Exp}(2)>t\}}
                            \mathbbm{1}_{\{{\rm Exp}(2)>t\}}
                            & \quad\text{if} \quad N(1-\alpha_N) \rightarrow 0
                        \end{dcases},
            \]
        \end{linenomath*}
        where $G_\lambda$ is the ancestral graph, defined in Definition~\ref{D: ancestral_graph}, and $T_\lambda$ is the coalescence time for a pair of independent simple random walks on $G_\lambda$, defined in Definition~\ref{D: T_lambda}.
    \end{theorem}

The proof of Theorem~\ref{T:MAIN_conditional} is given in the Appendix, Section~\ref{section: proofs}.

\begin{remark}\rm
        This result contrasts with the  asymptotic behavior of the marginal or unconditional probability $\Pdiff{\left( N^{-2} \tau^{(N)}>t \right)}$ in Theorem~\ref{theorem: unconditional_pairwise_coalescence_time_convergence}.  For example, assuming that $\alpha_N = 1$ and assuming that $\alpha_N \to 1$ at rate $N^{-2}$ lead to the same result in Theorem~\ref{theorem: unconditional_pairwise_coalescence_time_convergence} but correspond to negligible outcrossing and partial selfing, respectively, in Theorem~\ref{T:MAIN_conditional}.
    \end{remark}

    Theorem~\ref{T:MAIN_conditional} offers a description of the conditional distribution of the coalescence time $\tau^{(N)}$ for a sample of two gene copies from different individuals in a population of size $N$ given the pedigree. Simulations of the CDFs for five pedigrees under each of the three cases in Theorem~\ref{T:MAIN_conditional} are shown in Figure~\ref{fig:3_regimes_3_figures}.
    Theorem~\ref{T:MAIN_conditional} says that the law of $N^{-2}\tau^{(N)}$ for two lineages in different individuals is robust to the structure of the pedigree so long as the rate at which two lineages in the same individual undergo a splitting event, $(1-\alpha_N)N^{-1}$, is much greater than the rate at which two sample lineages in distinct individuals coalesce, $N^{-2} + O(N^{-3})$, specifically that the ratio of the former over the latter, approximately $N(1-\alpha_N)$, tends to infinity. % (see also Section~\ref{S:why3regimes}). 
    This robustness for the case of partial selfing is simply that the conditional limit agrees exactly with the marginal limit in Theorem~\ref{theorem: unconditional_pairwise_coalescence_time_convergence}. 
    
    If, on the other hand, the rate of coalescence is comparable to or dominates the rate of splitting of an individual, as in the cases of limited outcrossing and negligible outcrossing, Theorem~\ref{T:MAIN_conditional} demonstrates the non-robustness of the ``diff'' pairwise coalescence time to the pedigree. That is, the survival function of the conditional pairwise coalescence time does not converge to its mean.  Indeed, in this case the rates at which lineages split backwards in time is then comparable to the rate at which they coalesce. The subgraph of the pedigree which encodes the possible trajectories of the sample lineages takes the shape of an ancestral graph $G_\lambda$ as the population size $N$ goes to infinity, and the ancestries of the two sample lineages correspond to coalescing random walks thereon with coalescence time $T_\lambda$. Correspondingly, $N^{-2}\tau^{(N)}$ conditional on the pedigree asymptotically takes the distribution of $T_\lambda$ conditional on the ancestral graph.

    \begin{figure}[ht]
        \centering
        \includegraphics[scale=0.37]{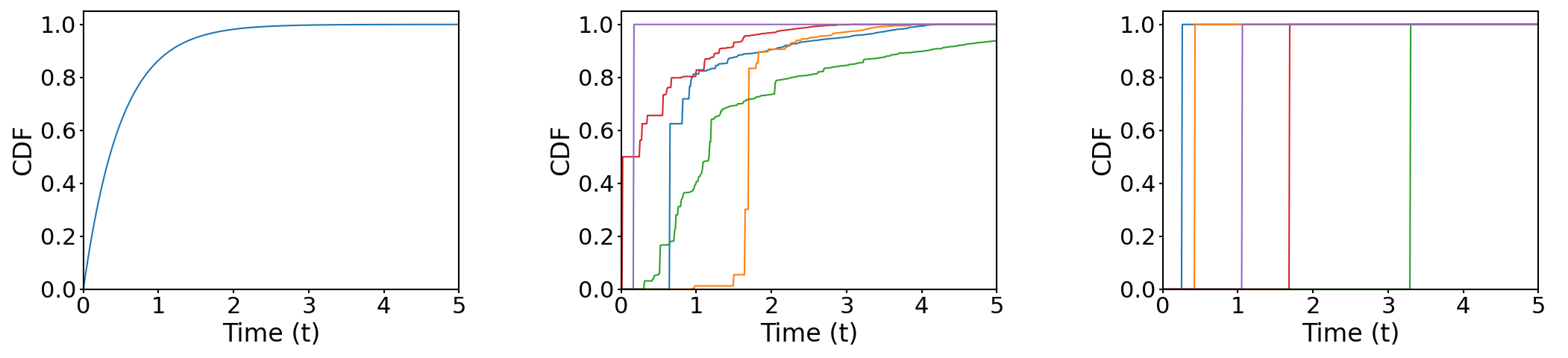}
        \caption{\small The limiting conditional CDF  
        $t\mapsto \lim_{N\to\infty}\Pdiff(N^{-2}\tau^{(N)} \leq  t | \mathcal{A}_N)$ for our Moran model under three assumptions about $\alpha_N$: partial selfing, limited outcrossing, and negligible outcrossing.  (\textbf{Left})
        Partial selfing, where $\lim_{N} N(1-\alpha_N)=\infty$. The limiting CDF is deterministic and is exactly the CDF of an exponential random variable ${\rm Exp}\left(\frac{2}{2-\alpha}\right)$, with $\alpha=1$ for comparison with those at center and right. (\textbf{Center}) Five realizations of the limiting CDF for $\lambda = 5$ under limited outcrossing, i.e.\ with $\lim_{N} N(1-\alpha_N)=\lambda = 5$. 
        The limiting CDF is random  and is described precisely in Theorem~\ref{T:MAIN_conditional}. 
        (\textbf{Right}) Five realizations of the limiting CDF for the case of negligible outcrossing, where $\lim_{N} N(1-\alpha_N)=0$. The limiting CDF is random, and it is the CDF of a random constant that is exponentially distributed with rate 2, i.e.\ the CDF is a Heaviside function with exponentially distributed jump time.}
        \label{fig:3_regimes_3_figures}
    \end{figure}

    To complete the exposition of Theorem~\ref{T:MAIN_conditional}, we now formally define the ancestral graph $G_\lambda$ and the pairwise coalescence time $T_\lambda$ of random walks theoreon.
    
    \begin{definition}[The ancestral graph $G_\lambda$]\label{D: ancestral_graph}
        Fix $\lambda \in \R_+$. By the ancestral graph $G_\lambda = \left(G_\lambda (t)\right)_{t \in \R_+}$ we denote a particle system beginning with two different particles, where each particle splits into two with independently with rate $\lambda$ and each pair of particles coagulate independently at rate $2$.
    \end{definition}
    The ancestral graph is equal in distribution to an ancestral recombination graph \citep{Griffiths1991, GriffithsAndMarjoram1997} or an ancestral selection graph \citep{kroneneuhauser1997,NeuhauserAndKrone1997}. Now consider coalescing random walks $x_\lambda$ and $y_\lambda$ on $G_\lambda$. Note that since $G_\lambda$ is the limiting object and $\alpha_N\to 1$, these walks coalesce with probability one the first time they meet.  Let $G_\lambda(0) = \{x(0), y(0)\}$ consist of the two different particles at $t = 0$. Furthermore, let $\left(x_\lambda(t), y_\lambda(t)\right)_{t \in \R_+}$ be a $G_\lambda$-valued continuous-time process starting at $(x(0), y(0))$ such that
    \begin{itemize}
        \item[(i)] $\{x_{\lambda}(t),y_{\lambda}(t)\} \subset G_{\lambda}(t)$ for all $t\in \R_+$,
        \item[(ii)] at any splitting event, each of 
        $(x_{\lambda}(t))_{t\in\R_+}$ and $(y_{\lambda}(t))_{t\in\R_+}$ will follow each of the two paths available with equal (i.e.\ 1/2) probability. 
    \end{itemize}
    Then, $(x_{\lambda}(t))_{t\in\R_+}$ and $(y_{\lambda}(t))_{t\in\R_+}$ are coalescing random walks on the ancestral graph $G_{\lambda}$ starting at two different particles.
    
    \begin{definition}[The coalescence time $T_\lambda$]\label{D: T_lambda}
        The first meeting time of $x_\lambda$ and $y_\lambda$ is
        \begin{linenomath*}
            \begin{equation}\label{Def:Tlambda}
             T_\lambda := \inf\{t\in\R_+:\,x_{\lambda}(t)=y_{\lambda}(t)\}.
            \end{equation}
        \end{linenomath*}
    \end{definition}

    How one may reconstruct the CDF of the $T_\lambda$, conditional on a realization of $G_\lambda,$ is discussed in Figure~\ref{fig: pedigree_limit}.  Two realizations of $T_\lambda$ on the same ancestral graph are demonstrated in Figure~\ref{fig: G_lambda_examples}.  

    \begin{figure}[ht]
		\centering
		\includegraphics[scale=0.4]{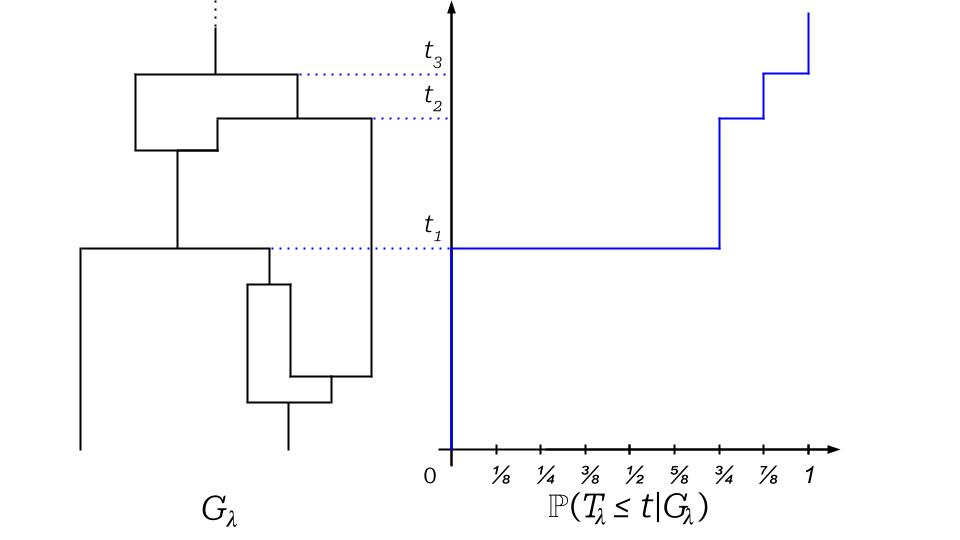}
		\caption{\small (\textbf{Left}) A realization of the random ancestral graph $G_{\lambda}$ starting with two particles/nodes. For this particular realization, the overlap times of the sample lineages are $t_1< t_2< t_3$.
                (\textbf{Right})
                The CDF corresponding to this  $G_\lambda$. The conditional distribution of the limiting coalescence time $T_\lambda$, given $G_\lambda$, is obtained by tracing ancestral genetic lineages backwards in time along the  graph $G_\lambda$.  Given this $G_\lambda$, $T_\lambda$ must take values in $\{t_1, t_2, t_3\}$. One can read off the conditional  CDF of $T_\lambda$ as follows, by tracing the ancestries of a hypothetically infinite number of unlinked loci.  The node of $G_{\lambda}$ at time $t_1$ contains $3/4$ of the ancestries from the right sample lineage and the whole of those from the left sample lineage, so $\mathbb{P}(T_\lambda=t_1 | G_\lambda ) = 3/4$. Between $t_1$ and $t_2$, half of the loci following the left ancestral lineage go left at the split and the other $1/2$ continue up. The lineage on the right between $t_1$ and $t_2$ contains $1/4$ of the ancestries from the right sample lineage. So when these two lineages overlap at $t_2$ we have $\mathbb{P}(T_\lambda = t_2 | G_\lambda) = 1/8$.  Finally the remainder of the right sample lineage meets the remainder of the left sample lineage at $t_3$ for $\mathbb{P}(T_\lambda = t_3) = 1/8$.}        
		\label{fig: pedigree_limit}
	\end{figure}

    Analogously to Theorem~\ref{T:MAIN_conditional}, we obtain the limiting conditional distribution of the coalescence time when our samples are taken from the same individual.
    
    \begin{theorem}\label{T:MAIN_conditional_same}
                Suppose $\alpha_N \rightarrow \alpha \in [0,1]$ as $N\to\infty$. For all $t\in\R_+$ we have convergence in distribution
    \begin{linenomath*}
    \begin{equation*}
    \Psame(N^{-2} \tau^{(N)} > t | \mathcal{A}_N) \to
        \begin{dcases}
        e^{-t} & \quad \text{if }\alpha=0\\
        2^{-U} e^{-\frac{2}{2-\alpha}t} & \quad \text{if }\alpha\in (0,1)\\[0.25em]
        0 & \quad \text{if }\alpha=1
        \end{dcases},
    \end{equation*}
    \end{linenomath*}
    where  $U$ is a random variable satisfying  $\mathbb{P}(U = k) = \alpha^k (1-\alpha)$ for $k\in\mathbb{Z}_+$.
    \end{theorem}

The proof of Theorem~\ref{T:MAIN_conditional_same} is given in the Appendix, Section~\ref{section: proofs}.

    %\begin{remark}\rm
    As in \eqref{eq:PUk}, the quantity $U$ represents the number of selfing events in the ancestry of the sampled individual, before the first splitting or outcrossing event. To make the connection with Theorem~\ref{T:MAIN_conditional}, the first two cases in Theorem~\ref{T:MAIN_conditional_same} pertain to partial selfing (or case one) in Theorem~\ref{T:MAIN_conditional}, where is the possible for a ``same'' sample to enter the ``diff'' process.  The third case in Theorem~\ref{T:MAIN_conditional_same} pertains to both limited outcrossing and negligible outcrossing (or cases two and three) in Theorem~\ref{T:MAIN_conditional}, where the two gene copies in a ``same'' sample coalesce instantaneously with probability $1$. %Different possible values of $U$ are one aspect of the different possible pedigrees of the population.

\subsection{Extension to larger samples}\label{S:condition_n}

We anticipate that our main results can readily be extended to samples of arbitrary size $n>2$.  We offer a detailed conjecture here but leave further exploration and proof to future work.
        
Suppose  we have a sample of $n$ genes from the population at time-step $0$, where $n>2$ is an arbitrary fixed number. As reviewed in \citet{Berestycki2009}, it is customary in formal treatments of coalescent theory to summarize the ancestral relationships of the sample by an ancestral process which we can call $(\Pi_k^{N,n})_{k \in \Z_+}$ and which takes values in  the space $\mathcal{E}_n$ of partitions of $\{1,\ldots,n\}$,  where indices belong to the same block of $\Pi^{N,n}_{k}$ if and only if the samples with these indices have a common ancestor $k$ steps in the past. Under $\Pdiff$, which now pertains to $n$ samples from $n$ distinct individuals, the initial state $\Pi_0^{N,n}$ of the ancestral process has $n$ singleton blocks. Furthermore, let $\mathcal{A}_{N}$ now be the $\sigma$-algebra generated by both the pedigree and the labels of the $n$ individuals from whom we sample the $n$ sample lineages. We make the following conjecture about diff samples.
        
\begin{conjecture}
Under the conditional law $\Pdiff(\cdot | \mathcal{A}_N)$, the sequence of $\mathcal{E}_n$-valued processes $\Pi^{N,n} = (\Pi^{N,n}_{\lfloor t N^2 \rfloor})_{t \in \mathbb{R}_+}$ converges in finite dimensional distributions, as $N\to\infty$, to the random process
\begin{linenomath*}
    \begin{equation}\label{eq:conjn}
            \Pi = 
            \begin{cases}
                \mathcal{K}^{(\alpha)}_n & \quad\text{if } N(1-\alpha_N)\rightarrow \infty \quad\text{and } \alpha_N \rightarrow \alpha\\
                P_\lambda \quad\text{under }\mathbb{P}(\cdot | G_\lambda^n) & \quad\text{if } N(1-\alpha_N)\rightarrow \lambda \in \R_+ %\\
               %K^{(1)}_n & \quad\text{if } N(1-\alpha_N)\rightarrow 0
            \end{cases}, 
    \end{equation}
\end{linenomath*}
where $\mathcal{K}^{(\alpha)}_n$ is the the Kingman $n$-coalescent \citep{kingman1982} with time rescaled by $\frac{2}{2-\alpha}$, $G_\lambda^n$ is an ancestral graph starting with $n$ nodes, and $P_\lambda = (P_\lambda(t))_{t \in \mathbb{R}_+}$ is a $\mathcal{E}_n$-valued process on $G_\lambda^n$ which is the extension of \eqref{Def:Tlambda} to larger samples, so that $i \sim j$ (coalesce) if and only if $x_{i,\lambda}(t) = x_{j,\lambda}(t)$.   
\end{conjecture}

%\begin{remark}
%    For references on the standard Kingman $n$-coalescent we refer the interested reader to \cite{kingman1982, Berestycki2009, Wakeley2009}.
%\end{remark}

Furthermore we expect that as $\lambda \rightarrow \infty$, $P_\lambda$ under $\mathbb{P}(\cdot | G_\lambda^n)$ converges  to the Kingman $n$-coalescent $\mathcal{K}^{(\alpha)}_n$ as implied by \eqref{E:extreme_lambda} below in Section~\ref{sec:limitedprops}, and that as $\lambda \rightarrow 0$, $P_\lambda$ under $\mathbb{P}(\cdot | G_\lambda^n)$ converges  to a single realization of $\mathcal{K}^{(1)}_n$, which is the limit in the negligible-outcrossing regime.

For the first case in \eqref{eq:conjn} we further conjecture, for a sample of size $n$ containing $2m$ gene copies together as pairs in $m$ individuals and $n-2m$ gene copies each in different individuals, that each of the $m$ individuals will have its own number of selfing generations before its first outcrossing event and that these will be given by $m$ independent draws of the random variable $U$ in \eqref{eq:PUk}.  This is a straightforward extension of the coalescent with partial selfing in \citet{NordborgAndDonnelly1997}.

\section{Additional properties of limited outcrossing}\label{sec:limitedprops}

Here we provide further analyses and illustrations of the novel scaling limits in Theorems~\ref{T:MAIN_conditional}~and~\ref{T:MAIN_conditional_same}, in particular the coalescing random walk on the ancestral graph. We note that the simulations for the Wright-Fisher model  in Figure \ref{fig:1} show similar behavior as those for our Moran model in Figure \ref{fig:3_regimes_3_figures}.

Observe first that the limited-outcrossing regime in Theorem~\ref{T:MAIN_conditional} includes the case $\lambda=0$, which means that the distribution of $T_{\lambda}$ under $ \mathbb{P}(\cdot | G_\lambda)$ when $\lambda=0$ is equal to that expected under negligible outcrossing, namely a single draw from an exponential distribution with rate 2.  The other extreme, $\lambda\to\infty$, is covered by \eqref{E:extreme_lambda} below.  Namely, as $\lambda\to\infty$, the distribution of $T_{\lambda}$ under $\mathbb{P}(\cdot | G_\lambda)$ converges to ${\rm Exp}(2)$, which corresponds to the partial-selfing regime with $\alpha=1$. Hence, for lineages in different individuals, the limited-outcrossing regime covers all possible cases for which $\alpha_N\to 1$, as we tune the magnitude of $\lambda$. This highlights the significance of the limited-outcrossing regime.

To understand how information relevant to conditional gene genealogies can persist even in the limit $N \to \infty$, we introduce a \textit{discrete-time ancestral graph}, $G^N = \left(G^N(k)\right)_{k \in \Z_+}$, which is a subset of the pedigree containing all potential individuals through which our sample of $n=2$ gene copies may have been transmitted.  Each realization of $G^N$ is constructed by the following two rules.  We refer to potential ancestral individuals of the sample as nodes.  When a node is the result of an outcrossing event, it fragments into two nodes.  When one node is the offspring of another node in a selfing event, the two coagulate into a single node.  Individual genetic lineages are transmitted through this graph according to Mendel's laws as in Section~\ref{section: moran_model}. A realization of $G^N$ is shown in Figure~\ref{fig: discrete_ancestral_graph}.

\FloatBarrier
\begin{figure}[!ht]
  \centering
  \includegraphics[scale=0.4]{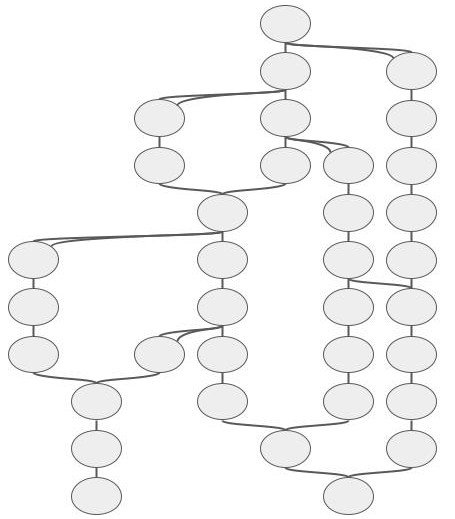}
  \caption{\small A realization of $G^N$ containing potential ancestors of the two sampled individuals (ovals) at the bottom. In past time-step $3$, the ancestral individual on the left undergoes  an outcrossing event. In past time-step $4$, one of the nodes from this event is the offspring of one of the potential ancestral individuals of the sampled individual on the right, and these two nodes coagulate. In the fifth time-step, the right most particle experiences an outcrossing in which exactly one of the parents is a node in the graph. This may appear in the discrete-time ancestral graph but will not occur in the limiting ancestral graph.} \label{fig: discrete_ancestral_graph}
\end{figure}
\FloatBarrier

This graphical structure will be poorly behaved as $N\to\infty$ if the number of nodes diverges as the population size goes to infinity. For the random variable $G^N$ to have a suitable limit as the population size goes to infinity we require that the rate at which nodes fragment, approximately $(1-\alpha_N)N^{-1}$, and the rate at which pairs of nodes coagulate, of order $N^{-2}$, are comparable. This comparability is taken to mean that the ratio of these two rates, $N(1-\alpha_N)$, converges as $N$ goes to infinity to $\lambda \in \R_+$. This reasoning gives rise to the following lemma.

\begin{lemma}\label{L: G_N_convergence_unrigorous}
  Suppose $N(1-\alpha_N) \to \lambda \in \R_+$. The discrete-time ancestral graph under the coalescence timescale $\left(G^N(\lfloor t N^2 \rfloor)\right)_{t \in \R_+}$ converges in distribution to the ancestral graph $G_\lambda$ as $N\to\infty$.
\end{lemma}
    
The proof of Lemma~\ref{L: G_N_convergence_unrigorous} is given in the Appendix, Section~\ref{section: proofs_critical}.  This requires additional details which we have omitted here for simplicity.  
% For example, the ancestral genetic lineages of the sample are not in fact confined to $G^N$ because it does not include coagulation events due to outcrossing.
The significance of Lemma~\ref{L: G_N_convergence_unrigorous} is that the simple subgraph $G^N$ captures all the relevant information in the pedigree for the limiting ancestral graph $G_\lambda$.  One difference between these graphs is that ancestral genetic lineages do not coalesce immediately at each overlap event in $G^N$, whereas they do necessarily at $T_\lambda$ in the limiting graph $G_\lambda$.

Next, we consider the variances of the
conditional survival probability $\mathbb{P}(T_\lambda > t | G_\lambda)$ and of the conditional expectation $\mathbb{E}\left[T_{\lambda}|G_{\lambda}\right]$. We obtain explicit formulas and their asymptotics as $\lambda\to 0$ and $\lambda\to \infty$. The starting point of our analysis is a pair of conditionally independent gene genealogies given the graph for the sampled individuals, cf.\  \citet{DFBW24} and \citet{Birkneretal2012}.  Figure~\ref{fig: G_lambda_examples} depicts realizations of coalescence times for two such loci.

\begin{lemma}\label{L:two_equalities}
For any $\lambda \in (0,\infty)$ and  $t \in (0,\infty)$, 
\begin{equation}\label{E:Var_Cov} 
 \mathrm{Var}{\left(\mathbb{P}(T_\lambda > t | G_\lambda)\right)}=\mathbb{P}(T_\lambda \wedge T_\lambda' > t) -e^{-4t}
\quad
\text{ and }\quad
\mathrm{Var}{\left(\mathbb{E}\left[T_{\lambda}|G_{\lambda}\right]\right)}
    %= \mathbb{E}[T_\lambda T_\lambda'] - \mathbb{E}[T_\lambda]^2 
    = \mathrm{Cov}{\left(T_\lambda, T_\lambda'\right)},
    \end{equation}
where
$T_\lambda$ and $T_\lambda'$ are two conditionally independent copies of the pairwise hitting time given the same $G_{\lambda}$, as defined in \eqref{Def:Tlambda} with  $(x_\lambda(0),\,y_\lambda(0))=(x_\lambda'(0),\,y_\lambda'(0))$ and $x_\lambda(0)\neq y_\lambda(0)$.
\end{lemma}

%\begin{remark}
%    Realizations of $T_\lambda$ and $T_\lambda'$ are presented in Figure~\ref{fig: G_lambda_examples}.
%\end{remark}

\begin{proof}
First, note that  $\mathbb{P}_{x_\lambda (0) \neq y_\lambda (0)}(T_\lambda > t | G_\lambda)$ has mean $\mathbb{P}_{x_\lambda (0) \neq y_\lambda (0)}(T_\lambda > t)=e^{-2t}$ which is the same for all $\lambda$. The last equality holds because every pair of particles in the ancestral graph $G_\lambda$ coagulates at rate $2$.
Therefore, the first equality in \eqref{E:Var_Cov} follows from the following equality:
\begin{align*}
\mathbb{E}\left[\mathbb{P}_{x_\lambda (0) \neq y_\lambda (0)}(T_\lambda > t | G_\lambda) ^ 2\right]
=&\,\mathbb{E}\left[
\mathbb{P}_{x_\lambda (0) \neq y_\lambda (0)}(T_\lambda > t | G_\lambda) \;\mathbb{P}_{x_\lambda' (0) \neq y_\lambda' (0)}(T'_\lambda > t | G_\lambda) \right]\\
=&\,\mathbb{E}\left[
\mathbb{P}_{x_\lambda' (0)=x_\lambda (0) \neq y_\lambda (0)=y_\lambda' (0)}(T_\lambda > t,\;T'_\lambda > t | G_\lambda)  \right]\\
%=&\, \mathbb{P}_{x_\lambda' (0)=x_\lambda (0) \neq y_\lambda (0)=y_\lambda' (0)}(T_\lambda > t,\;T'_\lambda > t ) \\
=&\, \mathbb{P}_{x_\lambda' (0)=x_\lambda (0) \neq y_\lambda (0)=y_\lambda' (0)}(T_\lambda \wedge T'_\lambda > t ).
\end{align*}
The second equality in \eqref{E:Var_Cov} follows from the conditional independence of $T_\lambda$ and $T_\lambda'$.
\end{proof}

%It following directly from the first equality in \eqref{E:Var_Cov}  that the variance of $\mathbb{P}(T_\lambda > t | G_\lambda)$ is equal to $\mathbb{P}(T_\lambda \wedge T_\lambda' > t) - e ^{-4t}$.

%\medskip

\FloatBarrier
		\begin{figure}[ht]
			\centering
			\includegraphics[scale=0.35]{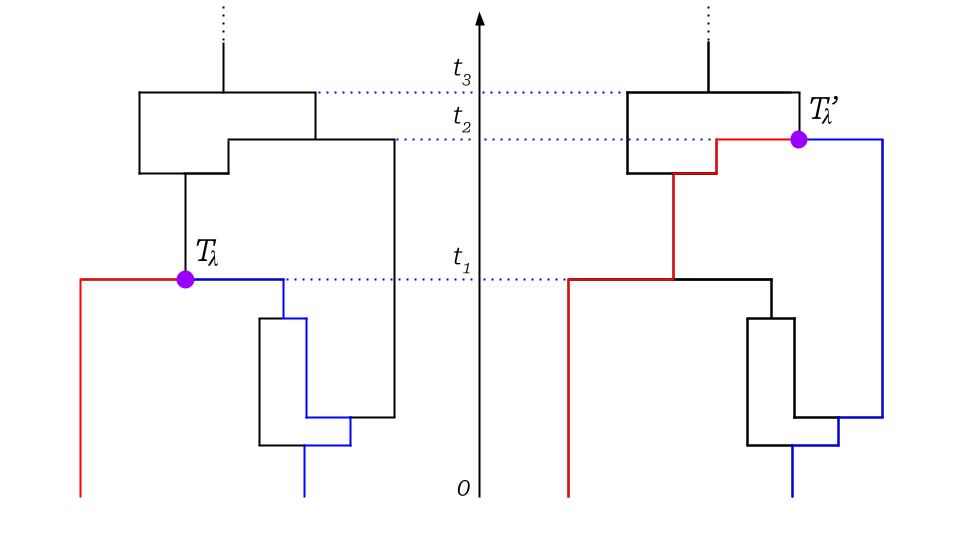}
            \caption{\small Two conditionally independent realizations of the pairwise process, $(x_\lambda, y_\lambda)$ on the same ancestral graph $G_{\lambda}$. The hitting times of these random walks are marked by a purple dot. These times correspond to $T_\lambda$ and $T_\lambda'$, the coalescence times of two unlinked loci given the same pedigree. In each realization, $x_\lambda$ is in blue and $y_\lambda$ in red. At each splitting event in the graph, where a particle splits into two, each walk chooses which particle path to follow fairly, i.e.\ with probability $1/2$ each. With respect to $\{t_1,t_2,t_3\}$ shown in Figure \ref{fig: pedigree_limit}, here $T_\lambda=t_1$ on the left and $T_\lambda'=t_2$ on the right. } 
			\label{fig: G_lambda_examples}
		\end{figure}
\FloatBarrier

The random variables $T_\lambda$ and $T_\lambda'$ in Lemma~\ref{L:two_equalities} are the pairwise meeting times for two conditionally independent processes $(x_\lambda,y_\lambda)$ and $(x_\lambda',y_\lambda')$ with the same starting point, i.e.\ with  $(x_\lambda(0),\,y_\lambda(0))=(x_\lambda'(0),\,y_\lambda'(0))$.  They correspond to the ancestral processes for $n=2$ at two unlinked loci starting from the same two sampled individuals, as in Figure~\ref{fig: G_lambda_examples}. To model the full two-pairs process, we consider all possible possible ancestral states of the samples.  For simplicity, we combine states where possible by symmetry and restrict ourselves to a process with the following five states.  
\begin{linenomath*}
  \begin{align*}
    s_0 & = 
      \begin{pmatrix}
        {\color{blue} \bullet }
      \end{pmatrix}
      \begin{pmatrix}
        {\color{blue} \bullet }
      \end{pmatrix}
      \begin{pmatrix}
        {\color{red} \bullet }
      \end{pmatrix}
      \begin{pmatrix}
        {\color{red} \bullet }
      \end{pmatrix} \\[4pt]
    s_1 & = 
      \begin{pmatrix}
        {\color{blue} \bullet }
      \end{pmatrix}
      \begin{pmatrix}
        {\color{red} \bullet }
      \end{pmatrix}
      \begin{pmatrix}
        {\color{blue} \bullet }\;{\color{red} \bullet }
      \end{pmatrix} \\[4pt]
    s_2 & = 
      \begin{pmatrix}
        {\color{blue} \bullet }\;{\color{red} \bullet }
      \end{pmatrix}
      \begin{pmatrix}
        {\color{blue} \bullet }\;{\color{red} \bullet }
      \end{pmatrix} \\[4pt]
    s_{\Delta,1} & = 
      \begin{pmatrix}
        {\color{blue} \bullet }
      \end{pmatrix}
      \begin{pmatrix}
        {\color{red} \bullet }
      \end{pmatrix}
      \begin{pmatrix}
        {\color{red} \bullet }
      \end{pmatrix} 
    \,\textrm{or}\,
      \begin{pmatrix}
        {\color{blue} \bullet }
      \end{pmatrix}
      \begin{pmatrix}
        {\color{blue} \bullet }
      \end{pmatrix}
      \begin{pmatrix}
        {\color{red} \bullet }
      \end{pmatrix} 
    \,\textrm{or}\,
      \begin{pmatrix}
        {\color{blue} \bullet }\;{\color{red} \bullet }
      \end{pmatrix}
      \begin{pmatrix}
        {\color{red} \bullet }
      \end{pmatrix} 
    \,\textrm{or}\,
      \begin{pmatrix}
        {\color{blue} \bullet }
      \end{pmatrix}
      \begin{pmatrix}
        {\color{blue} \bullet }\;{\color{red} \bullet }
      \end{pmatrix} \\[4pt]
    s_{\Delta,2} & = 
      \begin{pmatrix}
        {\color{blue} \bullet }
      \end{pmatrix}
      \begin{pmatrix}
        {\color{red} \bullet }
      \end{pmatrix} 
    \,\textrm{or}\,
      \begin{pmatrix}
        {\color{blue} \bullet }\;{\color{red} \bullet }
      \end{pmatrix}   
  \end{align*}
\end{linenomath*}
Here, two dots of the same color (blue or red) correspond to two particles of the same pair; say, the two blue dots correspond to  $(x_\lambda,y_\lambda)$ and the two red dots correspond to $(x_\lambda',y_\lambda')$. Two dots in parentheses correspond to taking the same value in $G_\lambda$, i.e.\ being in the same individual.  Hence, for $i\in\{0,1,2\}$, the state $s_i$ corresponds to when $i$ individuals or particles in $G_\lambda$ contain two conditionally independent lineages.

States $s_{0}$, $s_{1}$, and $s_{2}$ represent possible sample configurations of the two pairs among individuals.  Both members of both pairs are distinct in states $s_{0}$, $s_{1}$, and $s_{2}$, and we may note that $\left({\color{blue}\bullet}\,{\color{blue}\bullet}\right)\to\left({\color{blue}\bullet}\right)$ and $\left({\color{red}\bullet}\,{\color{red}\bullet}\right)\to\left({\color{red}\bullet}\right)$ instantaneously under limited outcrossing. State $s_{2}$ is the one for which $T_\lambda$ and $T_\lambda'$ are defined in Lemma~\ref{L:two_equalities}.  The composite states $s_{\Delta,1}$ and $s_{\Delta,2}$ are when one or both pairs have coalesced, respectively.  The previously defined time $T_\lambda \wedge T_\lambda'$, which is when at least one of the two pairs coalesces, is the time to enter state $s_{\Delta,1}$ or $s_{\Delta,2}$ starting from state $s_2$.  We take state $s_{\Delta,2}$ to be the absorbing state of this five-state process.  

\begin{lemma}\label{L:S_5states}
    Let $S = (S_t)_{t\in\R_+}$ be the process, with state space $\{s_0,\,s_1,\,s_2,\,s_{\Delta,1},\,s_{\Delta,2}\}$, which tracks the two pairs $(x_\lambda,y_\lambda)$ and $(x_\lambda',y_\lambda')$ at time $t$ (backward) on the coalescence time-scale, then its rate matrix $Q_{\lambda}$ is
\begin{linenomath*}
    \begin{equation}\label{E:Q}
    Q_{\lambda}=
    \begin{pmatrix}
        -12 & 8 & 0 & 4 & 0\\
        \frac{\lambda}{2} & -6-\frac{\lambda}{2} & 2 & 4 & 0\\
        0 & \lambda & -2-\lambda & 0 & 2\\
        0 & 0 & 0 & -2 & 2\\
        0 & 0 & 0 & 0 & 0
    \end{pmatrix}.
    \end{equation}
\end{linenomath*}
\end{lemma}

We omit the proof of Lemma \ref{L:S_5states} because it is contained in previous works, for example \citet{SimonsenAndChurchill1997}.
Observe that $Q_\lambda$ is equivalent to the corresponding matrices for the Kingman coalescent with recombination between two loci \citep{KaplanAndHudson1985,PluzhnikovAndDonnelly1996,SimonsenAndChurchill1997}; see also \citet[Chapter 7]{Wakeley2009} and \citet[Chapter 3]{durrett2008probability}. More precisely, our model here becomes that in \citet[Figure 7.7]{Wakeley2009} if we divide all our rates by two, and then make the substitution $\lambda/2 \to \rho$.  This $1/2$ on $\lambda$ comes from the fact that each splitting event is only observed, meaning that $\left({\color{blue}\bullet}\,{\color{red}\bullet}\right)\to\left({\color{blue}\bullet}\right)\left({\color{red}\bullet}\right)$, with probability $1/2$. 

Then, the covariances of $T_\lambda$ and $ T_\lambda'$ for the different sample configurations are  exactly the same as the classical formulas in two-locus models with recombination, given for example by equations (7.30), (7.29) and (7.28) in \cite{Wakeley2009}. Thus we have 
\begin{linenomath*}
\begin{align}  
\mathrm{Cov}_{s_0}{\left(T_\lambda,T_\lambda'\right)} &= \frac{4}{\lambda^2+26\lambda+72} , \label{E:CovT1T2s0} \\
\mathrm{Cov}_{s_1}{\left(T_\lambda,T_\lambda'\right)} &= \frac{6}{\lambda^2+26\lambda+72}, \label{E:CovT1T2s1} \\
\mathrm{Cov}_{s_2}{\left(T_\lambda,T_\lambda'\right)} &= \frac{1}{2}\frac{\lambda+36}{\lambda^2+26\lambda+72} \label{E:CovT1T2s2} .
\end{align}
\end{linenomath*}

\medskip
\begin{proposition}\label{L:variance_of_conditional_coal_time}
For any $\lambda \in \R_+$, the variance of $\mathbb{E}{\left[T_{\lambda}|G_{\lambda}\right]}$ is equal to 
\begin{linenomath*}
\begin{equation}\label{E:variance_conditional_time}  
\mathrm{Var}{\left(\mathbb{E}{\left[T_{\lambda}|G_{\lambda}\right]}\right)}=\mathrm{Cov}{\left(T_{\lambda},T'_{\lambda}\right)} = \frac{1}{2}\frac{\lambda+36}{\lambda^2+26\lambda+72} .
\end{equation}
\end{linenomath*}
In particular,
\begin{linenomath*}
\begin{equation*}
\mathrm{Var}{\left(\mathbb{E}{\left[T_{\lambda}|G_{\lambda}\right]}\right)} =
\begin{dcases}
\frac{1}{4} - \frac{1}{12}\lambda + O(\lambda^{2}) & \text{ as } \lambda \to 0 \\[4pt]
\frac{1}{2\lambda} + \frac{5}{\lambda^2} + O(\lambda^{-3}) & \text{ as } \lambda\to\infty
\end{dcases} . 
\end{equation*}
\end{linenomath*}
\end{proposition}

\begin{proof}
The first equality in \eqref{E:variance_conditional_time} follows directly from \eqref{E:Var_Cov}, the second from \eqref{E:CovT1T2s2}.
\end{proof} 
%Similar to \eqref{E:lambda to 0}, we can specify the asymptotic behavior of \eqref{E:variance_conditional_time} as 

%In the Discussion, we will consider the significance of  \eqref{E:variance_conditional_time} for an understanding how coalescence times depend on the pedigree, represented in this case by $G_\lambda$.  

%This representation will lead to the following asymptotic formula.
%The matrix exponential $e^{t R_\lambda}$ does not appear to admit a simple representation. However, we can use Taylor expansion around $\lambda = 0$ to obtain the first part the following proposition. The second part (as $\lambda\to\infty$) will be established by a different method. 

Next, we consider asymptotics of the variance of $\mathbb{P}(T_\lambda > t | G_\lambda)$.

\begin{proposition}\label{prop:survival_probability_estimates}
        For any $t \in (0,\infty)$,
\begin{linenomath*}
\begin{equation}\label{E:lambda to 0}
            \mathrm{Var}\left(\mathbb{P}(T_\lambda > t | G_\lambda)\right)
            = 
            \begin{dcases}
            e^{-2t}-e^{-4t} + \lambda
            \frac{e^{-2t}}{8}\left(1-4t-e^{-4t}\right) + O(\lambda ^ 2) &\quad \text{ as }\lambda\to 0\\
            \frac{2e^{-4t}}{\lambda}\left(1 + \frac{2(3 + 16t)}{\lambda}\right) + O(\lambda^{-3}) &\quad \text{ as }\lambda\to \infty
            \end{dcases}.
\end{equation} 
\end{linenomath*}
\end{proposition}
This follows from the computation of the second moment of $\mathbb{P}(T_\lambda > t | G_\lambda)$ in Lemma~\ref{L:survival_probability} in the Appendix.  %\MN{It should be noted that proposition 4.6 is a corollary of corollary 4.8.}
%Note that as $\lambda\to 0$, the right of \eqref{E:lambda to 0} converges to $1 - e ^ {-2t}$ which agrees with the second moment of the conditional CDF of the coalescence time in the super-critical case. Hence 
Then we have,
%the variance of $\mathbb{P}(T_\lambda > t | G_\lambda)$ converges to that of $\textbf{1}_{\{{\rm Exp}(2)>t\}}$  (negligible-outcrossing regime) as $\lambda\to 0$, and to 0  (partial outcrossing regime) as $\lambda\to \infty$. Therefore, 
for any $t\in (0,\infty)$,
\begin{linenomath*}
    \begin{equation}\label{E:extreme_lambda}
            \mathbb{P}(T_\lambda > t | G_{\lambda})\to    \begin{dcases}
            \textbf{1}_{\{{\rm Exp}(2)>t\}} &\quad \text{ as }\lambda\to 0\\
            e^{-2t}&\quad \text{ as }\lambda\to \infty
            \end{dcases}.
    \end{equation}
\end{linenomath*}

The proofs of both Lemma~\ref{L:survival_probability} and 
Proposition~\ref{prop:survival_probability_estimates} are given in the Appendix, Section~\ref{S:further_characterization}.
%Corollary~\ref{E:extreme_lambda} follows directly from Proposition~\ref{prop:survival_probability_estimates} because the latter implies that $\mathbb{P}(T_\lambda > t | G_{\lambda})$ converges to $e^{-2t}$ in $L^2(\mathbb{P})$ under the probability measure for $G_{\lambda}$.

\medskip

We observe that the  covariance $\mathrm{Cov}{\left(T_\lambda,T_\lambda'\right)}$ is connected to the difference between the conditional survival function $\mathbb{P}(T_\lambda > t | G_\lambda)$ and the unconditional survival function $\mathbb{P}(T_\lambda > t)=e^{-2t}$. %It also quantifies the expected squared $L^2$ distance between the survival function of the conditional coalescence time and that of a rate $2$ exponential random variable. 
%We can also obtain the convergence in the distances $L^2$.
\begin{proposition}\label{prop:L^2coincidence}
        For any $\lambda \in (0,\infty)$,
\begin{linenomath*}
\begin{equation}\label{E:L^2coincidence}
        \mathbb{E}\left[\int_0^\infty \left(\mathbb{P}(T_\lambda > t | G_\lambda) - e ^ {-2t}\right)^2 dt\right]
        = \mathrm{Cov}{\left(T_\lambda,T_\lambda'\right)}.%=\frac{1}{2}\frac{\lambda+36}{\lambda^2+26\lambda+72}
\end{equation} 
\end{linenomath*}
\end{proposition}

\begin{proof}
By  Lemma~\ref{L:two_equalities}, the left-hand side is equal to
\begin{equation*}
 \int_0^\infty (\mathbb{P}_{s_2}(T_\lambda \wedge T_\lambda' > t) - e ^ {-4t})dt= \mathbb{E}_{s_2}[T_\lambda \wedge T_\lambda']-\frac{1}{4}=\frac{1}{2}\frac{\lambda+36}{\lambda^2+26\lambda+72},
    \end{equation*}
where by $T_\lambda \wedge T_\lambda'$ we mean the smaller of $T_\lambda$ and $T_\lambda'$.
\end{proof}

\begin{remark}\rm\label{Rk:Coincidences}
In agreement with previous work---see (3) and (52) in \citet{SimonsenAndChurchill1997}---from \eqref{E:CovT1T2s2} and the fact $\mathrm{Var}{\left(T_{\lambda}\right)}=1/4$, we see that  
$\mathrm{Corr}{\left(T_\lambda,T_\lambda'\right)}$ is exactly the same as
$\mathbb{P}(T_\lambda =  T_\lambda')$. Equivalently, $\mathrm{Cov}{\left(T_\lambda,T_\lambda'\right)}=\frac{1}{4}\mathbb{P}(T_\lambda =  T_\lambda')$, and it is also equal to the left hand side of \eqref{E:L^2coincidence}. The covariance formula obtained here is consistent with \citet[Eq.~9]{wilton2015smc} by taking $2T_\lambda = T_1$ and $2T_\lambda' = T_2$; the rescaling is due to the fact that the marginals on $T_\lambda$ and $T_\lambda'$ are exponentially distributed with rate $2$, while those of $T_1$ and $T_2$ are with rate $1$.
\end{remark}

Depending on $\lambda\in (0,\infty)$, limited outcrossing displays a range of behaviors between the two qualitatively different extremes of negligible outcrossing and partial selfing.  Figure~\ref{fig:VarTcomponents} illustrates this using the law of total variance for the coalescence time of two gene copies sampled from two different individuals.  Let $T$ and $G$ represent this ``diff'' coalescence time and its associated ancestral graph, for simplicity dropping the subscript $\lambda$ we attached to these previously.  Then we can write 
\begin{linenomath*}
\begin{equation}
\textrm{Var}{\left(T\right)} = \mathbb{E}{\left[\textrm{Var}{\left(T|G\right)}\right]} + \textrm{Var}{\left(\mathbb{E}{\left[T|G\right]}\right)} . \label{eq:totalVarT}
\end{equation}
\end{linenomath*}
Note that $\textrm{Var}{\left(T\right)}=1/4$, because $T$ by itself is exponentially distributed with rate parameter $2/(2-1) = 2$.  The second term on the right-hand side of \eqref{eq:totalVarT}, $\textrm{Var}{\left(\mathbb{E}{\left[T|G\right]}\right)}$, is given by \eqref{E:variance_conditional_time}.  The first term  on the right-hand side of \eqref{eq:totalVarT} is just $\mathbb{E}{\left[\textrm{Var}{\left(T|G\right)}\right]}=1/4-\textrm{Var}{\left(\mathbb{E}{\left[T|G\right]}\right)}$, but both terms are displayed in Figure~\ref{fig:VarTcomponents} for the sake of illustration.

\begin{figure}[ht]
\centering 
\includegraphics[scale=1]{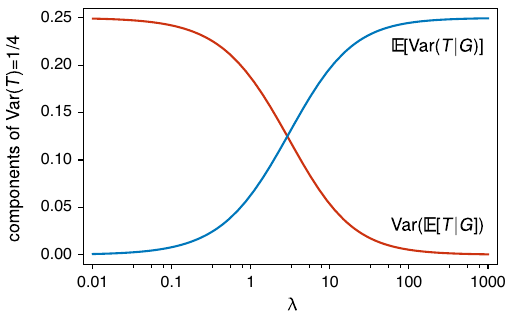}$\quad$
\caption{\small The two components, $\mathbb{E}{\left[\textrm{Var}{\left(T|G\right)}\right]}$ and $\textrm{Var}{\left(\mathbb{E}{\left[T|G\right]}\right)}$, of $\textrm{Var}{\left(T\right)}$ as functions of $\lambda$ under limited outcrossing. By the law of total variance, $\textrm{Var}{\left(T\right)}=\mathbb{E}{\left[\textrm{Var}{\left(T|G\right)}\right]}+\textrm{Var}{\left(\mathbb{E}{\left[T|G\right]}\right)}$.  Also, $\textrm{Var}{\left(T\right)}=1/4$ in this case because it is the average over ancestral graphs ($G$), that is to say over pedigrees, and because $\alpha_N \to 1$.}\label{fig:VarTcomponents}
\end{figure}

When $\lambda$ is very small, on the far left in Figure~\ref{fig:VarTcomponents}, variation in the mean coalescence time among graphs, $\textrm{Var}{\left(\mathbb{E}{\left[T|G\right]}\right)}$, accounts for the bulk of the total variation in $T$.  Correspondingly, there is little variation in coalescence times given the graph.  Figure~\ref{fig:VarTcomponents} displays the expectation of the latter, $\mathbb{E}{\left[\textrm{Var}{\left(T|G\right)}\right]}$, but $\textrm{Var}{\left(T|G\right)}$ itself will be small on any graph in this case.  With probability approaching one as $\lambda \to 0$, the individuals containing the two ancestral lineages will experience only selfing events in their ancestries, up to and including the time they descend from a common ancestral individual.  Pairwise coalescence times at every locus in the genome starting from these two individuals will converge on the time of this ancestor and $\textrm{Var}{\left(T|G\right)}$ should be minimal.  Variation in the waiting time to this first common ancestral individual among graphs will be the primary source of variation in $T$.  In the limit $\lambda\to 0$, this waiting time is exponentially distributed with rate parameter $2$, so $\textrm{Var}{\left(\mathbb{E}{\left[T|G\right]}\right)}$ approaches $\textrm{Var}{\left(T\right)}=1/4$.  It is given by a Kingman coalescent process made effectively haploid by selfing, just as under negligible outcrossing. 

Across the middle of Figure~\ref{fig:VarTcomponents}, ancestral graphs will include increasing numbers of splitting or outcrossing events as $\lambda$ increases.  Such graphs contain multiple pathways to coalescence because genetic lineages may trace back to either parent at each outcrossing event.  They include multiple possible coalescence times, as already evident in the CDFs of $T|G$ in Figure~\ref{fig: pedigree_limit} for a hypothetical graph with three splitting events and three possible coalescence times, in Figure~\ref{fig:3_regimes_3_figures} for each of five simulated ancestral graphs with $\lambda=2$, and in Figure~\ref{fig:1} for each of fifty pedigrees in three versions of a finite-$N$ Wright-Fisher model corresponding to $\lambda\in\{1,10,100\}$.  As a result, $\textrm{Var}{\left(T|G\right)}$ is positive and $\mathbb{E}{\left[\textrm{Var}{\left(T|G\right)}\right]}$ becomes the main source of variation in $T$ as $\lambda$ increases, eventually swamping variation of the mean coalescence times among graphs.   

When $\lambda$ is very large, on the far right in Figure~\ref{fig:VarTcomponents}, single ancestral graphs do little to constrain $T|G$.  These graphs are large.  In their construction backward in time, the number of ancestral lineages grows quickly to reach a quasi-stable distribution, specifically a zero-truncated Poisson distribution with mean $\lambda$, as \citet{Mano2009} showed for the ancestral recombination graph.  The waiting time to the event that completes the graph, when there is just one lineage left, will be long, $\sim 2 e^\lambda / \lambda^2$ in expectation using equation (1.1) in \citet{GriffithsAndMarjoram1997}.  Given the graph, each of the pair of genetic lineages starting from the two sampled individuals will trace back through a large number of splitting events before the two meet in a single individual and coalesce.  As $\lambda\to\infty$, all variation in $T$ will be due to this random process of coalescence given the graph.  Equation~\eqref{E:extreme_lambda} shows that this conditional coalescent process converges on that of the partial-selfing model with $\alpha_N \to 1$.  Proposition~\ref{prop:L^2coincidence} shows further that $\textrm{Var}{\left(\mathbb{E}{\left[T|G\right]}\right)}$ in \eqref{E:variance_conditional_time} is a key measure of the distance of the CDF under limited outcrossing to that under the corresponding model of partial selfing. 

To summarize what is shown in Figure~\ref{fig:VarTcomponents}, limited outcrossing becomes like negligible outcrossing as $\lambda\to 0$ and like partial selfing as $\lambda\to\infty$.  With respect to the pre-limiting Moran model, the results and intuitions about Figure~\ref{fig:VarTcomponents} concern the behavior at the boundary of complete or nearly complete selfing ($\alpha_N \to 1$).  All at this boundary, we find: first negligible outcrossing, then limited outcrossing, then an extreme case of partial selfing with $\alpha=1$.  The entire rest of the range of selfing probabilities $\alpha\in[0,1]$ is covered by the partial-selfing model.

Finally, although a treatment of recombination is beyond the scope of this work, it is relatively straightforward to describe the joint ancestral process of a sample of size two at two linked loci under limited outcrossing. That is, we can extend the analysis giving the rate matrix \eqref{E:Q} and the covariances \eqref{E:CovT1T2s0}-\eqref{E:CovT1T2s2} to the case of two linked loci. 

Consider two pairs of ancestral lineages at two linked loci with per-generation recombination probability $r \in [0,1/2]$ between them. The substate $\left({\color{blue}\bullet}\,{\color{red}\bullet}\right)$ present in some states of the process with rate matrix $Q_\lambda$ in \eqref{E:Q} actually includes two cases. These ancestral gene copies are either on the same chromosome or on different chromosomes in the individual.  Due to the timescale, in which each lineage of the graph $G_\lambda$ comprises an effectively infinite number of generations, the process determining the probabilities of these two states will always be at stationarity. %The value of this equilibrium determines the probability the two gene copies will separate when a splitting event occurs.  

An equilibrium is established by the balance of two events in a purely selfing individual lineage: gene copies on different chromosomes in an offspring will be on the same chromosome in the parent with probability $1/2$, and gene copies on the same chromosome in an offspring will be on different chromosomes in the parent with probability $r$.  The equilibrium probability that they are on different chromosomes is
\begin{linenomath*}
\begin{equation} 
p_r \coloneq \frac{2r}{1+2r}
\label{eq:pr}
\end{equation} 
\end{linenomath*}
and it is only in this case that the two will separate, one going to each parent, when a splitting event occurs. Thus, the effective rate of splitting is $\lambda p_r$.  Note that $r=1/2$ for unlinked loci and this gives $p_r=1/2$, as already noted regarding \eqref{E:Q}.  

In sum, the two locus ancestral graph process with recombination can be modeled with the same states as above and with rate matrix \eqref{E:Q} but with $\lambda p_r$ in place of $\lambda/2$.  Similarly, the covariances of coalescence times at two linked loci are given by \eqref{E:CovT1T2s0}-\eqref{E:CovT1T2s2} after making this same substitution.  We emphasize, cf.\ Lemma~\ref{L:two_equalities} and Figure~\ref{fig:VarTcomponents}, that these covariances are averages over the ancestral graph or the pedigree.  \citet{Kogan2023.10.18.563014} recently derived such covariances for arbitrary $N$, $r$ and $\alpha_N$ in a model of Wright-Fisher reproduction.  Their detailed exact expressions collapse to ours when the same limit is taken with $N(1-\alpha_N)\to\lambda$ and accounting for the difference in time scale between the Wright-Fisher and Moran models (results not shown).

\section{Discussion}\label{sec:discussion}

In this work, we have assumed that the pedigree of a population of constant size $N$ is the outcome of a random process of reproduction in which offspring are produced by self-fertilization with probability $\alpha_N$.  By conditioning on the pedigree and considering pairwise times to common ancestry, we found three different coalescent models in the limit $N\to\infty$, when time is measured in units of $N$ generations (which is $N^2$ time-steps)  in the Moran model.  The deciding factor is $N(1-\alpha_N)$, which may be interpreted either as the rate of outcrossing events along an ancestral lineage or as the expected number of outcrossed offspring in the population each generation.  The three models are `negligible outcrossing' which applies when $N(1-\alpha_N)\to0$, `limited outcrossing' which applies when $N(1-\alpha_N)\to\lambda\in(0,\infty)$, and `partial selfing' which applies when $N(1-\alpha_N)\to\infty$.  Negligible outcrossing and limited outcrossing both require $\alpha_N \to 1$, while partial selfing includes this as a special case.

Previous population-genetic analyses of selfing populations have assumed a fixed selfing probability, $\alpha_N=\alpha$ in our notation, and have implicitly averaged over pedigrees.  When this averaging is done, just one coalescent model for all $\alpha\in[0,1]$ emerges in the limit, specifically the coalescent model with partial selfing described by \citet{NordborgAndDonnelly1997} and \citet{Mohle1998a}.  When coalescence is conditioned on the pedigree and $\alpha$ is fixed, our very similar partial-selfing model is obtained but only for $\alpha\in[0,1)$ when $\alpha$ is fixed.  Assuming a fixed selfing probability of $\alpha=1$ leads to an entirely different model, namely negligible outcrossing.  

A more detailed look at how $\alpha_N$ approaches $1$ reveals the critical case of limited outcrossing in between negligible outcrossing and partial selfing, and characterized by the parameter $N(1-\alpha_N)\to\lambda\in(0,\infty)$.  In this regime, the parts of the population pedigree directly ancestral to the sampled individuals can be replaced by a random graph like the ones previously used to model recombination \citep{Griffiths1991,GriffithsAndMarjoram1997} and selection \citep{kroneneuhauser1997,NeuhauserAndKrone1997}.  It is distinguished from these only by what the splitting events in the graph represent and how these are treated in modeling gene genealogies.  Each splitting event in the graph corresponds to an outcrossing event in the genealogical ancestry of the sampled individuals.  Ancestral lineages tracing back to such a node in the graph make a $50$:$50$ choice of which branch they follow.  The rate at which ancestral lineages split is precisely our parameter $\lambda$. 

As with the well known Wright-Fisher diffusion and the corresponding standard neutral coalescent model, these new models of coalescence conditional on the pedigree are meant as robust approximations for large populations.  We expect, for example, that the Wright-Fisher model with selfing used to produce Figure~\ref{fig:1} will have the same three limiting cases we found for the Moran model with selfing.  Which of the three is most applicable for a given (necessarily finite) population will depend on $N$ and $\alpha_N$.  The special case of negligible outcrossing involves the strong predictions of identical gene genealogies across all loci in the genome and zero heterozygosity, and so may be of limited utility.  The critical case of limited outcrossing allows gene genealogies to differ across the genome but also predicts zero heterozygosity.  Still it is clear from the jumps in probability in Figure~\ref{fig:1} that if $N(1-\alpha_N)$ is not large, insights from the limited-outcrossing model are helpful even when the value of $\alpha_N$ would not prevent heterozygosity.  On the other hand, both Figure~\ref{fig:1} and Figure~\ref{fig:VarTcomponents} suggest that if $N(1-\alpha_N)$ is greater than about $100$, partial-selfing is a better description than limited-outcrossing, even if the outcrossing probability $1-\alpha_N$ is quite small.  

We chose to use the name ``partial selfing'' here because in this case the coalescent process conditional on the pedigree is very similar to the previously described partial-selfing model which did not condition on the pedigree \citep{NordborgAndDonnelly1997,Mohle1998a}.  Our model may be considered an update of the previous one.  In both, times to common ancestry for ``diff'' samples follow the Kingman coalescent process with effective population size $N_e = (2-s) N / 2$.  In our formulation, partial selfing is obtained when $N(1-\alpha_N)\to\infty$.  Similarly we may recall from Figure~\ref{fig:VarTcomponents} that as $\lambda\to\infty$ the predictions for ``diff'' samples under limited outcrossing converge on those for partial selfing.  The reason that pedigrees do not constrain ``diff'' coalescence times under the partial-selfing model is that outcrossing events dominate the ancestry of the population.  The same underlying phenomena are at work here as in the standard models of randomly mating populations, where the fact that individuals have two parents with high probability strongly influences the pedigree \citep{Chang1999,DerridaEtAl1999,coron2022pedigree} resulting in a fast mixing time of ancestral genetic processes given the pedigree \citep{BartonAndEtheridge2011}.  Because of this, the conditional and unconditional coalescent processes for these standard population models converge to the same Kingman coalescent process \citep{TyukinThesis2015,DFBW24}.

The key update in our partial-selfing model is the inclusion of the number of generations of selfing up to the first outcrossing event in the recent ancestry of sampled individuals.  In our model, a sampled individual has a single realization of the random variable $U$ in Theorem~\ref{T:MAIN_conditional_same} or in \eqref{eq:PUk}.  In the previous model, an average is taken over this distribution then this average, specifically the inbreeding coefficent $F=s/(2-s)$, is applied to every individual.  Conditioning on the pedigree produces a coalescent model which aligns better with other work on partial selfing in population and evolutionary genetics.  This is particularly true for multi-locus data, which are an explicit concern of conditional coalescent processes because all loci are transmitted through the same pedigree. 

Let us take for granted the conjecture we made in Section~\ref{S:condition_n} concerning a sample of size $n$ with $2m$ gene copies together as pairs in $m$ individuals.  At a locus $i$ in such a sample, there will be a random number of effectively instantaneous coalescent events $X_i$, then the remaining $n-X_i$ lineages will enter the Kingman coalescent process.  If $\{k_1,k_2,\ldots,k_m\}$ are the realized numbers of selfing generations in the ancestries of the $m$ individuals, then $X_i\sim\sum_{r=1}^{m} \textrm{Bernoulli}\left({1-2^{-k_r}}\right)$.  An unlinked locus, $j \neq i$, has the same $\{k_1,k_2,\ldots,k_m\}$ but a conditionally independent value of $X_j$ and an independent realization of the Kingman coalescent process.

This coalescent model with partial selfing conditional on the pedigree accurately predicts the variation in the level of inbreeding among individuals which  \citet{WeirAndCockerham1973} developed into a theory of identity disequilibrium; recall \eqref{eq:sigma2F} and \eqref{eq:covxixj}.  Like all coalescent models, the focus is now on the sample rather than the population.  If we write $Y_{i,r}$ for the indicator random variable of coalescence in the recent ancestry of sampled individual $r\in\{1,2,\ldots,m\}$, then $E[Y_{i,r}]=s/(2-s)$ as in \eqref{eq:Fdef} and $\textrm{Var}[Y_{i,r}]$ is given by~\eqref{eq:sigma2F} or~\eqref{eq:covxixj}.  If we were able to observe the outcomes $y_{i,r}$ for $r\in\{1,2,\ldots,m\}$, then the mean among loci $\bar{y}_{r}$ would be an estimate of $1-2^{-k_r}$ and the sample variance of $\bar{y}_{r}$ among individuals would be an estimate of the identity disequilibrium \eqref{eq:sigma2F}.

Identity by descent cannot be observed directly but only inferred based on patterns of identity in state \citep{Pollak1987,Rousset2002,Uyenoyama2024}.  A number of methods have been developed and applied to estimate inbreeding coefficients or selfing probabilities.  Some have made use of $U$ in \eqref{eq:PUk} or similar distributions to estimate selfing probabilities of populations  \citep{EnjalbertAndDavid2000,DavidEtAl2007,McClureAndWhitlock2012}.  Others have employed identity disequilibria \citep{RitlandAndJain1981,Ritland1996,SweigartEtAl1999,Yang2002,JarneAndDavid2008,SzulkinEtAl2010,VieiraEtAl2013} or runs of homozygosity \citep{Franklin1977,YiEtAl2022,ZeitlerAndGilbert2024} to estimate individual inbreeding coefficients due to selfing. Others have taken up the problem of estimating the number of generations back to the most recent outcrossing event for each individual together with the selfing probability \citep{GaoEtAl2007,WangEtAl2012,RedelingsEtAl2015}. Our results for partial selfing establish a coalescent framework for these methods, and provide for new coalescent-based estimators like the one described in \citet[Eq.~20]{NordborgAndDonnelly1997} but incorporating \eqref{eq:PUk} and $\{k_1,k_2,\ldots,k_m\}$ above in place of just $F=s/(2-s)$. 

Of these empirical works, the ones which appeal to the previous, unconditional coalescent model with partial selfing for justification do so with respect to the properties of ``diff'' samples, i.e.\ that they follow a Kingman coalescent with diploid effective size $N_e=(2-s)N/2$.  Again, this is also a property of the conditional model we have described.  But the coalescent model with partial selfing conditional on the pedigree further includes rigorous justification for the modeling of individual ancestries of selfing already done in these works.  For example, in the inference method for infinite-alleles data developed by \citet{RedelingsEtAl2015}, individual ancestries are precisely those of the conditional model.  Then, cf.\ Figure 1 in \citet{RedelingsEtAl2015}, the allelic states of the distinct ancestral lineages which emerge from these random outcomes at a locus, in our notation the $n-X_i$ lineages described above, are given by the Ewens sampling formula \citep{Ewens1972} with the mutation parameter for a Kingman coalescent with diploid effective size $N_e=(2-s)N/2$.

The existence of the limited-outcrossing coalescent model, with its simple structure of the random graph, can aid in the interpretation of multi-locus genetic data even if its extreme prediction of zero heterozygosity of individuals is not met.  Each branching event in the ancestral graph corresponds to an outcrossing event between ancestral individuals in the pedigree of the sample, in the case that these are rare, i.e.\ with probability proportional to $1/N$ in the model.  Forward in time and to the extent that the two ancestral individuals differ genetically, each such event would create a potentially novel recombinant inbred line.  The partial-selfing model approximates this when the selfing probability is high, except that some of these new lines may not become fully inbred.  

Samples from predominately selfing populations are often dominated by individuals homozygous for one of only a few distinct mutli-locus genotypes, together with some recombinant inbred or partially inbred individuals \citep{BonninEtAl2001,BakkerEtAl2006,BombliesEtAl2012,HartfieldEtAl2017,JullienEtAl2019}.  The \textit{Medicago truncatula} data for $24$ unlinked loci in a sample of $200$ individuals presented and analyzed by \citet{SiolEtAl2008} illustrates this.  They found four common multi-locus genotypes present in $76$, $34$, $25$ and $17$ individuals in the sample.  A further $21$ recombinant inbred genotypes, resulting from relatively ancient crosses between the common types, were distributed among $39$ individuals in the sample.  At lowest frequency in the sample were $9$ recombinant partially inbred (heterozygous) genotypes present in single individuals, presumably with relatively recent outcrossing events in their ancestries. 

We have analyzed a simple population model in some detail to understand how the population pedigree shapes the coalescent process for all loci in a genome.  The model includes just reproduction and selfing.  There is no selection, no mutation, no recombination, no population structure, and no change in the population size over time.  We expect that neutral mutation and changes in population size could be included the same way they are included in coalescent models which do not condition on the pedigree.  So, for example, applications of coalescent theory to ``diff'' samples under partial selfing which infer the ancestry of populations, such as those in \citet{StOngeEtAl2011}, \citet{BrandvainEtAl2013}, \citet{BeissingerEtAl2016}, are as well justified under our conditional partial-selfing model as they are under the unconditional model of \citet{NordborgAndDonnelly1997}.

It would not be as straightforward to incorporate selection or population structure because these can have their own effects on the pedigree.  Some simulations demonstrating this can be found in \citet{kuoavise2008}, \citet{WakeleyEtAl2016} and \citet{WiltonEtAl2017}.  It is an open question how recombination and coalescence operate conditional on the population pedigree.  We have dealt mostly with free recombination here. In particular, Theorem~\ref{T:MAIN_conditional} and Theorem~\ref{T:MAIN_conditional_same} are about the distribution of conditionally independent coalescence times given the pedigree.  In the case of limited outcrossing, we were able to include an analysis of two loci with recombination which gives covariances of coalescence times consistent with the discrete, exact results for the unconditional or pedigree-averaged case in \citet{Kogan2023.10.18.563014}.  

Finally, selfing is just one form of inbreeding.  The importance of modeling and disentangling the sources of inbreeding is well appreciated in the unconditional or pedigree-averaged case \citep{Uyenoyama1986,McClureAndWhitlock2012,JarneAndDavid2008,ZeitlerAndGilbert2024}.  It would be of interest to understand the pedigree effects of inbreeding generally in the conditional case.  \citet{Campbell2015} and \citet{SeversonEtAl2019} studied expected coalescence times for given rates of mating between cousins of varying degree, including sibs or ``$0$-degree cousins''.  \citet{SeversonEtAl2021} and \citet{CotterEtAl2021} extended these models and applied the separation-of-timescales result of \citet{Mohle1998a} to obtain the limiting distributions of coalescence times in the unconditional case.  Based on the results we presented here, we would expect to see some effects of the population pedigree in conditional coalescent processes for these and other forms of inbreeding.

\section*{Acknowledgements}
\noindent This work was supported by National Science Foundation grants DMS-2152103 and DMS-2348164.

\section*{Conflicts of interest}
\noindent The authors declare no conflicts of interest. 

\bibliography{arxiv_submission_bis}

%\newpage

\section{Appendix}\label{sec:appendix}

In this section, we prove Theorem~\ref{theorem: unconditional_pairwise_coalescence_time_convergence}, our main results Theorems~\ref{T:MAIN_conditional} and \ref{T:MAIN_conditional_same}, and Proposition \ref{prop:survival_probability_estimates}.

%\subsection{Notation}    

\subsection{Proofs for the unconditional distribution in Theorem~\ref{theorem: unconditional_pairwise_coalescence_time_convergence}}\label{A:unconditional}

We let $\hat{X}_k$ and $\hat{Y}_k$ be the labels of the individuals to which the genes $X_k$ and $Y_k$ belong, and
\begin{linenomath*}
\begin{equation}\label{Def:O}
\mathcal{O}:=|\{k\geq 1:\,\hat{X}_k= \hat{Y}_k,\,\hat{X}_{k-1}\neq \hat{Y}_{k-1}\}|  
\end{equation}
\end{linenomath*}
be the total  number of time-steps when the two sample lineages transition from belonging to two distinct individuals to belonging to a single individual.

Define the timings of the \textit{overlap events} by setting $\tau^O_0:=0$ and, for $i\geq 1$, 
\begin{linenomath*}
\begin{equation}\label{Def:tauO}
\tau^O_i:=\inf\{k>\tau^O_{i-1}:\,\hat{X}_k= \hat{Y}_k,\,\hat{X}_{k-1}\neq \hat{Y}_{k-1} \}.
\end{equation}
\end{linenomath*}
Next, we define the timings of \textit{splitting events}. We let $\tau^S_1:=\inf\{k\in\Z_+:\,\hat{X}_k\neq  \hat{Y}_k\}$, which is zero under $\Pdiff$ and is the time of the first splitting event under $\Psame$.  For $i\geq 2$, we let
\begin{linenomath*}
\begin{equation}\label{Def:tauD}
\tau^S_i:=\inf\{k>\tau^S_{i-1}:\,\hat{X}_k \neq \hat{Y}_k,\,\hat{X}_{k-1}= \hat{Y}_{k-1} \}
\end{equation}
\end{linenomath*}
be the time-step of the $i$-th splitting event in the past. By convention, $\inf\emptyset =\infty$. 
            
We shall decompose $\tau^{(N)}$ as in \eqref{E:tau_rep}; see Figure~\ref{fig: unconditional_pairwise_diagram}. %and describe the distribution of $\mathcal{O}$ in the following lemma. 
For simplicity, we write $X\sim {\textrm Geom}(r)$ for a geometric random variable with parameter $r$, that is, when $\PP(X=m)=(1-r)^{m-1}r$ for $m\in \Z_{>0}$.

    %\FloatBarrier  %%% for  the figure about $\mathcal{O}$
    \begin{figure}[ht]
        \centering
        \includegraphics[scale=0.4]{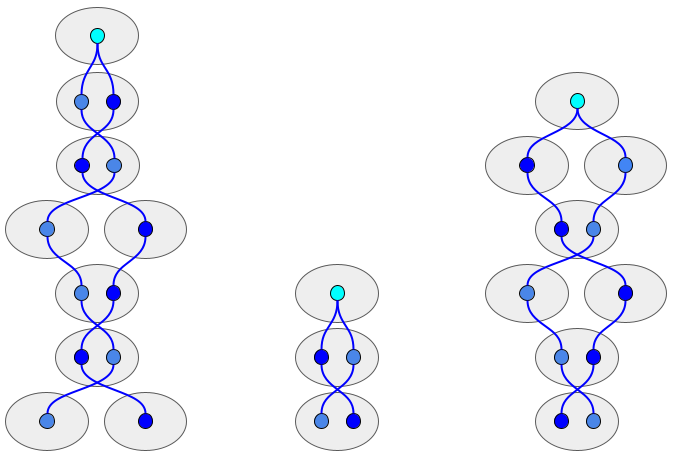}
        \caption{\small Three different realizations of $\tau^{(N)}$. The left realization is under $\Pdiff$ with $\mathcal{O} = 2$. The middle realization is under $\Psame$ with $\mathcal{O} = 0$. The right realization is under $\Psame$ with $\mathcal{O} = 2$. Note that the initial state makes no contribution to $\mathcal{O}$. In particular, under the $\Psame$ sample $\mathcal{O}$ may be zero.}
        \label{fig: unconditional_pairwise_diagram}
    \end{figure}
    %\FloatBarrier
    
    \begin{lemma}\label{lemma: unconditional_tau_N_2_decomposition}
        Under both $\Pdiff$ and $\Psame$, it holds that
\begin{linenomath*}
                    \begin{equation}\label{E:tau_rep}
                     \tau^{(N)} =   (\tau^{(N)} - \tau_{\mathcal{O}}^O) + \sum_{i=1} ^ {\mathcal{O}} \left(\tau_i^O - \tau_{i-1}^S\right) + \sum_{i=1} ^ {\mathcal{O}-1} \left(\tau_i^S - \tau_{i-1}^O\right) \quad \text{on the event }\{\mathcal{O}\geq 1\},
                    \end{equation}
\end{linenomath*}
        and that conditional on the event $\{\mathcal{O}\geq 1\}$, 
         $\mathcal{O}\sim {\textrm Geom}\left(\frac{1}{2-\alpha_N}\right)$. 
        Under $\mathbb{P}_{\rm same}$, the two lineages coalesce before splitting on  the event $\{\mathcal{O}=0\}$.
         % Furthermore, \eqref{E:POzero} holds.
    \end{lemma}
    
    \noindent{\bf Proof. }
        The telescoping sum \eqref{E:tau_rep} holds because on the event $\{\mathcal{O}\geq 1\}$, it holds that 
\begin{eqnarray*}
0=\tau^{S}_1<\tau^O_1<\tau^S_2<\tau^O_2<\cdots < \tau^O_{\mathcal{O}} \leq \tau^{(N)} &\quad \Pdiff-a.s.\quad\text{and}\\
0=\tau^O_0<\tau^{S}_1<\tau^O_1<\tau^S_2<\tau^O_2<\cdots < \tau^O_{\mathcal{O}} \leq \tau^{(N)} &\quad \Psame-a.s.
\end{eqnarray*}

Under $\mathbb{P}_{\textrm same}$, the two lineages coalesce before splitting on the event $\{\mathcal{O}=0\}=\{\tau^{S}_1=\infty\}$ because  $\tau^{(N)}<\infty$ almost surely; see the middle realization in Figure~\ref{fig: unconditional_pairwise_diagram}. Furthermore, 
\begin{linenomath*}
        \begin{equation*}
            \Psame(\mathcal{O} = 0) =\mathbb{E}\left[1 - 2 ^ {-U_N}\right]   = \frac{\alpha_N}{2-\alpha_N},
        \end{equation*}      
\end{linenomath*}
where $U_N$ is the number of selfing events before a splitting event. The first equality holds because there is a $\frac{1}{2}$ chance of coalescing for each selfing event and so, given $U_N$,
there is a $2 ^ {-1} + \ldots + 2 ^ {-U_N} = 1 - 2 ^ {-U_N}$
chance of coalescing during those selfing events. The second equality follows from \eqref{eq:PUk} and \eqref{eq:Fdef}.

That $\Pdiff(\mathcal{O} = 0) = 0$ follows simply from the fact that the two lineages cannot coalesce without belonging to the same individual, whether at coalescence time or not.
        
It remains to prove that, conditional on the event $\{\mathcal{O} \geq 1\}$, it holds that $\mathcal{O} \sim Geom(\frac{1}{2-\alpha_N})$. In each overlap event, it is equally as likely that we coalesce instantly as we do not. However, if we have not coalesced instantly, we have two sample lineages in the same individual. The probability that these two sample lineages coalesce before a splitting event is  $\frac{\alpha_N}{2-\alpha_N}$ as was shown. Therefore the probability that an overlap event is the final overlap is $\frac{1}{2} + \frac{1}{2} \frac{\alpha_N }{2 - \alpha_N} = \frac{1}{2 - \alpha_N}$.
        This gives the claim.   
    \hfill \qed \\

    Now that we have characterized the distribution of $\mathcal{O}$, we characterize the distribution of each term in the decomposition \eqref{E:tau_rep} in the lemmas below. 
    \begin{lemma}\label{lemma: unconditional_pairwise_RV_characterization}
        Under each of $\Pdiff$ and $\Psame(\cdot\,|\,\mathcal{O}\geq 1)$, the $2\mathcal{O}$ terms on the right of \eqref{E:tau_rep} are independent random variables with the following distributions:
        \begin{itemize}
            \item $\{\tau_{i}^S - \tau_{i-1}^O\}_{i=1}^{\mathcal{O}}$ are  geometric random variables with parameter $\frac{2-\alpha_N}{2N}$.
            \item $\{ \tau_i^O - \tau_{i-1}^S\}_{i=1}^{\mathcal{O}-1}$ are  geometric random variables with parameter $2N^{-2}$.
            \item $\tau^{(N)} - \tau_{\mathcal{O}}^O$ is $0$ with probability $1 - \frac{1}{2}\alpha_N$ and, conditioned on not being zero, is a geometric random variable with parameter $\frac{\alpha_N(2-\alpha_N)}{2N}$.
            \end{itemize}  
    \end{lemma}
    \noindent{\bf Proof. }
        We begin with the gap $\tau_i^S - \tau_{i-1}^O$ between  splitting times. Note  there is a $1 - \frac{\alpha_N }{2 - \alpha_N} = \frac{2(1 - \alpha_N)}{2 - \alpha_N}$ chance of splitting before a selfing coalescence. The odds that we split in the next step given the two sample lineages do not coalesce via selfing before a split event is therefore
\begin{linenomath*}
        \begin{equation*}
            (1 - \alpha_N)N ^ {-1}\frac{2 - \alpha_N}{2(1 - \alpha_N)} = \frac{2 - \alpha_N}{2N}.
        \end{equation*}
\end{linenomath*}
        We can see then that $\tau_i^S - \tau_{i-1}^O$ is geometric with this same rate.

        Now we may characterize the gap between the final overlap and coalescence, $\tau^{(N)} - \tau_{\mathcal{O}}^O$.Given that we have reached the final overlap, $\tau^{(N)} - \tau_\mathcal{O}^O$ is $0$ if coalescence occurred at the overlap. There is, unconditioned, a $\frac{1}{2}$ chance of this occurring. There is, unconditioned, a $\frac{1}{2}\frac{\alpha_N}{2 - \alpha_N}$ chance of coalescing due to selfing. This gives the conditioned probability that $\tau^{(N)} - \tau_\mathcal{O}^O$ is $0$ is 
\begin{linenomath*}
        \begin{equation*}
            \frac{\frac{1}{2}}{\frac{1}{2} + \frac{1}{2}\frac{\alpha_N}{2 - \alpha_N}} = 1 - \frac{1}{2}\alpha_N.
        \end{equation*}
\end{linenomath*}
        
        If $\tau^{(N)} - \tau_\mathcal{O}^O \neq 0$ then the final coalescence was due to a selfing event that occurred before splitting. The probability that this coalescence via selfing occurs in a single generation given that this is the last overlap is
\begin{linenomath*}
        \begin{equation*}
            \frac{\alpha_N N ^ {-1} \frac{1}{2}}{\frac{1}{2 - \alpha_N}} =
            \frac{\alpha_N (2 - \alpha_N)}{2N}.
        \end{equation*}
\end{linenomath*}
%        This completes the characterization of $\tau^{(N)} - \tau_{\mathcal{O}}^O$.

        It remains  to characterize $\tau_i^O - \tau_{i-1}^S$. It is equally likely that one coalesces as that one overlaps without coalescing, and as the odds of coalescing in a single time-step when we have two sample lineages in two distinct individuals is $N^{-2}$, in each time-step there is a $2N^{-2}$ chance of an overlap. Therefore, $\tau_{i}^O - \tau_{i-1}^S$ is geometric with parameter $2N^{-2}$. %This completes the lemma.
    \hfill \qed \\

Lemma \ref{lemma: unconditional_pairwise_RV_characterization} immediately implies the following.
        \begin{corollary}\label{cor: unconditional_pairwise_rv_convergence}
            %Take the decomposition given by Lemma~\ref{lemma: unconditional_tau_N_2_decomposition}. 
            Suppose $\alpha_N \rightarrow \alpha \in [0,1]$ as $N\to\infty$. The under each of $\Pdiff$ and $\Psame(\cdot | \mathcal{O} \geq 1)$, the $2\mathcal{O}$ terms on the right of \eqref{E:tau_rep} converge with the $N^{-2}$ rescaling as follows:
            \begin{itemize}
                \item $N^{-2}(\tau^{(N)} - \tau_{\mathcal{O}}^O)$ converges to $0$ in distribution.
                \item $\{N^{-2}(\tau_i^S - \tau_{i-1}^O)\}_{i=1}^\mathcal{O}$ converge in distribution to $0$.
                \item $\{N^{-2}\left(\tau_{i}^O - \tau_{i-1}^S\right)\}_{i=1}^\mathcal{O}$ converge to independent exponential random variables with rate $2$.
            \end{itemize}            
        \end{corollary}

        %In particular, as $\lambda_N Geom(\lambda_N)$ converges to a unit rate exponential random variable when $\lambda_N \rightarrow 0$, and as the splitting times coalesce instantaneously with the time rescaling, we expect that $N^{-2}\tau^{(N)}$ should converge to a geometric sum of exponential random variables, which is an exponential random variable. Indeed this is the case, as one can see in Theorem~\ref{theorem: unconditional_pairwise_coalescence_time_convergence}.

%\medskip

\begin{proof}[Proof of Theorem~\ref{theorem: unconditional_pairwise_coalescence_time_convergence}]
                The distribution of $\tau^{(N)}$ is determined in the decomposition \eqref{E:tau_rep} and Lemma~\ref{lemma: unconditional_pairwise_RV_characterization}. 
    
                Let $\varphi_N, \psi_N, \sigma_N$ denote the characteristic functions of  $\tau_i^O - \tau_{i-1}^S, \tau_{i}^S - \tau_{i}^O$, and $\tau^{(N)} - \tau_{\mathcal{O}}^O$, respectively. 
                By the decomposition \eqref{E:tau_rep} and then Lemma \ref{lemma: unconditional_tau_N_2_decomposition},
\begin{linenomath*}
            \begin{align*}
\Ediff\left[e ^ {itN^{-2} \tau^{(N)}} \right]
=&\, \sigma_N(tN^{-2})\sum_{j = 0} ^ {\infty} \Pdiff(\mathcal{O} = j) \varphi_N(tN^{-2}) ^ j \psi_N(tN^{-2}) ^ j\\
=&\, \sigma_N(t N^{-2}) r_N \frac{\varphi_N(tN^{-2}) \psi_N(tN^{-2})}{1 - (1 - r_N)\varphi_N(tN^{-2}) \psi_N(tN^{-2})},         
            \end{align*}
\end{linenomath*}
where $r_N = \frac{1}{2 - \alpha_N}$. By Corollary~\ref{cor: unconditional_pairwise_rv_convergence} we have $\psi_N(tN^{-2})$ and $\sigma_N(tN^{-2})$ converge to $1$ and $\varphi_N(tN^{-2})$ converges
                $(1 + it 2^{-1})^{-1}$. Since $\alpha_N\to \alpha$, we have $r_N\to r=\frac{1}{2 - \alpha}$ and the last displayed expression converges to
                $r \frac{\varphi(t)}{1 - (1 - r) \varphi(t)} = (1 + it (\frac{2}{2 - \alpha}) ^ {-1}) ^ {-1}$. This is the claim for $\Pdiff$.

%                \begin{equation}\label{eqn: unconditional_limiting_exponential_function}                   \sigma_N(t N^{-2}) r_N \frac{\varphi_N(tN^{-2}) \psi_N(tN^{-2})}{1 - (1 - r_N)\varphi_N(tN^{-2}) \psi_N(tN^{-2})}    \rightarrow r \frac{\varphi(t)}{1 - (1 - r) \varphi(t)} = (1 + it (\frac{2}{2 - \alpha}) ^ {-1}) ^ {-1}.
%                \end{equation} 

The result holds equivalently under $\Psame(\cdot \,|\, \mathcal{O} \geq 1)$ as the distributions given by Corollary~\ref{cor: unconditional_pairwise_rv_convergence} are the same under both laws.
\end{proof}   

Finally, we record an immediate corollary of Theorem~\ref{theorem: unconditional_pairwise_coalescence_time_convergence} and Lemma~\ref{lemma: unconditional_tau_N_2_decomposition}.
    \begin{corollary}\label{corollary: unconditional_pairwise_time_convergence_same}
    If $\alpha_N \rightarrow \alpha \in [0,1]$, then $N^{-2}\tau^{(N)}$ converges in distribution to $0$ 
    under  $\Psame(\cdot\,|\,\mathcal{O} = 0)$ and to an  exponential random variable with rate $\frac{2}{2-\alpha}$
    under  $\Psame(\cdot\,|\,\mathcal{O} \geq 1)$.
    \end{corollary}
    
%    \noindent{\bf Proof. }
%        The two lineages coalesce instantaneously with probability $\frac{\alpha_N}{2-\alpha_N}$, else it splits instantaneously to two lineages in two distinct individuals. The result thus follows from the Markov property, $\alpha_N\rightarrow\alpha$, and Theorem~\ref{theorem: unconditional_pairwise_coalescence_time_convergence}.
%    \hfill \qed \\

\subsection{Proofs for Theorem~\ref{T:MAIN_conditional} and Theorem~\ref{T:MAIN_conditional_same}}\label{section: proofs}

In subsections \ref{section: proofs_critical}-\ref{section: proofs_subcritical} below, we consider convergences for each of the three regimes. The  proofs for Theorems~\ref{T:MAIN_conditional} and \ref{T:MAIN_conditional_same} are summarized in subsection \ref{SS:sumup}.
%for both $\mathbb{P}_{\rm diff}$ and $\mathbb{P}_{\rm same}$. 

\subsubsection{Summary of the main proofs}\label{SS:sumup}

\begin{proof}[Proof of Theorem~\ref{T:MAIN_conditional}]

    The proof of Theorem~\ref{T:MAIN_conditional} in the partial selfing, limited outcrossing, and negligible-outcrossing regimes are contained in Theorem~\ref{theorem: conditional_convergence_pairwise_coalescence_subcritical}, Theorem~\ref{theorem: convergence_to_FC_graph_rw}, and Corollary~\ref{C:diff_conditional_negligible}, respectively.
    
    The proof of Theorem~\ref{theorem: conditional_convergence_pairwise_coalescence_subcritical} is obtained by expanding the $L^2(\Pdiff)$ distance between the conditional survival probability and the survival probability of an exponential random variable. This is performed using a second moment estimate of the conditional survival probability provided by Lemma~\ref{lemma: conditional_splitting_of_char_fctn}. The cross-term and the square of the survival probability of the exponential simplify by Theorem~\ref{theorem: unconditional_pairwise_coalescence_time_convergence}, reducing the problem to convergence in distribution of the minimum of two conditionally independent copies of the pairwise coalescence time. This convergence is established in Theorem~\ref{theorem: conditional_distinct_first_time_convergence}.

    The proof for the limited-outcrossing and negligible-outcrossing regimes are the same, both following from Theorem~\ref{theorem: convergence_to_FC_graph_rw}. The subgraph of the pedigree generated by following all outcrossings and the random walks thereon are shown to converge by Lemma~\ref{L:convergeG^Nandwalks}. This generates a coupling between random walks on the pedigree and random walks on this convergent subgraph. The convergence given by Lemma~\ref{L:convergeG^Nandwalks} is sufficient to imply that the conditional survival probability converges due to its dependence only on the discrete structure of the subgraph and the lengths of its edges. This sufficiency is demonstrated in Theorem~\ref{theorem: convergence_to_FC_graph_rw}.
\end{proof}

\begin{proof}[Proof of Theorem~\ref{T:MAIN_conditional_same}]

    The proof of Theorem~\ref{T:MAIN_conditional_same} in the $\alpha = 0$, $\alpha \in (0,1)$, and $\alpha = 1$ regimes are contained in Corollary~\ref{C:same_conditional_alpha_is_zero}, Theorem~\ref{T: MAIN_conditional_subcritical_same}, and Lemma~\ref{lemma: conditional_psame_critical}, respectively.
    
    The proof %in the $\alpha \in [0,1]$ and $N(1-\alpha_N) \to \infty$ regime of Theorem~\ref{T:MAIN_conditional_same} is contained in 
    of Theorem~\ref{T: MAIN_conditional_subcritical_same} follows the same $L^2(\Psame)$  argument as in the partial-selfing regime of Theorem~\ref{T:MAIN_conditional}, including the same use of a second moment calculation as in Lemma~\ref{lemma: conditional_splitting_of_char_fctn}. The first and third summands in the expansion can be calculated to be the same by showing that the minimum of two conditionally independent realizations of $\tau^{(N)}$ under the ``same'' sampling configuration is $0$ with positive probability, and is otherwise exponential with rate $\frac{4}{2-\alpha}$ using Lemma~\ref{lemma: conditional_identical_individual_splitting_probability}. The cross term can be shown to converge to negative of the above sum using independence of the number of selfing events before the two sample lineages split and the pairwise coalescence time conditional on undergoing a split before coalescence, and then also using the fact that the time of the first splitting event converges to $0$ with the $N^2$ time rescaling by Lemma~\ref{lemma: conditional_identical_individual_splitting_probability}.

    In the case where $N(1-\alpha_N)\to 1$, it follows that the conditional coalescence time converges to $0$ in distribution by Lemma~\ref{lemma: conditional_psame_critical} and Theorem~\ref{theorem: unconditional_pairwise_coalescence_time_convergence}.
\end{proof}

\subsubsection{Proofs for the limited-outcrossing regime}\label{section: proofs_critical}

In this section, we prove our main results for the limited-outcrossing regime, in which $N(1-\alpha_N) \to \lambda \in \R_+$. We start by analyzing the discrete-time ancestral graph $G^N$ and proving Lemma~\ref{L: G_N_convergence_unrigorous}.

    We first give an explicit description of the discrete-time ancestral graph $G^N$. To do so we consider a larger process $\tilde{G}^N$ taking values in finite subsets of
    $\{1,2,\ldots, N\} \times \Z_+ \times \{0,1\}$.       
    An element $(x,y,z)$ in $\tilde{G}^N(k)$ for $x \in \{1,2,\ldots, N\}$, $y \in \Z_+$, $z\in \{0,1\}$ will represent the presence of an ancestor for our sample at individual $x$ in time-step $k$, that this ancestor is given an index $y$, and a binary switch $z$ which is switched 
    from a $0$ to a $1$ (or vice versa) when this ancestral lineage is involved in a splitting or overlap event. The index $y$ is to track which individual is involved in which splitting and overlap event. A new element is added to $\tilde{G}^N$ when there is a splitting event, and another is removed at an overlap event.

    The discrete-time ancestral graph $G^N$ is obtained by projecting $\tilde{G}^N$ to its last two components, i.e. to finite subsets of $\Z_+ \times \{0,1\}$. We call this state space $\mathcal{P}$. The space $\mathcal{P}$ when given the metric of symmetric difference is a Polish space. %Separability follows from countability of the state space. Completeness follows from discreteness of the metric.

    We now describe the transition probabilities of $G^N$. First, %and then prove weak convergence  using standard theory of  Markov processes.
    \begin{linenomath*}
        \[\scalebox{0.85}{$
        G^N(1) = 
        \begin{cases}
            \{((1, 0), (2, 0)\} & \text{neither $\hat{X}_0$ nor $\hat{Y}_0$ is a offspring}\\  %$\hat{X}_0, \hat{Y}_0$ not in a reproduction event nor undergoes selfing overlap
            \{(1, 1), (3, 1), (2, 0)\} & \text{$\hat{X}_0$ splits without overlap with $\hat{Y}_0$} \\
            \{(1, 0) (2, 1), (\pi_{k,2}, 1)\} & \text{$\hat{Y}_0$ splits without overlap with $\hat{X}_0$} \\
            \{(1, 1), (2, 1)\} & \text{$\hat{X}_0$ splits and one parent is $\hat{Y}_0$} \\
            \{(1, 1), (2, 1)\} & \text{$\hat{Y}_0$ splits and one parent is $\hat{X}_0$}\\
            \{(1,0), (2,0)\} & \text{$\hat{X}_0$ self-fertilizes and $\hat{Y}_0$ is not the parent}\\
            \{(1,0), (2,0)\} & \text{$\hat{Y}_0$ self-fertilizes and $\hat{X}_0$ is not the parent}\\
            \{(1, 1)\}, & \text{selfing event between $\hat{X}_0$ and $\hat{Y}_0$}
        \end{cases}.
        $}
        \]
    \end{linenomath*}
    % \begin{linenomath*}
    %         \[\scalebox{0.85}{$
    %         Z_1^N = 
    %         \begin{cases}
    %             \{((1,0),1), ((0,1),2)\} & \text{neither $\hat{X}_0$ nor $\hat{Y}_0$ is a offspring}\\  %$\hat{X}_0, \hat{Y}_0$ not in a reproduction event nor undergoes selfing overlap
    %             \{((\frac{1}{2}, 0),1), ((\frac{1}{2}, 0),3) ((0,1),2)\} & \text{$\hat{X}_0$ splits without overlap with $\hat{Y}_0$} \\
    %             \{((1, 0),1) ((0, \frac{1}{2}),2), ((0, \frac{1}{2}),3)\} & \text{$\hat{Y}_0$ splits without overlap with $\hat{X}_0$} \\
    %             \{((\frac{1}{2}, 0),1), ((\frac{1}{2}, 1),2)\} & \text{$\hat{X}_0$ splits and one parent is $\hat{Y}_0$} \\
    %             \{((1, \frac{1}{2}),1), ((0, \frac{1}{2}),2)\} & \text{$\hat{Y}_0$ splits and one parent is $\hat{X}_0$}\\
    %             \{((1,1),1)\}, & \text{selfing event between $\hat{X}_0$ and $\hat{Y}_0$}
    %         \end{cases}
    %         $}
    %         \]
    % \end{linenomath*}
        % Suppose $Z_k^N = z = \{(p_i,l_i)\}_{i=1}^{m}$ for $p_i\in[0,1]^2$ and $l_i\in\Z_+$. Let $m_k$ denote the smallest positive integer not contained in $\{l_i\}$. Recall $P(k) = \{j_i\}_{1\leq i \leq m}$ denotes the corresponding set of population-level labels such that $l_i$ is the label of the individual to whom the the label $l_i$ (i.e.\ such that $(p_i,l_i,j_i)\in\tilde{Z}_k^N$). By $\hat{z}_{r s...i_j}$ we denote the set $z\setminus [z_{i_l}: 1 \leq l \leq j]$. The one-step dynamics are described in such a way that  $Z_{k+1}^N$ is equal to
        Suppose $G^N(k) = g = \{(y_i, z_i)\}_{i=1}^{m}$ for $y_i \in \Z_+$ and $z_i\in\{0,1\}$. Let $m_k$ denote the smallest positive integer not contained in $\{y_i\}$. Let $P(k) = \{x_i\}_{1\leq i \leq m} \subset \{1,2,\ldots, N\}$ denotes the corresponding set of population-level labels such that $x_i$ is the label of the individual to whom the the label $y_i$ is assigned (i.e.\ such that $(x_i, y_i, z_i)\in\tilde{G}^N(k)$). By $\hat{z}_{I}$ we denote the set $z\setminus [z_{i}: i \in I]$. 
        
        The one-step dynamics are described in such a way that  $G^N(k+1)$ is equal to
\begin{linenomath*}
        \[\scalebox{0.85}{$
        \begin{cases}
            g, 
            &  \text{if }\epsilon_k \notin P(k), \text{ or } \{\epsilon_k, \pi_{k,1}, \pi_{k,2}\} \cap P(k) = \{\epsilon_k\}, \pi_{k,1}=\pi_{k,2}  \\
            
            \{(y_r, z_r+1)\} \cup \hat{g}_{rs}, 
            & \text{if } x_s=\epsilon_k, \pi_{k,1} = \pi_{k,2}=x_r\\
            
            \{(y_r, z_r+1), (y_s, z_s+1)\} \cup \hat{g}_{rst}, 
            & \text{if } x_t = \epsilon_k, P(k) \cap \{\pi_{k,2}, \pi_{k,1}\} = \{x_s, x_t\} \\
            
            \{(y_r,z_r+1), (y_s, z_s+1)\} \cup \hat{g}_{rs}, 
            & \text{if } x_s = \epsilon_k, P(k) \cap \{\pi_{k,1}, \pi_{k,2}\} = \{x_r\} \\
            
            \{(y_r, z_r+1), (m_k, 1)\} \cup \hat{g}_r,
            & \text{if } x_r = \epsilon_k, P(k) \cap \{\pi_{k,1}, \pi_{k,2}\} = \emptyset.
            
        \end{cases}
        $
        }
        \]
\end{linenomath*}
        for any $r<s<t$.  Under our Moran model, $G^N$ is a discrete-time Markov chain with   one-step transition probabilities $\{\mathbb{P}(G^N(k+1) = \cdot | G^N(k) = g) \}$ given as follows:
\begin{linenomath*}
    \begin{equation}\label{ZN transition rates}    
        g\rightarrow
        \begin{cases}
            g, 
            &  \text{w/prob. } \frac{N-m}{N} + \alpha_N\frac{m}{N}\frac{N-m-1}{N-1}\\
            
            \{(y_r, z_r+1)\} \cup \hat{g}_{rs}, 
            & \text{w/prob. } \alpha_N \frac{2}{N} \frac{1}{N-1}\\
            
            \{(y_r,z_r+1), (x_s, z_s+1)\} \cup \hat{g}_{rst}, 
            & \text{w/prob. } (1-\alpha_N)\frac{2}{N(N-1)(N-2)} \\
            
            \{(y_r,z_r+1), (y_s,z_s+1)\} \cup \hat{g}_{rs}, 
            & \text{w/prob. } (1-\alpha_N) \frac{2}{N(N-1)}\frac{N-m-2}{N-2} \\
            
            \{(y_r,z_r+1), (m_k,1)\} \cup \hat{g}_r,
            & \text{w/prob. } (1-\alpha_N) \frac{1}{N}\frac{N-m-1}{N-1}\frac{n-M-2}{n-2}
            
        \end{cases}
    \end{equation}
\end{linenomath*}
        for each fixed $r$ and  $s \neq r, s \neq t, t \neq r$.
        
When $N(1-\alpha_N) \to \lambda$, these transitions are
\begin{linenomath*}\label{LN transition rates}
        \[
g\rightarrow
        \begin{cases}
            g, 
            &  \text{w/prob. } 1-O(N^{-2})\\
            
            \{(y_r, z_r+1)\} \cup \hat{g}_{rs}, 
            & \text{w/prob. } 2N^{-2} + O(N^{-3})\\
            
            \{(y_r,z_r+1), (x_s, z_s+1)\} \cup \hat{g}_{rst}, 
            & \text{w/prob. } o(N^{-3}) \\
            
            \{(y_r,z_r+1), (y_s,z_s+1)\} \cup \hat{g}_{rs}, 
            & \text{w/prob. } o(N^{-3}) \\
            
            \{(y_r,z_r+1), (m_k,1)\} \cup \hat{g}_r,
            & \text{w/prob. } \lambda N^{-2} + o(N^{-2})
            
        \end{cases},
        \]
        where elements of $O(f(N))$ are sequences $a_N$ such that $a_N f(N)^{-1}$ is a bounded sequence, and  elements of $o(f(N))$ are sequences $a_N$ such that $a_N f(N)^{-1}$ converges to zero, as $N\to\infty$.
\end{linenomath*}
        \begin{remark}\label{remark: labels}
            The labels of $G^N(k)$, for any $k\in\Z_+$, are contained in $\{1,\ldots,\sup_{i\leq k} |G^N(i)| \}$.
        \end{remark}
    
%     We let $G_k^N$ be the projection of the set $Z_k^N$ to its labels in $\Z_+$. From the above calculations, $(G_k^N)_{k\in\Z_+}$ is itself a Markov process with one-step transition probabilities 
% \begin{linenomath*}
%     \begin{equation}\label{LN transition rates}
%         l\rightarrow 
%         \begin{cases}
%             \hat{l}_r\cup \{m_k\} & \text{w/prob. }\quad \lambda |G_k^N| N^{-2} + o(N^{-2})\\
%             l & \text{w/prob. }\quad 1-O(N^{-2})\\
%             \hat{l}_r & \text{w/prob. }\quad |G_k^N|(|G_k^N|-1)N^{-2} + O(N^{-2})
%         \end{cases}
%     \end{equation}
% \end{linenomath*}
%     for some $r$. $G_k^N$ takes values in finite subsets of $\Z_+$ $\Lambda = \bigsqcup_{m=1}^{\infty} \{l_i\}_{i=1}^m$ for $l_i \in \Z_+$, which is a Polish space under the metric $d(l,m) = |l \cap m|$. Importantly, this metric engenders the same topology as by the projection map from $\mathcal{P}_{m,2}$. %We denote the time-rescaling of $(G_k^N)_{k\in\Z_+}$    by $G^N := (G_{\lfloor t N^2 \rfloor}^N)_{t\in\R_+}$.

%With the above precise description of the process $G^N$, 
Lemma \ref{L: G_N_convergence_rigorous} is a rigorous restatement of Lemma~\ref{L: G_N_convergence_unrigorous}.     We refer to \cite{ethier2009markov} for the theory of weak convergence in the Skorokhod space.
\begin{lemma}\label{L: G_N_convergence_rigorous}
    Suppose that $N(1-\alpha_N) \to \lambda \in \R_+$. Then the time-rescaled discrete-time ancestral graph $(G^N(\lfloor tN^2 \rfloor))_{t \in \R_+}$ converges weakly in the Skorokhod space $\mathcal{D}\left(\R_+, \mathcal{P}\right)$ to $G_\lambda$.
\end{lemma}

    \noindent{\bf Proof.}\label{lemma:G_N_proof}
        We shall show that if $N(1-\alpha_N) \rightarrow \lambda \in \R_+$, then
        $(G_{\lfloor t N^2 \rfloor}^N)_{t \in \R_+}$ converges in distribution under $\Pdiff$ in $\mathcal{D}(\R_+, \mathcal{P})$ to a continuous-time Markov process $G_\lambda=(G_\lambda(t))_{t\geq 0}$ with initial distribution $G_\lambda(0)=\{(1,0),(2,0)\}$ and transition rates
\begin{linenomath*}
        \begin{equation}\label{L transition rates}
        l \rightarrow 
        \begin{cases}
            l\cup\{(m(l),1)\} & \text{ with rate}\quad \lambda |l|\\
            \hat{l}_r & \text{ with rate}\quad |l|(|l|-1)\\
        \end{cases},
        \end{equation}
\end{linenomath*}
        where $m(l)$ denotes the smallest positive integer not contained in $l$.
        
        Let $T^{(N)}$ be the linear operator on the space $\mathcal{C}_b(\mathcal{P})$ of bounded continuous functions on $\mathcal{P}$ defined by
        $T^{(N)} f(n) =\mathbb{E}\left[f(G^N(1))\,|\,G^N(0)=l\right]$. 
        The generator $\mathcal{L}^N$ of the discrete-time process $G^N$ is given by
\begin{linenomath*}
        \begin{equation}\label{Def:Generator_DTMC}
        \mathcal{L}^Nf(n):=\mathbb{E}\left[f(G^N(1))-f(G^N(0))\,|\,G^N(0)=l\right]=(T^{(N)}-I)f(l)
        \end{equation}
\end{linenomath*}
        which can be explicitly computed using
        \eqref{LN transition rates}.
        Let $\mathcal{L}$ be the infinitesimal generator of the continuous-time process $G_\lambda=(G_\lambda (t))_{t\geq 0}$. That is, by \eqref{L transition rates},
\begin{linenomath*}
        \begin{equation}\label{Def:Generator_CTMC}
        \mathcal{L}f(l)=\lambda |l| [f(l \cup \{m(l)\})-f(l)] + \sum_{r = 1} ^ {|l|} 
        f(\hat{l}_r) r(r-1)].
        \end{equation}
\end{linenomath*}
It follows from \eqref{LN transition rates} that for all $f\in \mathcal{C}_b(\mathcal{P})$ with finite support,
\begin{linenomath*}
        \begin{equation}\label{Conv_generator}
           \sup_{l\in \mathcal{P}}| N^2\mathcal{L}^N f(l) - \mathcal{L}f(l)| \to 0 \quad \text{as }N\to\infty.
        \end{equation}
\end{linenomath*}
    
        Let $\{T(t)\}_{t\in\R_+}$ be the semigroup on $\mathcal{C}_b(\mathcal{P})$ 
        of the continuous-time process $G_\lambda$. By Theorem 6.5 of  \citet[chapter 1]{ethier2009markov} and \eqref{Conv_generator}, it holds that
\begin{linenomath*}
                \begin{equation}\label{eqn: G^N_convergence_conditions}
                    \lim_N \sup_{0 \leq t \leq t_0} \sup_{l \in \mathcal{P}} |\left(T^{(N)}\right)^{\lfloor tN^2\rfloor}f(l) - T(t)f(l)| = 0
                \end{equation}
\end{linenomath*}
        for all $f$ in the domain of $\mathcal{L}$. 
    
        Next, since the splitting rate is linear while the overlap rate is quadratic,  we can show as in \citet[Theorem 3]{Griffiths1991} that the compact containment condition \citep[(7.9), p.\ 129]{ethier2009markov} holds. That is, for any $\epsilon\in(0,1)$ and $T\in (0,\infty)$, there is a constant $K_{\epsilon,T}\in (0,\infty)$ such that
\begin{linenomath*}
        \begin{equation}\label{CCC_N}
        \limsup_{N\to\infty}\PP\left(  \sup_{t\in[0,T]} |G^N(\lfloor t N^2\rfloor )| \geq K_{\epsilon,T} \right) \leq \epsilon.
        \end{equation}
\end{linenomath*}

%This suffices for the compact containment condition of $G^N$ as we recall from remark \ref{remark: labels} that $G_k^N$ is contained in
%\begin{linenomath*}
%        \begin{equation*}
%            \left\{1,2,\ldots,\sup_{0\leq i \leq k} |G_i^N| \right\}
%        \end{equation*}
%\end{linenomath*}
%        for all $k \in \Z_+$, and is clearly equivalent to a compact containment condition for the process $L^N:= (|G_{\lfloor tN^2 \rfloor}^N|)_{t\in \R_+}$ that counts the number of particles at each cross-section.

We now demonstrate how to prove \eqref{CCC_N} by a martingale argument by examining the Markovian process $L := |G_\lambda|$ that is the size of the cross-section of $G_\lambda$ at any time-point. Let $\mathcal{Q}$ denote the generator of $L$. We can see that, for any $f$ in $\mathcal{C}_b(\Z_+)$ that
\begin{linenomath*}
        \begin{equation}
            \mathcal{Q}f(n) = \lambda n (f(n + 1) - f(n)) + n(n-1)(f(n-1) - f(n)).
        \end{equation}
\end{linenomath*}
First, we invoke a general relation between Markov processes and martingales.
        For any $f\in \mathcal{C}_b(\Z_+)$,
\begin{linenomath*}
        \begin{equation} \label{E:mtg_MF}
        M(t) := f(L_t)-f(L_0)-\int_0^t
        \mathcal{Q}f(L_s)\,ds
        \end{equation}
\end{linenomath*}
        is a martingale with quadratic variation $\langle M\rangle_t=\int_0^t\mathcal{Q}(f^2)(L_s)-2f(L_s)\mathcal{Q}f(L_s)\,ds$; see \citet[Lemma 5.1]{kipnis1998scaling} or  \citet[Proposition 4.1.7]{ethier2009markov}. A truncation argument will enable us to take $f$ to be  the identity function
        (i.e. when $f(n)=n$ for all $n\in\Z_+$) to obtain 
        %from \eqref{Def:Generator_CTMC} that
\begin{linenomath*}
        \begin{equation}\label{E:Lt identity}
        L_t=L_0 + \int_0^t
        -L_s^2+(\lambda+1)L_s\,ds + M^{\rm Id}(t) \quad \text{for }t\in\R_+,
        \end{equation}
\end{linenomath*}
        where $M^{\rm Id}$ is a martingale with quadratic variation
        $\langle M^{\rm Id}\rangle_t=\int_0^t L_s^2 +(\lambda-1)L_s\,ds$.
        In the above, we used the fact that, from \eqref{Def:Generator_CTMC}, when $f$ is the identity function, 
\begin{linenomath*}
        \begin{align}
        \mathcal{Q}f(n) =& \lambda n - n(n-1) = -n^2+(\lambda+1)n  \leq \frac{(\lambda+1)^2}{4}, \label{Ineq:LFn}\\
        \mathcal{Q}f^2(n)=&\lambda n[(n+1)^2-n^2] + n(n-1)[(n-1)^2-n^2] \notag\\
        %=\lambda n(2n+1)-n(n-1)(2n-1) 
        &\,=-2n^3 +(3+2\lambda)n^2 +(\lambda-1)n \label{Ineq:LF^2n}, \text{ and }\\
        \mathcal{Q}f^2(n)-2f(n)\mathcal{Q}f(n)=& n^2 +(\lambda-1)n. \notag
        \end{align}
\end{linenomath*}

        Therefore, from \eqref{E:Lt identity} and \eqref{Ineq:LFn}, 
\begin{linenomath*}
        \[
        \sup_{t\in[0,T]}L_{t} \leq L_0 + \frac{(\lambda+1)^2}{4}\,T + \sup_{t\in[0,T]}M^{\mathrm{Id}}(t).
        \]
\end{linenomath*}
        By Doob's maximal inequality, the above formula for $\langle M^{\mathrm{Id}}\rangle_t$, and the fact that $\sup_{t\in[0,T]}\mathbb{E}[L^2_t]<\infty$, which follows from \eqref{Ineq:LF^2n}, there exists a constant $C_T\in(0,\infty)$ such that
\begin{linenomath*}
        \begin{equation*}
        \PP\left( \sup_{t\in[0,T]}M^{\mathrm{Id}}(t)\geq K \right) \leq \frac{C_T}{K^2} \quad\text{for all } K\in (0,\infty).
        \end{equation*}
\end{linenomath*}
From these last two inequalities, for any $\epsilon\in (0,1)$ there exists a constant $K_{\epsilon, T}$ such that
\begin{linenomath*}
        \begin{equation}\label{CCC}
        \PP\Big( \,\sup_{t\in[0,T]}L_{t}\geq K_{\epsilon, T} \Big) \leq \epsilon.
        \end{equation}
\end{linenomath*}
        
        Finally, \eqref{CCC_N} can be obtained by the same argument above, when $L_t = |G_\lambda(t)|$ and $\int_0^t\mathcal{Q}f(L_s)\,ds$ are replaced, respectively, by $L^N(k) :=|G^N(k)|$ and $\sum_{i=0}^{k-1} \mathcal{Q}^N f(L^N_i)$ for $\mathcal{Q}^N$ the generator of $L^N = (L_k^N)_{k \in \Z_+}$. Equipped with \eqref{CCC_N}, it then follows from 
        Corollary 8.9 of  \citet[Chapter 4]{ethier2009markov} 
        and \eqref{eqn: G^N_convergence_conditions} that 
        $G_\lambda$ converges in distribution in $\mathcal{D}(\R_+, \mathcal{P})$ to the process $G_\lambda=(G_\lambda(t))_{t\geq 0}$.
    
        \hfill \qed \\

%\medskip

    Next, we couple random walks on $(G^N(k))_{k\in\Z_+}$ with the lineage dynamics on $\mathcal{G}_N$ in such a way that the first meeting time of the random walks is exactly the same as the time of the first overlap of the sample lineages ($T_\lambda^N=\tau^O_1$ in below).

    \paragraph{Potential ancestors of two sample lineages.}\label{FC-process}

        The projected process $G^N:=(G^N(k))_{k\in\Z_+}$  can be viewed graphically as having two initial nodes $(1,0)$ and $(2,0)$ under the sampling scheme $\Pdiff$. These two nodes correspond to the two lineages in the two individuals from whom we sampled the two lineages at time-step $0$, namely $\hat{X}_0$ and $\hat{Y}_0$. 

        We now consider two conditionally independent simple random walks on $G^N$.    
       % We establish a coupling between our pair of random walks on the pedigree and some new pair of random walks on $Z^N$ using this correspondence.
        Let $x_0^N = (1,0)$ and $y_0^N = (2,0)$. Suppose  $x_j^N$ is defined for all $j \leq k$, for induction. If $\hat{X}_k$ undergoes an outcrossing at time-step $k$, then we define $x_k^N$ to follow the path of the outcrossing at $G^N(k)$ that corresponds to the edge along which $\hat{X}_k^N$ travels. Similarly,  $y_k^N$ can be defined in the same way. This gives two discrete-time random walks on $G^N$ starting at $((1,0),\, (2,0))$  and satisfying
        \begin{itemize}
            \item[(i)] $\{x_k^N,y_k^N\} \subset G^N(k)$ for all $k\in \Z_+$,
            \item[(ii)] at any splitting event, each of 
            $(x_k^N)_{k\in\Z_+}$ and $(y_k^N)_{k\in\Z_+}$ will follow each of the two paths available with equal (i.e. 1/2) probability.
        \end{itemize}
Analogous to \eqref{Def:Tlambda}, we let $T_\lambda^N$ be the first meeting time of the random walks $x^N$ and $y^N$ on $G^N$, i.e.
\begin{linenomath*}
            $$T_\lambda^N:=\inf\{k \in \mathbb{Z}_+: x_k^N = y_k^N \}.$$
\end{linenomath*}

Lemma~\ref{L: G_N_convergence_rigorous} can readily be extended to the joint convergence in Lemma~\ref{L:convergeG^Nandwalks}.
\begin{lemma}\label{L:convergeG^Nandwalks}
Suppose $\lim_{N\to\infty}N(1-\alpha_N)=\lambda \in \R_+$. Then
$(G^N(\lfloor t N^2 \rfloor), x^N(\lfloor t N^2 \rfloor), y^N(\lfloor t N^2 \rfloor))_{t \in \mathbb{R}_+}$ converges in distribution in $\mathcal{D}(\mathcal{P} \times (\Z \times \{0,1\})^2)$ under $\Pdiff$ to a  process $(G_\lambda, x_\lambda, y_\lambda)$, where $G_\lambda=(G_\lambda(t))_{t\in\R_+}$ is an ancestral graph with rate $\lambda$. The joint process $(x_\lambda,y_\lambda) = (x_{\lambda}(t),y_{\lambda}(t))_{t \in \mathbb{R}_+}$ is a continuous-time $G_\lambda \times G_\lambda$-valued process  describe before \eqref{Def:Tlambda}.

\end{lemma}

    \paragraph{Conditional limit for distinct individuals.}\label{section: critical_distinct}

Note that  under $\mathbb{P}_{\rm diff}$, it holds almost surely that $\tau^O_1=\inf\{k \in \mathbb{Z}_+: \hat{X}_{k} = \hat{Y}_{k} \}$ defined in \eqref{Def:tauO}  is the time of the first overlap of our two sample lineages.
We first show in Lemma \ref{lemma: instantaneous_coalescence_critical} that, in the continuous-time limit, coalescence between overlapping particles is instantaneous. Note that this is, in general, not true in the discrete-time model.

        \begin{lemma}\label{lemma: instantaneous_coalescence_critical}
        Suppose $\lim_{N\to\infty} N (1-\alpha_N)=\lambda\in\R_+$.
            Then $\tau^{(N)}-\tau^O_1$ converges to $0$ in distribution as $N\to\infty$ under $\Pdiff(\,\cdot \,| \mathcal{A}_N)$.
        \end{lemma}
        \noindent{\bf Proof. }
            Fix $\epsilon > 0$. Observe that, by the conditional Markov inequality,
\begin{linenomath*}
            \begin{equation*}
                \Pdiff\big(N^{-2} (\tau^{(N)} - \tau^O_1) > \epsilon | \mathcal{A}_N\big)
                \leq \frac{1}{\epsilon} \Pdiff\big(N^{-2} (\tau^{(N)} - \tau^O_1) > \epsilon\big)
            \end{equation*}
\end{linenomath*}
which goes to zero as the probability that there is only a single overlap before coalescence is $\frac{1}{2-\alpha_N}$, which converges to one, by Lemma~\ref{lemma: unconditional_tau_N_2_decomposition} and the time it takes between coalescence and that last overlap has mean going to zero by Lemma \ref{lemma: unconditional_pairwise_RV_characterization}.
        \hfill \qed \\

        We can now prove Theorem~\ref{T:MAIN_conditional} in the limited-outcrossing regime. Recall $T_\lambda$  defined in \eqref{Def:Tlambda}. 
        \begin{theorem}\label{theorem: convergence_to_FC_graph_rw}

            Suppose Suppose $\lim_{N\to\infty} N (1-\alpha_N)=\lambda\in\R_+$. 
            Then, for any fixed $t>0$, the following {\color {blue}convergence in distribution } holds as $N\to\infty$:
\begin{linenomath*}
            \begin{equation*}
             \Pdiff(N^{-2}\tau^{(N)} > t | \mathcal{A}_N)\to   \mathbb{P}_{x_0 \neq y_0} (T_\lambda > t | G_\lambda).
            \end{equation*}
\end{linenomath*}
        \end{theorem}

        \noindent{\bf Proof.}
        \label{P:MAIN_conditional}
            \noindent 
We have constructed the triple $(x^N, y^N, G^N)$ in such a way that $\tau^O_1 = T_\lambda^N$ pointwise. In particular,  $x_k^N$ and $y_k^N$ described before Lemma \ref{L:convergeG^Nandwalks} are the projections of $\hat{X}_k$ and $\hat{Y}_k$ onto $G^N(k)$. 
 %These are paths on $G^N$ which follow each path of each outcrossing with equal likelihood. 
Also, knowing the pedigree tells exactly as much information about $\tau^O_1$ as knowing $G^N$. 
Therefore,
\begin{linenomath*}
            \begin{equation}\label{eqn: theorem: convergence_to_FC_graph_rw}
                \Pdiff(N^{-2}\tau^O_1 > t | \mathcal{A}_N) = \mathbb{P}_{x_0^N \neq y_0^N}(N^{-2} T_\lambda^N > t | G^N).
            \end{equation}
\end{linenomath*}
            It follows from Lemma \ref{L:convergeG^Nandwalks} that the right hand side of \eqref{eqn: theorem: convergence_to_FC_graph_rw} converges as $N\to\infty$ to 
            $\mathbb{P}(T_\lambda > t | G_\lambda)$ in distribution. 

            Indeed, for any graph $g\in \mathcal{D}(\mathbb{R}_+, \mathcal{P})$ with two initial particles, we can define $\tilde{T}$ as in \eqref{Def:Tlambda}. 
The functional $\Psi:\,\mathcal{D}(\mathbb{R}_+, \mathcal{P})\to [0,1]$ defined by $\Psi(\tilde{G}) := \mathbb{P}(\tilde{T} > t \mid g)$ is continuous
because it depends only on the discrete structure of the graph $\tilde{G}$ and its edge lengths up to time $t$. 

By lemma \ref{L:convergeG^Nandwalks} the convergence of $G^N$ to $G_\lambda$ in $\mathcal{D}(\mathbb{R}_+, \mathcal{P})$ 
            and Proposition 5.3 of \citet[Chapter 3]{ethier2009markov},
            the discrete structure of $G^N$ is eventually fixed for large enough $N$ on any compact interval $I$ of which $[0,t]$ is a proper subspace. Furthermore, on $I$ the edge lengths converge uniformly in $N$.
                        
            From the definition of $\Psi$, we have
           $\Psi(G^N)= \mathbb{P}(\tilde{T}_\lambda^N > t \mid G^N)$,
            where $\tilde{T}_\lambda^N$ is the continuous-time analogue of $T_\lambda^N$; precisely,
\begin{linenomath*}
            \[
           \tilde{T}_\lambda^N:= \inf\{t \geq 0: x_{\lfloor tN^2\rfloor}^N = y_{\lfloor tN^2 \rfloor}^N\}.
            \]
\end{linenomath*}
Note that $\tilde{T}_\lambda^N \geq N^{-2}T_\lambda^N$ for all $N$ and their difference $\tilde{T}_\lambda^N - N^{-2}T_\lambda^N$ is pointwise bounded by $\frac{1}{N^2}$. Thus, their limits conditional on the pedigree are the same. Therefore $\mathbb{P}(N^{-2}\tau^O_1>t|\mathcal{A}_N)$ converges as $N\to\infty$ to $\mathbb{P}(T_\lambda > t | G_\lambda)$  in distribution.
                        
            The result follows from Lemma \ref{lemma: instantaneous_coalescence_critical} which says that the limits conditional on the pedigree of $\tau^O_1$ and $\tau^{(N)}$ under $\mathbb{P}_{\text{diff}}(\,\cdot \mid \mathcal{A}_N)$ must be the same.
            \hfill \qed \\

\paragraph{Conditional limit under $\Psame$.}
            
            Selfing is nearly ubiquitous in the limited-outcrossing regime. In particular, any sample of two lineages in the same individual will undergo very many selfing reproduction events before they have the chance to undergo a split. This is reflected in Lemma~\ref{lemma: conditional_psame_critical}.
            
            \begin{lemma}\label{lemma: conditional_psame_critical}
                If $\alpha_N \rightarrow 1$, then $N^{-2}\tau^{(N)}$ converges to $0$ in distribution under $\Psame(\cdot | \mathcal{A}_N)$.
            \end{lemma}
            \noindent{\bf Proof. }
                Fix $\epsilon > 0$. By the conditional Markov inequality, for $t>0$,
\begin{linenomath*}
                \begin{equation}
                    \Psame(\Psame(N^{-2} \tau^{(N)} > t | \mathcal{A}_N) > \epsilon)
                    \leq \frac{1}{\epsilon} \Psame(N^{-2} \tau^{(N)} > t).
                \end{equation}
\end{linenomath*}
The right hand side tends to 0 as $N\to\infty$ because 
$\Psame(\mathcal{O} = 0) = \frac{\alpha_n}{2-\alpha_N}\to 01$ and, 
under the law of $\Psame(\cdot|\mathcal{O}=1)$, $\tau^{(N)}$ converges in distribution to $0$ by Theorem~\ref{theorem: unconditional_pairwise_coalescence_time_convergence}.
            \hfill \qed \\

\subsubsection{Proofs for the negligible-outcrossing regime}\label{section: supcritical_distinct}

    The proofs for the case of negligible outcrossing are straightforward. They follow directly from the results in the limited-outcrossing regime.
    
    \begin{corollary}\label{C:diff_conditional_negligible}
        If $N(1-\alpha_N) \rightarrow 0$, then for any $t\in\R_+$, $\Pdiff(N^{-2} \tau^{(N)} > t | \mathcal{A}_N)$ converges in distribution as $N\to\infty$  to a  Bernoulli random variable with success probability $e^{-2t}$. 
    \end{corollary}

    \noindent{\bf Proof. }
        By Theorem~\ref{theorem: convergence_to_FC_graph_rw}, $\Pdiff(N^{-2} \tau^{(N)} > t | \mathcal{A}_N)\to\mathbb{P}_{x_0 \neq y_0}(T_0 > t | G_0)$. The ancestral graph
        $G_0$  coagulates at rate $2$. In particular, conditional on $G_0$, $T_0$ is known exactly and is distributed unconditionally as an exponential random variable with rate $2$. Therefore, 
\begin{linenomath*}
        \begin{equation*}
            \mathbb{P}(\mathbb{P}_{x_0 \neq y_0}(T_0 > t | G_0) = 1) = e^{-2t} = 1 - \mathbb{P}(\mathbb{P}_{x_0 \neq y_0}(T_0 > t | G_0) = 0).
        \end{equation*}
\end{linenomath*}
\qed \\

\subsubsection{Proofs for the partial-selfing regime}\label{section: proofs_subcritical}

     Throughout this section, we assume that $ N (1-\alpha_N)\to\infty$ and $\alpha_N\to\alpha\in[0,1]$ as $N\to\infty$.

\paragraph{Two conditionally independent \textit{pairs} of genes.}\label{S: TwoPairs}

         In the partial-selfing regime, direct analysis regime of $\tau^{(N)}$ conditional on $\mathcal{A}_N$ is more difficult.  We can proceed by connecting the behavior of $\tau^{(N)}$ to that of two copies of $\tau^{(N)}$ conditionally independent with respect to the pedigree, as in Lemma \ref{L:two_equalities}.

        \begin{lemma}\label{lemma: conditional_splitting_of_char_fctn}
            For any fixed  $t\in\R_+$ and  $N \geq 2$,
\begin{linenomath*}
            \begin{equation*}
                \Ediff \left[\Pdiff(N^{-2}\tau^{(N)} > t | \mathcal{A}_N) ^ 2 \right] = 
                \Pdiff\left(\tau \wedge \tau' > N^2 t\right).
            \end{equation*}
\end{linenomath*}
        \end{lemma}
%        \noindent{\bf Proof. }
 %           By conditional independence we have
%\begin{linenomath*}
%            \begin{equation*}
%                \Ediff\left[\Pdiff(N^{-2}\tau > t | \mathcal{A}_N)\right]\,\Ediff\left[\Pdiff(N^{-2}\tau' > t | \mathcal{A}_N)\right]
%                =  \Ediff\left[\Pdiff(\tau \wedge \tau' > N^2 t | \mathcal{A}_N)\right].
%            \end{equation*}
%\end{linenomath*}
 %      \hfill \qed \\

% The proof follows that of Lemma \ref{L:two_equalities} and is omitted.

        Let $\tau$ and $\tau'$ be conditionally independent pairwise coalescence times with respect to $\mathcal{A}_N$. $\tau$ and $\tau'$ are the pairwise coalescence times of two conditionally independent pairs of random walks on the same pedigree, denoted by $(X,Y)$ and $(X',Y')$. The individual occupied by $X_k, Y_k,X'_k,Y'_k$ are denoted respectively by $\hat{X}_k, \hat{Y}_k, \hat{X}'_k, \hat{Y}'_k$ respectively. We denote by 
\begin{equation}\label{E: sampled Indiv}
\hat{X}_0=\hat{X}'_0 \quad \text{and}
\quad \hat{Y}_0=\hat{Y}'_0
\end{equation}
the sampled individuals.     
        The set of individuals occupied by (at least) one of these random walks $k$ time-steps in the past is denoted by 
\begin{equation}\label{Def:I_k}
\widehat{I}_k := \{\hat{X}_k, \hat{Y}_k, \hat{X}'_k, \hat{Y}'_k\}.    
\end{equation}
In other words, $\widehat{I}_k$ is the set of ancestral individuals for two conditionally independent genealogies  $k$ time-steps in the past.

As in the unconditional case, the key objects of study are overlaps and splits. In this case, though, we need to track 4 sample lineages, not simply two. Luckily, at least in the partial-selfing regime, our analysis will work in much the same way. 

We now define the \textit{overlap times} and the \textit{splitting times} as in \eqref{Def:tauO} and \eqref{Def:tauD}, but for the 4 sample lineages. We let $\zeta_0^O = 0$, and, for $i \geq 1$,
\begin{linenomath*}
        \begin{equation*}
            \zeta_i^O := \inf \{k \in \Z_+: k > \zeta_{i-1}^O, |\widehat{I}_k| < |\widehat{I}_{k-1}|  \}.
        \end{equation*}
\end{linenomath*}
        These can be understood as the overlap times, where (at least) two of our four sample lineages transition to belonging in fewer individuals than before, with each other. For splitting times, we define
\begin{linenomath*}
        \begin{equation*}
           \zeta_0^S:= \inf\{k \in \Z_+: |\widehat{I}_k| = 4\},
        \end{equation*}
\end{linenomath*}
the first time that all four sample lineages are in four distinct individuals. We then iteratively define, for $i \geq 1$,
\begin{linenomath*}
        \begin{equation*}\zeta_i^S:=
            \inf\{k\in\Z_+: |\widehat{I}_k|>|\widehat{I}_{k-1}|\}.
        \end{equation*}
\end{linenomath*}
%The infinmum of an empty set is taken to be infinite, by convention.

        We let $\mathcal{Z}$ be the total number of finite overlap times before $\tau \wedge \tau'$, i.e.,
\begin{linenomath*}
        \begin{equation*}\mathcal{Z}:=
            |\{i \geq 1:\, \zeta_i^O \leq \tau \wedge \tau'\}|.
        \end{equation*}
\end{linenomath*}

        \begin{lemma}\label{lemma: tau_rep_conditional_two_pair}
            Let $\Gamma_N$ be the event
\begin{linenomath*}
            \begin{equation}\label{eqn: two_pair_order}
               \left\{ 0 =\zeta_0^O < \zeta_0^S < \zeta_1^O <\zeta_1^S<\ldots<\zeta_{\mathcal{Z}-1}^S < \zeta_{\mathcal{Z}}^O \leq \tau\wedge\tau' \right\}.
            \end{equation}
\end{linenomath*}
Then $\mathcal{Z}$ under $\Pdiff(\cdot | \,\Gamma_N)$ is geometric with parameter $\frac{1}{6} \frac{\alpha_N}{2-\alpha_N}$. Furthermore, $\lim_{N\to\infty}\Pdiff(\Gamma_N) = 1$.
        \end{lemma}

        The primary use of this lemma  is to get a telescoping sum representation of $\tau \wedge \tau'$:
\begin{linenomath*}
        \begin{equation}\label{eqn: tau_two_pair_rep}
            \tau\wedge\tau' = \tau\wedge \tau' - \zeta_{\mathcal{Z}}^O + \sum_{i=0}^\mathcal{Z} \left(\zeta_i^O - \zeta_{i-1}^S\right) + \sum_{i=1}^{\mathcal{Z}-1} \left(\zeta_i^S - \zeta_i^O\right).
        \end{equation}
\end{linenomath*}

        In order to prove Lemma~\ref{lemma: tau_rep_conditional_two_pair}, we first prove two other lemmas. We begin by calculating the probability that $\zeta_0^S < \zeta_1^O$. We then calculate the probability that any overlap which occurs before $\tau\wedge \tau'$ is followed either by a split or coalescence. We also quantify the times of the splits.

        \begin{lemma}\label{lemma: generic_to_4}
            As $N\to \infty$,
            \begin{equation*}
                \Pdiff(\zeta_0^S < \zeta_1^O) \to 1.
            \end{equation*}
            Furthermore, $N^{-2} \zeta_0^S$ converges to $0$ in distribution under the law $\Pdiff(\cdot \,\mid \zeta_0^S < \zeta_1^O)$.
        \end{lemma}
        \noindent{\bf Proof. }
Recall that the two sampled individuals for the two pairs of conditionally independent gene copies are given by \eqref{E: sampled Indiv}.  The probability that either individual splits in a generation is $(1 - \alpha_N) \frac{2}{N}$.
Let $\lambda_N$  be the rate at which either individual splits without an overlap event for either pair of the gene copies.
Then
\begin{linenomath*}
            \begin{equation*}
               \lambda_N= (1 - \alpha_N) \frac{2}{N} \left[\frac{N - 2}{N - 1} \frac{N - 3}{N - 2} \frac{1}{2}
                +  \frac{2}{N - 1} \frac{1}{4} \right]= (1 - \alpha_N)\frac{N - 2}{N (N - 1)}.
            \end{equation*}
\end{linenomath*}
The term $\frac{N - 2}{N - 1} \frac{N - 3}{N - 2}$ is the probability that the parents of the splitting individual are different from the other individual. The term $\frac{2}{N - 1}$  is the probability that one of the parents of the splitting individual is the other individual.

Let $\mu_N$ be the rate at which there is an overlap event for either pair of the gene copies. Then
\begin{linenomath*}
            \begin{equation*}
            \mu_N=    (1 - \alpha_N) \frac{2}{N} \frac{2}{N - 1} \frac{3}{4} + \alpha_N \frac{2}{N (N - 1)}
                =\frac{3 - \alpha_N}{N (N - 1)}.
            \end{equation*}
\end{linenomath*}

Therefore, the probability that either individual splits without an overlap event is $ \frac{\lambda_N}{\mu_N + \lambda_N}=\frac{(1-\alpha_N)(N - 2)}{3-\alpha_N + (1-\alpha_N)(N - 2)}$. The time it takes for such a desired split to occur, $b_0$, can therefore be observed to be geometric with parameter
\begin{linenomath*}
            \begin{equation*}
                \frac{\lambda_N}{1 - \frac{\mu_N}{\lambda_N}} = \frac{\lambda_N ^ 2}{\lambda_N - \mu_N}.
            \end{equation*}
\end{linenomath*}
    
            Once either individual splits without an overlap event, we have $4$ sample lineages in $3$ individuals. The rate $\lambda_N'$ at which the occupied individual containing two sample lineages splits without an overlap event is
\begin{linenomath*}
            \begin{equation*}
                (1 - \alpha_N)\frac{1}{N} \left(\frac{2}{N-1} \frac{N - 3}{N - 2}
                + \frac{N-3}{N-1} \frac{N-4}{N-2}\right)
                = \frac{(1-\alpha_N)(N-3)(2N-7)}{2N(N-1)(N-2)}.
            \end{equation*}
\end{linenomath*}
            The rate $\mu_N'$ depends more precisely on which individual is chosen as the offspring. If the offspring contains two conditionally independent sample lineages and one parent is occupied, then there is a $\frac{3}{4}$ chance of having an overlap event. Regardless of which occupied offspring is chosen, if both parents are occupied then an overlap event is guaranteed. If an overlap is due to selfing then an overlap occurs if and only if both parent and offspring are occupied. Thus, $\mu_N'$ can be calculated as
\begin{linenomath*}
            \begin{equation*}
                (1-\alpha_N)\left(\frac{1}{N} \frac{4}{N-1}\frac{N-3}{N-2} \frac{3}{4}
                +\frac{2}{N}\frac{4}{N-1}\frac{N-3}{N-2}\frac{1}{2}
                + \frac{3}{N} \frac{2}{N-1} \frac{1}{N-2}\right) + \alpha_N \frac{6}{N(N-1)}.
            \end{equation*}
\end{linenomath*}
            By some simplification this becomes
\begin{linenomath*}
            \begin{equation*}
                \frac{N(7- \alpha_N) - 15 + 3\alpha_N}{N(N-1)(N-2)}.
            \end{equation*}
\end{linenomath*}
            Therefore the probability that we have four lineages in four individuals before any overlap event is
\begin{linenomath*}
            \begin{equation*}
                1 - \frac{\mu_N'}{\lambda_N' + \mu_N'} = 1 - 
                \frac{2(N(7-\alpha_N) - 15 + 3\alpha_N)}{(1-\alpha_N)(2N^2 -15N + 18) + 12N - 12},
            \end{equation*}
\end{linenomath*}
            and the time, $b_1$ it takes for this to occur is geometric with parameter $\frac{(\lambda_N')^2}{\lambda_N' - \mu_N'}$. This gives the claim using the fact that $N(1-\alpha_N)\rightarrow\infty$.
         \hfill \qed \\

        Recall $\widehat{I}_k$ defined in \eqref{Def:I_k}. We consider the case where there are two sample lineages from the same pair in the same individual, and where the other two sample lineages are in distinct individuals. Given this state, what is the probability $p_N$ that we return to four sample lineages in four individuals before coalescence, or that there is a coalescence before another overlap event? This is answered in the limit as $N\to\infty$ by the following lemma.
         \begin{lemma}\label{lemma: splitting_probability_n_n_minus_one}
          As $N\to\infty$, we have $p_N\to 1$ where 
\begin{linenomath*}
             \begin{equation*}
                 p_N := \Pdiff\left(\tau \wedge \tau' < \zeta_{i+1}^O \text{ or } \zeta_{i+1}^S < \tau \wedge \tau'\Big|\;|\widehat{I}_{\zeta_i^O}|=3, \hat{X}_{\zeta_i^O} = \hat{Y}_{\zeta_i^O} \text{ or } \hat{X'}_{\zeta_i^O}=\hat{Y'}_{\zeta_i^O}\right).
             \end{equation*}
\end{linenomath*}
         \end{lemma}

         \noindent{\bf Proof. }
            The probability $\lambda$ that an occupied individual containing two sample lineages splits or else has these two lineages coalesce via splitting in a single time-step, both without an overlap event, is
\begin{linenomath*}
            \begin{align}
               & (1-\alpha_N) \frac{1}{N} \frac{N-3}{N-1} \frac{N-4}{N-2}
               +(1-\alpha_N)\frac{1}{N} \frac{4}{N-1} \frac{N-3}{N-2} \frac{1}{2}
               + \alpha_N \frac{1}{N}\frac{N-3}{N-1} \frac{1}{2}\\
               & = \, 
               \frac{(N-3)((2-\alpha_N)N + 2(1-\alpha_N) - 2)}{2N(N-1)(N-2)}.
            \end{align}
\end{linenomath*}
            The probability the coalescent event occurs before the splitting event can be calculated as the ratio of selfing coalescence probability and $\lambda$.
            The probability $\mu$ that there is an overlap event involving the three occupied individuals is
\begin{linenomath*}
            \begin{equation*}
               \frac{N(7-\alpha_N) -3(5-\alpha_N)}{N(N-1)(N-2)}.
            \end{equation*}
\end{linenomath*}
            Therefore the probability of the two sample lineages in the same individual coalescing or splitting before another overlap event is
\begin{linenomath*}
            \begin{align*}
                \frac{\lambda}{\lambda + \mu}
                =  1 - \frac{\mu}{\lambda + \mu} 
                = 1- \frac{(14-2\alpha_N)N - 15 + 3\alpha_N}{(2-\alpha_N)N ^ 2 + (8-\alpha_N)N -30 + 12 \alpha_N}.
            \end{align*}
\end{linenomath*}
            This gives the claim.
        \hfill \qed \\

%We can now proceed with the proof of Lemma \ref{lemma: tau_rep_conditional_two_pair}.
        
        \noindent{\bf Proof of Lemma \ref{lemma: tau_rep_conditional_two_pair}. }
           % We begin with $\mathcal{Z}$. 
            The behavior of $\mathcal{Z}$ is determined by the structure of the representation \eqref{eqn: two_pair_order}. Indeed, as $\zeta_0^S$ is finite on $\Gamma_N$ we reach a state where all four lineages are in four distinct individuals. From this state, it is only possible for $|\widehat{I}_k|$ to transition from four to three. The presence of splits between each overlap indicates that $|\widehat{I}_k|$ oscillates between 3 and 4 after $\zeta_0^S$ and before $\tau \wedge \tau'$. 
            
            Therefore, the relevant overlaps of pairs occur are those that occur from the state in which all four lineages are in distinct individuals. There are $\binom{4}{2}=6$ equally likely possible pairs of lineages to overlap, but only two of them (those between $x$ and $y$ or $x'$ and $y'$) that could result in coalescence. Therefore there is a $\frac{1}{3}$ chance of either of these two potentially coalescent pairs overlapping at an overlap time. These lineages coalesce instantaneously with a further probability of $\frac{1}{2}$, but they may coalesce before the next split with some positive probability.

            Let $V$ denote the number of selfing events, in the single individual containing two sample lineages, before the next outcrossing for this individual. Here we use $V$ in place of $U$ to reflect the fact that this number of selfing events is not for a single individual containing both sample lineages in the quenched ``same'' configuration but for one of three individuals containing a sample lineage. The conditional probability, given $\Gamma_N$, that we have a selfing event before a split is $\alpha_N$.
            Therefore,
                $\mathbb{P}(V = j) = \alpha_N ^ j (1 - \alpha_N)$ as in \eqref{eq:PUk}.
            When $V = k$, there is a
$2 ^ {-1} + \ldots + 2 ^ {-k} = 1 - 2 ^ {-k}$
            chance of coalescing. Therefore the probability of coalescing due to any given overlap, conditional on $\Gamma_N$, is
\begin{linenomath*}
            \begin{equation*}
                \frac{1}{3}\left(\frac{1}{2} + \frac{1}{2} \mathbb{E}\left[1 - 2 ^ {-V})\right]\right) = \frac{1}{6} \frac{\alpha_N}{2-\alpha_N}.
            \end{equation*}    
\end{linenomath*}
        It suffices now to show that the probability of $\Gamma_N$ converges to one as $N\to\infty$. This is precisely that probability that $\zeta_0^S < \zeta_1^O$ multiplied by the probability that each of the next $\mathcal{Z}-1$ splits each occur before the next overlap, or else coalesce. The probability that $\zeta_0^S <\zeta_1^O$ converges to one by Lemma~\ref{lemma: generic_to_4}. 
        
        The probability that we have the structure in $\Gamma_N$ after $\zeta_0^S$ is exactly
            $\mathbb{E}\left[p_N^\mathcal{Z}\right]$,
        where $p_N$ is the probability given by Lemma~\ref{lemma: splitting_probability_n_n_minus_one}. 
        
        As $p_N$ converges to one as $N\to\infty$ and $\mathcal{Z}$ is a geometric random variable with parameter strictly bounded away from $0$, $\mathbb{E}\left[p_N^{\mathcal{Z}}\right]$ converges to one as $N\to\infty$. This gives the claim.
         \hfill \qed \\

         It remains to quantify the gap between overlaps and splits given by the following lemma, and the gap between the final overlap and coalescence.
         \begin{lemma} \label{lemma: conditional_split_vanishes}
         As $N\to\infty$, the conditional expectations 
         $\left\{\,\mathbb{E}\left[N^{-2} (\zeta_i^S - \zeta_{i-1}^O) | \zeta_i^S < \zeta_i^O\right]\, \right\}_{i\in \Z_{>0}}$ and $\mathbb{E}\left[N^{-2}(\tau^{(N)} - \zeta_\mathcal{Z}^O)) \;\Big|\; |\widehat{I}_{\zeta_{\mathcal{Z}}^O}| = 3\right]$ 
all converge to $0$ in distribution.
        \end{lemma}

        \noindent{\bf Proof. }
            We demonstrate the first limit. The probability of splitting of the multiply-occupied individual in a single step without any overlap event is 
\begin{linenomath*}
            \begin{equation*}
                \frac{(1-\alpha_N)(N-3)(N-4)}{N(N-1)(N-2)}.         
            \end{equation*}
\end{linenomath*}
            Therefore, conditioning on knowing $\zeta_i^S < \zeta_i^O$ implies $(\zeta_i^S - \zeta_{i-1}^O)$ is geometric with this parameter. This gives the claim.
            We now prove the second limit. The difference $\tau^{(N)} - \zeta_{\mathcal{Z}}^O$ is stochastically dominated by the time it takes for the relevant occupied individual to split, which is geometric with parameter $(1 - \alpha_N)N^{-1}$. Therefore,
\begin{linenomath*}
            \begin{equation*}
                \mathbb{E}\left[N^{-2}(\tau^{(N)} - \zeta_{\mathcal{Z}}^O)\right] \leq \frac{N}{1 - \alpha_N} N^{-2} = \frac{1}{N(1-\alpha_N)} \to 0.
            \end{equation*}
\end{linenomath*}
 \qed \\

        We are now prepared to prove the limiting behavior of $\tau \wedge \tau'$ in the partial-selfing regime.

        \begin{theorem} \label{theorem: conditional_distinct_first_time_convergence}
            As $N\rightarrow\infty$, $N^{-2}(\tau \wedge \tau')$ converges in distribution to an exponential random variable with rate $\frac{4}{2-\alpha}$.
        \end{theorem}

        \noindent{\bf Proof. }
            By Lemma~\ref{lemma: tau_rep_conditional_two_pair}, for any $t\in \R$,
\begin{linenomath*}
            \begin{equation}\label{eqn: conditional_characteristic_fctn_1}
                \Ediff\left[e ^ {it N^{-2} (\tau \wedge \tau')}\right]
                =\Ediff\left[e ^ {it N^{-2} (\tau \wedge \tau')} | \Gamma_N\right] + o(1).
            \end{equation}
\end{linenomath*}
            We therefore can decompose $\tau \wedge \tau'$ as in \eqref{eqn: tau_two_pair_rep}. 
            
            Let $\sigma_N$ denote the characteristic function of $\zeta_i^S - \zeta_{i}^O$ conditional on $\Gamma_N$ for $i \geq 1$, $\sigma_N'$ denote the characteristic function of $\zeta_0^S$, and $\varphi_N$ denote the characteristic function of each of the $\zeta_i^O - \zeta_{i-1}^S$, conditional on $\Gamma_N$. The $\{\zeta_i^O - \zeta_{i-1}^S\}_{i = 1}^{\mathcal{Z}}$ all have the same characteristic function as they are all independent and identically distributed. Finally, let $\psi_N$ denote the characteristic function of $\tau\wedge\tau'-\zeta_\mathcal{Z}^O$, again conditional on $\Gamma_N$.
            
            By Lemmas~\ref{lemma: conditional_split_vanishes}~and~\ref{lemma: generic_to_4}, we have that
            $\psi_N(tN^{-2}),\,\sigma_N(tN^{-2}),\,\sigma_N'(tN^{-2})$ all converge to 1 as $N\to\infty$. Note that
            $\varphi_N$ and $\psi_N$ are the characteristic functions of geometric random variables with rates $2\binom{4}{2}N^{-2} + O(N^{-3})$ and $\frac{1}{2}\alpha_N N^{-1}$, respectively. Therefore,
            \begin{equation*}
               \varphi_N(t)= \frac{2\binom{4}{2}N^{-2} + O(N^{-3})}{1 - e^{it} (1- 2\binom{4}{2}N^{-2} - O(N^{-3}))}
\quad\text{and} \quad
                \psi_N(t) =\frac{\frac{1}{2}\alpha_N N^{-1}}{1 - e^{it}(1-\frac{1}{2}\alpha_N N^{-1})}.
            \end{equation*}
            By Lemma~\ref{lemma: generic_to_4}, 
\begin{linenomath*}
            \begin{equation}\label{eqn: conditional_distinct_first_time_convergence}
                \Ediff\left[e ^ {it N^{-2} (\tau \wedge \tau')} | \Gamma_N\right]\,=\,\sigma_N'(tN^{-2})\psi_N(tN^{-2})\sum_{j=1}^{\infty}\Pdiff(\mathcal{Z}=j | \Gamma_N)\varphi_N(tN^{-2})^j \sigma_N^{j-1}.
            \end{equation}
\end{linenomath*}
            Let $\beta_N = \frac{1}{6}\frac{\alpha_N}{2-\alpha_N}$. Since $\Pdiff(\mathcal{Z}=j | \Gamma_N)=\left(\beta_N\right)^j(1-\beta_N)$ by Lemma~\ref{lemma: generic_to_4}, \eqref{eqn: conditional_distinct_first_time_convergence} becomes
\begin{linenomath*}
            \begin{equation}\label{eqn: conditional_distinct_first_time_convergence_2}
                \sigma_N'(tN^{-2})\psi_N(tN^{-2}) \frac{\beta_N \varphi_N(tN^{-2})}{1-(1-\beta_N)\varphi_N(tN^{-2})\psi_N(tN^{-2})}.
            \end{equation}
\end{linenomath*}
            
            Note, $\varphi_N$ is the characteristic function of a geometric random variable with parameter $2\binom{4}{2}N^{-2}+O(N^{-3})$. Therefore $\varphi_N(tN^{-2})$ converges to the characteristic function $\varphi$ of an exponential random variable with rate $2\binom{4}{2}=12$. Furthermore, $\beta_N$ converges to $\beta=\frac{1}{6}\frac{\alpha}{2-\alpha}$.
            Therefore \eqref{eqn: conditional_distinct_first_time_convergence_2} converges as $N\to\infty$ to $ \frac{\beta \varphi}{1-(1-\beta)\varphi}$ which is the characteristic function of an exponential random variable with rate $12\beta = \frac{4}{2-\alpha}$. This is the claim.
        \hfill \qed \\

        We can now arrive at the main proof of this section. %the limit conditional on the pedigree of $N^{-2}\tau^{(N)}$.

        \begin{theorem} \label{theorem: conditional_convergence_pairwise_coalescence_subcritical}
            For any fixed $t > 0$, as $N\rightarrow\infty$, $\Pdiff (N^{-2}\tau^{(N)} > t \mid \mathcal{A}_N)\to e ^{-\frac{2}{2-\alpha}t}$ in $L^2(\Pdiff)$.
        \end{theorem}

        \begin{proof}
            Expanding out the squared $L^2(\Pdiff)$ distance between $\Pdiff (N^{-2}\tau^{(N)} > t \mid \mathcal{A}_N)$ and $e^{-2t}$, we obtain, from Lemma \ref{lemma: conditional_splitting_of_char_fctn},
\begin{equation}\label{E:partial_selfing_theorem_1}
                \Pdiff(N^{-2} (\tau \wedge \tau') > t) - 2 \Pdiff(N^{-2}\tau^{(N)} > t) e ^{-\frac{2}{2-\alpha}} + e^{-\frac{4}{2-\alpha} t}.
\end{equation}
By Theorem~\ref{theorem: unconditional_pairwise_coalescence_time_convergence}, the second summand in \eqref{E:partial_selfing_theorem_1} converges, as $N\to\infty$, to $e^{-\frac{4}{w-\alpha} t}$. It follows that the squared $L^2(\Pdiff)$ distance in equation \eqref{E:partial_selfing_theorem_1} is 
            \begin{equation}
                \Pdiff(N^{-2}(\tau \wedge \tau') > t ) - e ^ {-\frac{4}{2-\alpha}t} + o(1),
            \end{equation}
which converges to 0  by  Lemma \ref{theorem: conditional_distinct_first_time_convergence}. 
        \end{proof}

In particular, in the partial-selfing regime, we have that the classical approximation by averaging over the pedigrees when we begin in distinct individuals is robust to the pedigree.
        
\paragraph{Pairwise conditional convergence starting from a single individual.}\label{A: conditional_same}

We turn now to the case where both sample lineages are taken from the same individual.  %The results we have already established for $\Pdiff$ greatly facilitate the analysis of $\Psame$.
Recall that $\tau^S_1$ is the first  splitting time under $\Psame$ and define $\tau^{D'}_1$ similarly. Then $\{\tau \wedge \tau' > \tau^S_1\}=\{\tau^{S}_1=\tau^{D'}_1<\infty\}$ almost surely under $\Psame$.
        
\begin{lemma}\label{lemma: conditional_identical_individual_splitting_probability}
For all $N\geq 2$,  $\Psame(\tau \wedge \tau' > \tau^S_1) = \frac{4(1-\alpha_N)}{4-\alpha_N}$.
            Furthermore, $\Esame\left[N^{-2} \tau^S_1 \,\Big| \,\tau^S_1 < \infty\right] \to 0$  in distribution as $N\to \infty$.
        \end{lemma}

        \noindent{\bf Proof. }
            Let $U_N$ denote the number of selfing events of the single occupied individual before its first outcrossing, as in the proof of Lemma \ref{lemma: unconditional_tau_N_2_decomposition}. The rate at which selfing events in this individual occur is
                $\alpha_N N ^ {-1}$.
            The rate of splitting is 
               $ (1-\alpha_N)N^{-1}$.
            Therefore the probability of a splitting event before coalescing is
                $1-\alpha_N$.
            Therefore $\Psame(U_N = k) = \alpha_N^{k}(1 - \alpha_N)$. 
            In particular,             $\Esame\left[4 ^ {-U_N}\right] = \frac{4(1-\alpha_N)}{4-\alpha_N}$.
            
            The probability that neither pair coalesces during a selfing event is $\frac{1}{4}$. Therefore 
\begin{linenomath*}
            \begin{equation*}
                \Psame(\tau \wedge \tau' > \tau^S_1| U=k) 
                = 4 ^ {-k}.
            \end{equation*}
\end{linenomath*}
            That $N^{-2} \tau^S_1$ converges to $0$ follows from it being geometric with parameter $(1-\alpha_N)N^{-1}$.
        \hfill \qed \\

        \begin{corollary}\label{corollary: two_pair_conditional_same}
        Suppose $\alpha_N\rightarrow\alpha\in[0,1]$. 
         As $N\to\infty$,
\begin{linenomath*}
            \begin{equation*}
               \Esame\left[e ^ {it N^{-2} (\tau\wedge\tau')}\right]\to \frac{3\alpha}{4-\alpha} + \left(1-\frac{3\alpha}{4-\alpha}\right)\left(1 + it (\frac{4}{2 - \alpha}) ^ {-1} \right) ^ {-1}.
            \end{equation*}
\end{linenomath*}
        \end{corollary}
%Here a geometric random variable with parameter $0$ is taken simply to be infinite.
        
        \noindent{\bf Proof. }
            We can decompose $\Esame\left[e ^ {it N^{-2} (\tau\wedge\tau')}\right]$ into where coalescence occurs before the first splitting time $\tau^S_1$. The probability that it occurs before the first split is $\frac{3\alpha_N}{4-\alpha_N}$ and the time for it to occur, given that we have coalesce converges in distribution to $0$ with the time rescaling by Lemma~\ref{lemma: conditional_identical_individual_splitting_probability}. Therefore, conditional on $\{\tau \wedge \tau' < \tau^S_1\}$, with the $\Psame$ sampling scheme $\tau \wedge \tau'$ converges in distribution to $0$.

            Similarly we can condition on $\{\tau \wedge \tau' > \tau^S_1\}$, where we transition into the initial conditions for the two pairs under $\Pdiff$. Therefore
\begin{linenomath*}
            \begin{equation}\label{eqn: two_pair_conditional_same}
                \Psame \Big(N^{-2} (\tau \wedge \tau') > t \big| \tau \wedge \tau' > \tau^S_1\Big)
                =\Pdiff(N^{-2}(\tau\wedge \tau'+\tau^S_1)>t).
            \end{equation}
\end{linenomath*}
            Again by Lemma~\ref{lemma: conditional_identical_individual_splitting_probability} we know $N^{-2} \tau^S_1$, conditional on being finite, converges to $0$ in distribution as $N\to\infty$ so \eqref{eqn: two_pair_conditional_same} converges, by Theorem~\ref{theorem: conditional_convergence_pairwise_coalescence_subcritical} to $e ^ {- \frac{4}{2-\alpha}t}$ as $N\to\infty$. In particular, then, $\tau \wedge \tau'$ under the law $\Psame(\cdot | \tau \wedge \tau' > \tau^S_1)$ converges to an exponential random variable with rate $\frac{4}{2-\alpha}$, who has a characteristic function $\big(1+it(\frac{4}{2-\alpha})^{-1}\big)^{-1}$.

            Combining the two conditional limits gives the result.
        \hfill \qed \\

        \begin{theorem}\label{T: MAIN_conditional_subcritical_same}
        Suppose $\alpha_N \rightarrow \alpha \in [0,1]$ and $N(1-\alpha_N)\rightarrow \infty$. Then for  any fixed $t > 0$,  
\begin{linenomath*}
            \begin{equation*}
             \Psame(N^{-2}\tau^{(N)} > t | \mathcal{A}_N)\to    2^{-U}e^{-\frac{2}{2-\alpha}t}
            \end{equation*}
\end{linenomath*}
            in distribution as $N\to\infty$, where $\mathbb{P}(U=k)=\alpha^k(1-\alpha)$ for $k$ in $\mathbb{Z}_+$.
        \end{theorem}

        \noindent{\bf Proof. }
            Let $U_N$ be as in the proof of \ref{lemma: conditional_identical_individual_splitting_probability} and $\tau^S_1$ be the time-step of this split. Clearly, $U_N$ converges in distribution to a random variable $U$ as in the statement of the theorem.

            The key observation is that $\tau^{(N)}$ coalesces ``instantaneously'' with probability $1 - 2^{-U_N}$, else it reaches a state where the two lineages are in distinct individuals and that these two distinct individuals are directly readable from the pedigree, reducing to the case where we started in distinct individuals.
            Equivalent to the statement is showing that
                $\Psame(N^{-2}\tau^{(N)} > t | \mathcal{A}_N)$
            converges to $2^{-U}e ^ {-\frac{2}{2-\alpha} t}$ in $L^2(\Psame)$.
            
            Let $B_N$ be the event $\{N^{-2} \tau^S_1 < t\}$. $B_N$ converges to one in probability by Lemma~\ref{lemma: conditional_identical_individual_splitting_probability}. In particular,
 \begin{linenomath*}
           \begin{equation*}
                \Psame(N^{-2}\tau^{(N)} > t | \mathcal{A}_N) = \Psame(N^{-2}\tau^{(N)} > t | \mathcal{A}_N, B_N) + o(1).
            \end{equation*}
\end{linenomath*}
            Furthermore,
\begin{linenomath*}
            \begin{equation*}
                \Psame(N^{-2}\tau^{(N)} > t | \mathcal{A}_N, B_N)
                = 2 ^ {-U_N} \Psame(N^{-2}\tau^{(N)} > t | \mathcal{A}_N, B_N, \tau^{(N)}>\tau^S_1) .
            \end{equation*}
\end{linenomath*}
            Therefore 
\begin{linenomath*}
            \begin{align*}
                &\Esame\left[ (\Psame(N^{-2}\tau^{(N)} > t | \mathcal{A}_N) - 2^{-U_N} e ^ {-\frac{2}{2-\alpha} t})^2 \right]
               \\ =& \Esame\left[ (\Psame(N^{-2}\tau^{(N)} > t | \mathcal{A}_N) - 2^{-U_N} e ^ {-\frac{2}{2-\alpha} t})^2 \middle| B_N \right] + o(1)
            \end{align*}
\end{linenomath*}
            is equal to
\begin{linenomath*}
            \begin{equation}\label{eqn: pairwise_conditional_same}
            \begin{aligned}
                &\Psame(N^{-2}(\tau \wedge \tau') > t) + \Esame\left[4^{-U_N}\right]e^{-\frac{4}{2-\alpha} t} \\
                &- 2 \Esame\left[ 4^{-U_N} \Psame(N^{-2}\tau^{(N)} > t \middle| \tau^{(N)} > \tau^S_1, B_N, \mathcal{A}_N) \right]e^{-\frac{2}{2-\alpha} t} + o(1) . 
            \end{aligned}
            \end{equation}
\end{linenomath*}
            See that the sum of the first and second summands converges to
                $\frac{8(1-\alpha)}{2-\alpha}e^{-\frac{4}{2-\alpha}t}$
            by Corollary \ref{corollary: two_pair_conditional_same} and that $\Esame\left[4^{-U_N}\right]=\frac{4(1-\alpha_N)}{2-\alpha}$. It suffices therefore to show that the third summand converges to the negation of this term.
            
            The key observation is that $4^{-U_N}$ is independent to the pairwise coalescence time conditional on the pedigree given that the coalescence time occurs after the split. In particular,
\begin{linenomath*}
            \begin{equation*}
                \Esame\left[4^{-U_N}\Psame(N^{-2}\tau^{(N)} > t \big| N^{-2}\tau^{(N)}>\tau^S_1,B_N)\;\Big|\;\mathcal{A}_N \right]
            \end{equation*}
\end{linenomath*}
            decomposes into a product
\begin{linenomath*}
            \begin{equation*}
                \Esame\left[4^{-U_N}\right]\Psame(N^{-2}\tau^{(N)} > t | \tau^{(N)}>\tau^S_1,B_N).
            \end{equation*}
\end{linenomath*}
            $\Esame\left[4^{-U_N}\right]$ may be calculated directly as $\frac{4(1-\alpha_N)}{4-\alpha_N}$ and $\Psame(N^{-2}\tau^{(N)} > t | N^{-2}\tau^{(N)}>N^{-2}\tau^S_1,B_N)$, which converges to $e^{-\frac{2}{2-\alpha} t}$ by Corollary~\ref{corollary: unconditional_pairwise_time_convergence_same} and the fact that conditioning on $N^{-2}\tau^{(N)} > \tau^S_1$ is equivalent to conditioning on there being a split before coalescence. Therefore the third summand in \eqref{eqn: pairwise_conditional_same} converges to
\begin{linenomath*}
            \begin{equation*}
                \frac{-8(1-\alpha)}{2-\alpha} e ^ {-\frac{4}{2-\alpha}t},
            \end{equation*}
\end{linenomath*}
            which we established was sufficient for the claim.
        \hfill \qed \\

    The special case of $\alpha = 0$ is of interest and is an immediate consequence of Theorem \ref{T: MAIN_conditional_subcritical_same}.
    \begin{corollary}\label{C:same_conditional_alpha_is_zero}
        Suppose $\alpha_N \to 0$. Then for any fixed $t > 0$, as $N \to \infty$,
        \begin{equation*}
            \Psame\left( N^{-2} \tau^{(N)} > t \,\mid \mathcal{A}_N \right)
            \to e^{-t}
        \end{equation*}
        in probability with respect to $\Psame$.
    \end{corollary}

\subsection{Proof for Proposition \ref{prop:survival_probability_estimates}}\label{S:further_characterization}

Observe that the two pairs of particles $(x_\lambda,y_\lambda)$ and $(x_\lambda',y_\lambda')$ absorbed at $T_\lambda \wedge T_\lambda'$ is a continuous-time Markov chain, with state space $\{s_0,s_1,s_2,s_\Delta\}$, where we  collapsed $s_{\Delta,2}$ and $s_{\Delta,2}$ into a single state called $s_\Delta$.
This absorbed process, called $\widetilde{S} = (\widetilde{S}_t)_{t\in\R_+}$,
has 4 states and has transition rate matrix 
\begin{linenomath*}
\begin{equation}\label{E:Rlambda}
        R_\lambda
            =   \begin{pmatrix} A_{\lambda} & {\bf v}^{\,\rm T} \\   {\bf 0} & 0 
            \end{pmatrix},
        \end{equation}
\end{linenomath*}
by Lemma \ref{L:S_5states},
where $A_{\lambda}$ is the matrix \eqref{E:Alambda} and ${\bf v}= \begin{pmatrix}
                4 & 4 & 2
            \end{pmatrix}$.
Note that
\begin{equation}\label{E:relation}
T_\lambda \wedge T_\lambda'=\inf\{t\in\R_+:\,S_t\in \{s_{\Delta,1},s_{\Delta,2}\}\} =\inf\{t\in\R_+:\,\widetilde{S}_t=s_{\Delta}\}
\end{equation}
is the first time
the process $\widetilde{S}$  reaches $s_{\Delta}$. The joint distribution of $T_\lambda$ and $T_\lambda'$ can be computed as in \cite[Lemma A1]{DFBW24} but we omit this here since it is not needed.

We now compute the second moment of the conditional survival probability $\mathbb{P}(T_\lambda > t | G_\lambda)$.

\begin{lemma}\label{L:survival_probability}
For any  $\lambda \in (0,\infty)$ and  $t \in (0,\infty)$, $\mathbb{E}[\mathbb{P}(T_\lambda > t | G_\lambda)^2]$ is equal to the sum of the entries of the last row of the 3 by 3 matrix $e^{t A_{\lambda}}$, where
\begin{linenomath*}
\begin{equation}\label{E:Alambda}
        A_\lambda=
            \begin{pmatrix}
                -12 & 8 & 0 \\
                \frac{\lambda}{2} & -6-\frac{\lambda}{2} & 2  \\
                0 & \lambda & -2-\lambda  
            \end{pmatrix}.
        \end{equation}
\end{linenomath*}
\end{lemma}

\noindent{\bf Proof of Lemma~\ref{L:survival_probability}.}
We first show that
\begin{linenomath*}
            \begin{equation}\label{E:EPT2}
            \mathbb{E}[\mathbb{P}(T_\lambda > t | G_\lambda)^2]=
                \begin{pmatrix}
                    0 & 0 & 1 & 0
                \end{pmatrix}
                e ^ {t R_\lambda}
                \begin{pmatrix}
                    1 & 1 & 1 & 0
                \end{pmatrix}^T.
            \end{equation}
\end{linenomath*}
Recall the absorbed chain $\widetilde{S}_\lambda = (\widetilde{S}_\lambda(t))_{t\in\R_+}$   mentioned earlier in this section. Since $x_\lambda(0) = x_\lambda'(0)\neq y_\lambda(0) = y_\lambda'(0)$, we have the initial condition $\widetilde{S}_\lambda(0) = s_2$. 
%with infinitesimal generator $R_\lambda$ tracks the configuration of the four particles given by conditionally independent pairs of random walks $(x_\lambda,y_\lambda)$ and $(x_\lambda',y_\lambda')$ on $G_\lambda$. 
Then for any $t\in \R_+$, by \eqref{E:Var_Cov},
\begin{align*}
\mathbb{E}\left[\mathbb{P}(T_\lambda > t | G_\lambda) ^ 2\right]
=&\, \mathbb{P}(T_\lambda \wedge T'_\lambda > t ) \\
=&\, \mathbb{P}_{s_2}(\widetilde{S}_\lambda(t)\in \{s_0,s_1,s_2\} ) \\
=&\,(0,0,1,0)\,e^{tR_\lambda} \,(1,1,1,0)^T,
\end{align*}
where the second last equality follows from \eqref{E:relation}. Hence \eqref{E:EPT2} is established. 
The right hand side of \eqref{E:EPT2} is the sum of the entries of the last row of the 3 by 3 matrix $e^{t A_{\lambda}}$. Hence
Lemma~\ref{L:survival_probability} follows. 
\hfill \qed \\

\medskip

\noindent{\bf Proof of Proposition \ref{prop:survival_probability_estimates}.}
We will employ equation \eqref{E:EPT2} in the proof of Lemma~\ref{L:survival_probability}.
We begin by establishing the rate of convergence as $\lambda\to 0$. We write
$R_\lambda = C + D \lambda$, where
 \begin{linenomath*}
           \begin{equation}
                C:=\begin{pmatrix}
                    -12 & 8 & 0 & 4\\
                    0 & -6 & 2 & 4\\
                    0 & 0 & -2 & 2\\
                    0 & 0 & 0 & 0
                \end{pmatrix}
                \quad\text{and}\quad
                D:=
                \begin{pmatrix}
                    0 & 0 & 0 & 0\\
                    \frac{1}{2} & -\frac{1}{2} & 0 & 0 \\
                    0 & 1 & -1 & 0 \\
                    0 & 0 & 0 & 0
                \end{pmatrix}.
            \end{equation}
\end{linenomath*}

Then, using for instance \citet[eqn. 10.40]{Higham2008}, 
\begin{equation}\label{E:block_matrix_trick}
                 \exp\left(t\begin{pmatrix}
                    C & D \\
                    0 & C
                \end{pmatrix}\right)
                = \begin{pmatrix}
                    e^{tC} & \int_0^t e ^ {C(t-s)} D e^{As} ds\\
                    0 & e ^{tC}
                \end{pmatrix}.
            \end{equation}

By applying the derivative in $\lambda$ at $\lambda = 0$ to the Suzuki-Trotter identity, e.g.\ \citet[eqn. 10.9]{Higham2008}, for the matrix exponential we simultaneously obtain that
            \begin{equation}\label{E:integral_equality_R_lambda}
                \frac{d e^{tR_\lambda}}{d\lambda}\Bigg|_{\lambda = 0} = 
                \frac{d}{d\lambda}\lim_{N\to\infty} \left[e^{\frac{C}{N}} e^{\frac{D\lambda}{N}}\right]^N \Bigg|_{\lambda = 0}
                = \lim_{N\to\infty} \frac{1}{N}\sum_{i = 1}^{N} e ^ {\left(1-\frac{i}{N}\right)C}  D e ^ {\frac{i}{N} C}
                =\int_0^t e ^ {C(t-s)} D e^{As} ds.
            \end{equation}
            Combining \eqref{E:integral_equality_R_lambda} and \eqref{E:block_matrix_trick} gives
            \begin{equation*}
                \frac{d e^{tR_\lambda}}{d\lambda}\Bigg|_{\lambda = 0} =
                \begin{pmatrix}
                    I,0
                \end{pmatrix}\
                \exp{\left(t\begin{pmatrix}
                    C & D \\
                    0 & C
                \end{pmatrix}\right)}
                \begin{pmatrix}
                    0 \\
                    I
                \end{pmatrix}.
            \end{equation*}
            By a \citet[Mathematica]{Mathematica} calculation, detailed in \cite{manuscript_git},
            \begin{equation}
                (0,0,1,0)
                \begin{pmatrix}
                    I & 0
                \end{pmatrix}
                \exp\left(t\begin{pmatrix}
                    C & D \\
                    0 & C
                \end{pmatrix}\right)
                \begin{pmatrix}
                    0 \\
                    I
                \end{pmatrix}
                (1,1,1,0)^T  
                % =\frac{1}{4}\left(e^{-2t}(1-4t)-e^{-6t} \right)
                =\frac{e^{-2t}}{8}\left(1-4t-e^{-4t}\right)
            \end{equation}
            The result for $\lambda\to 0$ thus follows by Taylor's theorem for matrix functions and \eqref{E:EPT2}.

            The proof for $\lambda\to\infty$ is done via a \citet{Mathematica} calculation in \cite{manuscript_git}. The proof of Proposition \ref{prop:survival_probability_estimates} is complete.
        \hfill \qed

\end{document}